\documentclass[oneside]{amsart}
\usepackage[a4paper,body={17.2cm,24.5cm},centering]{geometry}
\usepackage{microtype} 
\usepackage{tikz}
\usetikzlibrary{arrows,calc,matrix,topaths,positioning,scopes,shapes,decorations,decorations.markings,decorations.pathreplacing,decorations.pathmorphing} 
\usepackage{forest}
\usepackage{stmaryrd}
\usepackage{textcomp}
\usepackage{amscd,amsfonts}
\usepackage{amsthm,amsmath}
\usepackage{amstext,amssymb}
\usepackage{xcolor}
\usepackage[hidelinks]{hyperref}
\reversemarginpar
\usepackage{enumerate}
\usepackage{enumitem}
\newlist{enumD}{enumerate}{1}
\setlist[enumD]{label=(D\arabic*)}
\newlist{enumPL}{enumerate}{1}
\setlist[enumPL]{label=(PL\arabic*)}
\newlist{enumR}{enumerate}{1}
\setlist[enumR]{label=(R\arabic*)}
\newlist{enumEQ}{enumerate}{1}
\setlist[enumEQ]{label=(EQ\arabic*)}
\newlist{enumM}{enumerate}{1}
\setlist[enumM]{label=(M\arabic*)}
\newlist{enumRe}{enumerate}{1}
\setlist[enumRe]{label=(R\arabic*)}

\usepackage{xargs}                      
\usepackage{xcolor}  
\usepackage[colorinlistoftodos,textsize=tiny]{todonotes}
\usepackage{cancel}

\newcommandx{\unsure}[2][1=]{\todo[linecolor=red,backgroundcolor=red!25,bordercolor=red,#1]{#2}}
\newcommandx{\change}[2][1=]{\todo[linecolor=blue,backgroundcolor=blue!25,bordercolor=blue,#1]{#2}}
\newcommandx{\info}[2][1=]{\todo[linecolor=cyan,backgroundcolor=cyan!25,bordercolor=cyan,#1]{#2}}
\newcommandx{\improvement}[2][1=]{\todo[caption={Short note},linecolor=violet,backgroundcolor=violet!25,bordercolor=violet,size=\tiny,#1]{#2}}
\newcommandx{\thiswillnotshow}[2][1=]{\todo[disable,#1]{#2}}
\numberwithin{equation}{section}
\numberwithin{figure}{section}
\theoremstyle{plain}
\newtheorem{theorem}{Theorem}[section]
\theoremstyle{definition}
\newtheorem{definition}[theorem]{Definition}
\theoremstyle{plain}
\newtheorem{proposition}[theorem]{Proposition}
\theoremstyle{plain}
\newtheorem{corollary}[theorem]{Corollary}
\theoremstyle{plain}
\newtheorem{lemma}[theorem]{Lemma}
\theoremstyle{remark}
\newtheorem{remark}[theorem]{Remark}
\theoremstyle{definition}
\newtheorem{example}[theorem]{Example}
\theoremstyle{definition}

\newcommand{\ie}{\emph{i.e.} }
\newcommand{\ot}{\otimes}
\newcommand{\cat}[1]{\mathbf{#1}}
\newcommand{\ca}[1]{\mathcal{#1}}
\def\co{\colon\thinspace}

\newcommand{\id}{\operatorname{id}}

\def\h{{\mathfrak h}}

\def\Cross{\operatorname{Cross}}
\def\Der{\operatorname{Der}}
\def\End{\operatorname{End}}
\def\Aut{\operatorname{Aut}}

\newcommand{\e}{\mathfrak e}
\newcommand{\g}{\mathfrak g}
\newcommand{\der}{\operatorname{Der}}

\tikzstyle{my circle}=[draw, fill, circle, minimum size=3pt, inner sep=0pt]	
\tikzstyle{sq}=[draw, fill, rectangle, minimum size=3pt, inner sep=0pt]		
\tikzstyle{whitesq}=[draw, fill=white, rectangle, minimum size=4pt, inner sep=0pt]		
\tikzstyle{lie}=[draw,thin,fill, circle, minimum size=4pt, inner sep=0pt]	
\tikzstyle{prelie}=[draw,thin, fill=white, circle, minimum size=4pt, inner sep=0pt]	

\newcommand{\pl}{\;\triangleright}

\newcommand{\PSB}{\ca{PRT}}
\newcommand{\SB}{\ca{RT}}

\title{What is the Magnus expansion?}

\author{Kurusch Ebrahimi-Fard \and Igor Mencattini \and Alexandre Quesney}
\address[KEF]{Norwegian University of Science and Technology (NTNU), NO-7491 Trondheim, Norway. Centre for Advanced Study (CAS),  Drammensveien 78, 0271 Oslo, Norway.}
\email{kurusch.ebrahimi-fard@ntnu.no}
\urladdr{https://folk.ntnu.no/kurusche/}
\address[IM]{Universidade de S\~ao Paulo (USP), S\~ao Carlos, SP, Brazil.}
\email{igorre@icmc.usp.br}
\urladdr{https://sites.icmc.usp.br/igorre/}
\address[AQ]{Universidad Polit\'ecnica de Madrid (UPM), Campus de Montegancedo, Avenida de Montepr\'{i}ncipe s.n., 28660 Boadilla del Monte, Madrid, Spain.}
\email{alexandre.quesney@upm.es}


\date{\today}

\begin{document}

\begin{abstract}
 The Magnus expansion, a specific infinite Lie series, was introduced by Wilhelm Magnus in 1954. His work tackled a fundamental question in applied mathematics: computing the logarithm of the solution to a first-order homogeneous linear differential equation involving a linear operator. Since its discovery it has evolved into a pivotal tool used across diverse disciplines, including physics, chemistry, and engineering. Over the past 25 years, however, the Magnus expansion has undergone significant mathematical developments which revealed an intricate interplay between algebra, combinatorics, and geometry. 
 By emphasizing a modern perspective based on pre- and post-Lie algebras, we explore the Magnus expansion within the Guin--Oudom framework. This specifically includes Magnus' original expansion, which is notable for its pre-Lie algebra structure first observed by Agrachev and Gamkrelidze. This approach enables us to interpret the inverse of the Magnus expansion as a crossed morphism of formal groups. Consequently, the question posed in the title is addressed by showing that the Magnus expansion can be identified as a relative Rota--Baxter operator.
\end{abstract}


\maketitle


	
\section{Introduction} 
\label{sec:intro}

In his groundbreaking 1954 work \cite{Magnus}, Wilhelm Magnus addressed a fundamental problem in applied mathematics: expressing the logarithm of the solution to a first-order homogeneous linear initial value problem defined in terms of a linear operator. He provided a solution in the form of an invertible infinite Lie series known as the Magnus expansion. This expansion offers significant insights together with powerful computational tools for approximating solutions of differential equations across various scientific and engineering disciplines. We refer the reader to Blanes et al.~\cite{Blanesetal2008} for an authoritative account on applications of the Magnus expansion.

Studying the Magnus expansion quickly reveals an intricate interplay of geometric, algebraic and combinatorial structures, most notably pre- and post-Lie algebras. Consequently, the Magnus expansion has seen significant mathematical developments over the past 25 years, including applications outside the classical realm of ordinary differential equations. Given these advancements, it is natural to wonder about a more precise understanding of the mathematical nature of this specific Lie series. This article explores the Magnus expansion in the context of the Guin--Oudom framework, which permits to characterize the inverse of the Magnus expansion as a crossed morphism of formal groups. In return, this result allows us to answer the question posed in the title: the Magnus expansion is identified as a relative Rota--Baxter operator. It is worth remarking that Magnus' original expansion provides an obvious example thanks to its natural pre-Lie algebra structure, first observed by Agrachev and Gamkrilidze in the context of studying group-theoretic properties of flows, defined by nonstationary differential equations on manifolds \cite{AG}.

\medskip


\subsection{The origin of the Magnus expansion}

Starting from a matrix-valued first-order homogeneous linear initial value problem
\begin{equation*}
		\dot{Y}(t) = A(t) Y(t), \qquad Y(0)=Y_0,
\end{equation*}
the computation of the logarithm 
$$
	\log (Y(t))=:\Omega(A)(t)
$$ 
permits to express the solution as a proper matrix exponential
\begin{equation*}
		Y(t) = \exp(\Omega(A)(t)) Y_0.
\end{equation*} 

In the 1954 landmark paper \cite{Magnus}\footnote{See also Magnus' earlier research report on ``Algebraic aspects in the theory of systems of linear differential equations" \cite{Magnus-report}.}, Wilhelm Magnus described the function $\Omega(A) (t)$ as the unique solution of the following differential equation 
\begin{equation}
\label{eq:MagnusODE}
	\dot{\Omega}(A)(t) 
	= A(t) + \sum_{n\geq 1} \frac{B_n}{n!} \text{ad}_{\Omega(A)}^n (A)(t)
	 = \frac{\text{ad}_{\Omega(A)}}{e^{\text{ad}_{\Omega(A)}} -1} (A)(t),
\end{equation}
with initial value $\Omega(A)(0)=0$. 
Here $\text{ad}^n$ is the $n$-th iteration of the Lie adjoint map $\text{ad}_X(Y) := [X,Y]$. The coefficients $B_n$ are the Bernoulli numbers, 
$$
	B_0=1,\; 
	B_1=-\frac 12,\;
	B_2=\frac 16,\;
	B_3=0,\;
	B_4=-\frac {1}{30},\;
	B_5=0,\;
	B_6=-\frac{1}{42},\;
	B_7=0,\ \ldots
$$
which are defined in terms of the exponential generating function
$$
	\frac{z}{e^z-1}=\sum_{m\ge 0}B_m\frac{z^m}{m!}.
$$

By introducing a formal parameter $h$, one obtains an infinite series known as Magnus expansion 
\begin{equation}
\label{MagnusExpansion}
	\Omega(hA)(t) = \sum_{k\geq 1} h^k\Omega_{k}(A)(t), 
\end{equation}
where 
\begin{equation*}
	\Omega_{1}(A)(t) = \int_{0}^{t}A(s)ds 
\end{equation*}	
and higher order terms can be expressed as a combination of iterated integrals and Lie brackets by expanding and integrating
\begin{align}
 \label{plm-rec2}
	\dot\Omega_k(A)
	&=\sum_{m=1}^{k-1}\frac{B_m}{m!}\sum_{r_1+\cdots+r_m=k-1 \atop r_i >0 }
	\Big[\Omega_{r_1}(A),\big[ \Omega_{r_2}(A), \cdots [\Omega_{r_m}(A), A]\big]\cdots\Big],\quad k>1.
\end{align}
In particular: 
\begin{equation*}
	\Omega_{2}(A)(t) = - \frac{1}{2}\int_{0}^{t} \left(  \big[\int_{0}^{s_1}A(s_2)ds_2,A(s_1)\big] \right)ds_1		
\end{equation*}
and
\begin{multline*}
	\Omega_{3}(A)(t) = \frac{1}{4}\int_{0}^{t} \left(  
			\Big[\int_{0}^{s_1}\big[\int_{0}^{s_2} A(s_3)ds_3,A(s_2)\big]ds_2,A(s_1)\Big] \right)ds_1  \\+
	 \frac{1}{12}\int_{0}^{t} \left( \Big[\int_{0}^{s_1} A(s_2)ds_2, \big[\int_{0}^{s_1}A(s_3)ds_3,A(s_1)\big]\Big] \right)ds_1 . 
\end{multline*}

Writing $\Omega(hA) = h\Omega_1(A) + \tilde{\Omega}(hA)$, one may interpret the higher-order terms, $\tilde{\Omega}(hA)$, as an infinite series of correction terms needed in light of non-commutativity of the matrix product. These contributions have to be included so that the exponential $\exp\big(h\int_0^t A(s)ds + \tilde{\Omega}(hA)(t)\big)$ matches the solution $Y(t)$ expressed as a --formal-- series of iterated integrals called time-ordered exponential
\begin{equation}
\label{eq:timeord}
	Y_h(t)=\mathcal{T}\!\!\exp(h\int_0^t A(s)ds)Y_0
	:= \Big(\text{id} + \sum_{n>0} h^n \idotsint\limits_{0 \le s_1 \le \cdots \le s_n \le t} A(s_1) \cdots A(s_{n})\ ds_1 \cdots ds_{n} \Big)Y_0.
\end{equation}

\begin{remark}
\label{rmk:someremarks1}
\begin{enumerate}

\item \label{item:rightMagnus} 
Considering instead the initial value problem $\dot{Y}(t) = Y(t) A(t)$, $Y(0)=Y_0,$ one can show that the computation of $\log(Y(t))$ results in the slightly different Magnus expansion $\tilde{\Omega}(A)$ defined in terms of modified Bernoulli numbers $\tilde{B}_k$ as coefficients, where $\tilde{B}_1= -B_1=1/2$ and $\tilde{B}_k=B_k$ for all $k \neq 1$. 

\smallskip

\item \label{item:convergence} 
Regarding error estimates and convergence results for the expansion $\Omega(A)$, we refer the reader to reference \cite{Blanesetal2008} which provides a rather comprehensive survey on the Magnus expansion and its wide range of applications across diverse disciplines such as physics, chemistry, and engineering. More recently, the Magnus expansion was revisited in \cite{Beauchardetal2023,Beauchardetal2024} in the context of nonlinear systems in control theory. 

\smallskip

\item \label{item:ChenStrichartz} 
Following B.~Mielnik's and J.~Plebański's classical reference \cite{MP1970}, the Magnus expansion can be expressed in a rather different form widely known as Chen--Strichartz formula \cite{Strichartz1987} 
\begin{equation}
\label{eq:ChenStrich}  
	\Omega(A)(t) 
	= \sum_{n \ge 1} \sum_{\sigma \in S_n} 
	\frac{(-1)^{d_\sigma}}{n^2{n-1 \choose d_\sigma}}\;  
	\idotsint\limits_{0 \le s_1 \le \cdots \le s_n \le t}
	\Big[A(s_{\sigma_1}),\big[A(s_{\sigma_2}), \cdots [A(s_{\sigma_{n-1}}),A(s_{\sigma_{n}})]\cdots \big] \Big]
	ds_1 \cdots ds_{n}.
\end{equation}
Here $S_n$ denotes the symmetric group of order $n$ and, for any permutation $\sigma \in S_n$, $d_\sigma:=|D(\sigma)|$, where
$$
	D(\sigma):=\{i \ :\ \sigma (i)>\sigma(i+1),\ 1\le i <n\} \subseteq \{1,\ldots,n-1\}
$$
is the descent set of the permutation $\sigma$. We refer the reader to the more recent article \cite{Arnaletal2018} where the authors employ the Malvenuto--Reutenauer Hopf algebra of permutations to give a general expression for the terms in the Magnus expansion as iterated integrals of independent right-nested commutators.
\end{enumerate}
\end{remark}


\subsection{From the Magnus expansion to pre-Lie algebras}

Returning to the differential equation \eqref{eq:MagnusODE}, a closer look at the terms in the expansion \eqref{MagnusExpansion} leads to the following key observation, first made by A.~Agrachev and R.~Gamkrelidze in \cite{AG}, which identifies a specific bilinear operation on matrix-valued functions 
\begin{equation}
\label{eq:chrono1}
	(X \pl Y) (t) := \big[\int_{0}^{t}X(s)ds,Y(t)\big] = \int_{0}^{t}X(s)ds\, Y(t) - Y(t) \int_{0}^{t}X(s)ds
\end{equation}
as a crucial structural component in the Magnus expansion. For instance, the second order term of $\Omega(A)(t)$ writes 
\begin{align*}
	\Omega_{2}(A)(t) 
	&= - \frac{1}{2}\int_{0}^{t} \left(  \big[\int_{0}^{s_1}A(s_2)ds_2,A(s_1)\big] \right)ds_1\\
	&= - \frac{1}{2}\int_{0}^{t} (A \pl A)(s_1)ds_1.
\end{align*} 
Using the Lie adjoint map, we see that 
\begin{align}
\label{adjoint}
\begin{aligned}
	\Omega_{2}(A)(t) 
	&= - \frac{1}{2}\int_{0}^{t}  \text{ad}_{\int_{0}^{s_1}\!\! A(s_2)ds_2}(A(s_1)) ds_1\\
	&= - \frac{1}{2}\int_{0}^{t} 	L_{A\pl}(A)(s_1) ds_1,
\end{aligned}	
\end{align} 
where the left-multiplication map 
$$
	L_{X\! \pl }(Y)(t):= (X\! \pl Y)(t)
$$ 
and $L^n_{X\! \pl }(Y)(t):= L^{n-1}_{X\! \pl } (X\! \pl Y)(t)$, for $n>1$. This gives rise to both a deeper understanding as well as a more transparent presentation of the Magnus expansion. Regarding the latter claim, we note that, by following example \eqref{adjoint}, Lie brackets and integrals can be absorbed by rewriting $\dot{\Omega}(A)$ in terms of \eqref{eq:chrono1} and the identity
\begin{equation}\label{eq:idexpl}
	\text{ad}_{\Omega(A)}(A)(t) 
	 = [\Omega(A)(t),A(t)] 
	 =\text{ad}_{\int_{0}^{t}\! \dot{\Omega}(A)(s) ds}(A(t))
	 = L_{\dot{\Omega}(A)\! \pl }(A)(t), 
\end{equation}
where we used $\Omega(A)(0)=0$. This results in
\begin{align}
\label{eq:preLieMag1a}
	\dot{\Omega}(hA)(t) 
	= hA(t) + \sum_{n\geq 0} \frac{B_n}{n!} \text{ad}_{\Omega(A)}^n (hA)(t)
	= hA(t) + \sum_{n\geq 1} \frac{B_n}{n!} L_{\dot{\Omega}(A)\pl }^n (hA)(t).
\end{align}
Under this rewriting, Formula \eqref{plm-rec2} becomes
\begin{align}
 \label{plm-rec3}
	\dot{\Omega}_k(A)(t)
	&=\sum_{m=1}^{k-1}\frac{B_m}{m!}\sum_{r_1+\cdots+r_m=k-1 \atop r_i >0 }
	\Big(\dot{\Omega}_{r_1}(A) \pl \big( \dot{\Omega}_{r_2}(A) \pl \cdots (\dot{\Omega}_{r_m}(A) \pl A)\big)\Big)(t), \quad k>1.
\end{align}
Expanding -- up to order three -- gives an arguably simpler series 
\begin{equation}
 \label{preLieMagnusseries}
	\dot{\Omega}(A)(t)
	= hA(t) - \frac{h^2}{2} (A \pl A)(t) + \frac{h^3}{4} ((A \pl A)\pl A)(t) + \frac{h^3}{12} (A \pl (A\pl A))(t) + \cdots .
\end{equation}

Agrachev and Gamkrelidze called the binary operation \eqref{eq:chrono1} chronological product and showed that it satisfies the so-called chronological identity \cite{AG}
\begin{equation}
\label{eq:chronoB1}
	(X \pl Y) \pl Z - X \pl (Y \pl Z) = (Y \pl X) \pl Z - Y \pl (X \pl Z),
\end{equation}
which defines the notion of chronological algebra; continuous matrix-valued functions therefore provide an example of chronological algebra. In fact, any associative algebra satisfies  \eqref{eq:chronoB1}. Moreover, identity \eqref{eq:chronoB1} is sufficient for the commutator 
\begin{equation}
\label{eq:chronoLie}
	[X,Y]:= X \pl Y - Y \pl X
\end{equation}
to define a Lie bracket. Hence, chronological algebras are Lie admissible. Chronological calculus is an integral part of control theory \cite{AgraSary2004}.

Interestingly enough, identity \eqref{eq:chronoB1} was already well-known in the mathematics literature under the name (left) pre-Lie or quasi-symmetric relation. It can be traced back to seminal works by E.~Vinberg \cite{Vinberg1963} and M.~Gerstenhaber \cite{Gerstenhaber1963}, where it has played important roles in geometry, algebra, and combinatorics. For more details and background, we refer the reader to D.~Manchon's short review \cite{Manchon2009} as well as the recent textbook by P.~Cartier and F.~Patras \cite{CarPat2021}. In Section \ref{sec:Magnus}, see also Subsection \ref{ss:Magpost}, these introductory comments will be extended to the case of a post-Lie algebra. From this very general setting we will be able to recover the Magnus expansion for a general pre-Lie algebra, see Corollary \ref{cor:inversecumulants} and Remark \ref{rem:classvsmod}, of which the chronological case discussed above is a particular instance.


\subsection{The Magnus expansion and other structures}

We return to the Magnus expansion. In the form presented in \eqref{eq:preLieMag1a}, the series is known as pre-Lie Magnus expansion in the context of general pre-Lie algebras \cite{ChaPat2013,KM,EFPatras2013}. Returning to Agrachev's and Gamkrelidze's concrete pre-Lie product \eqref{eq:chrono1} and defining the binary left- and right-integration operations 
\begin{equation}
\label{eq:halfshuffle}
	(X \succ Y) (t) = \int_{0}^{t}X(s)ds\, Y(t) 
	\qquad
	(X \prec Y) (t) = X(t)\, \int_{0}^{t}Y(s)ds ,
\end{equation}
we observe that   
\begin{equation}
\label{eq:chronoB}
	(X \pl Y) (t) = (X \succ Y) (t) - (Y \prec X) (t).
\end{equation}
The fact that \eqref{eq:chrono1} respectively \eqref{eq:chronoB} define pre-Lie products can be deduced from the observation that the operations \eqref{eq:halfshuffle} satisfy the so-called noncommutative shuffle (or dendriform) identities
\begin{align}
\label{shufflealgebra}
\begin{aligned}
	(X \prec Y) \prec Z	
	&=	X \prec (Y \ast Z), 						\\
 	(X \succ Y) \prec Z	
	&=	X \succ (Y \prec Z), 						\\
	(X \ast Y) \succ Z	
	&=	X \succ (Y \succ Z),		
\end{aligned}				
\end{align}
where the noncommutative shuffle (or dendriform) product is
\begin{equation}
\label{eq:shuffle}
 	X * Y := X \succ Y + X \prec Y. 
 \end{equation}
The shuffle identities \eqref{shufflealgebra} imply that the product \eqref{eq:shuffle} is associative. Observe that the Lie bracket defined in terms of the associative product \eqref{eq:shuffle} coincides with that defined in terms of the pre-Lie product \eqref{eq:chronoLie}. Hence, continuous matrix-valued functions provide an example of a noncommutative shuffle algebra. For instance, note that in explicit form, the first identity of \eqref{shufflealgebra} just encodes  integration by parts
\begin{eqnarray*}
	((X \prec Y) \prec Z)(t) 
	&=&  X(t)\int_{0}^{t}Y(s)ds \int_{0}^{t}Z(r)dr\\
	&=&  X(t)\int_{0}^{t} \Big(\int_{0}^{s}Y(r)drZ(s) + Y(s)\int_{0}^{s}Z(r)dr  \Big)ds 
	=  X(t)\int_{0}^{t} (Y * Z) (s)ds.	
\end{eqnarray*}
The shuffle and pre-Lie perspective on the Magnus expansion becomes more interesting when noting that there are examples of shuffle and pre-Lie algebras which do not derive from an integral operator. See for instance \cite{CEFPP2022,CP2023}, where relations between different types of cumulants in noncommutative probability theory are described using the pre-Lie Magnus expansion \eqref{eq:preLieMag1a}. See \cite{KM,EFPatras2021} for more background and details. 

\smallskip


\subsection{The Magnus expansion and post-Lie algebras}
\label{ss:Magpost}

Recently, it has become evident that the concept of a pre-Lie algebra is a specific instance of the broader concept of a post-Lie algebra, which was first introduced in \cite{BV2007}. See also \cite{MKW}, where the closely related notion of D-algebra appeared independently. In a post-Lie algebra, the Lie brackets of two Lie algebras, $\mathfrak{h}$ and ${\mathfrak{g}}$, defined over the same vector space $V$, are related in terms of a third bilinear product $\pl : V \times V \to V$ such that 
$$      
	[x,y]_{{\mathfrak{g}}} = x \pl  y - y \pl  x + [x,y]_{\mathfrak{h}}.
$$
The relations between the product $\pl$ and the Lie bracket on $\mathfrak{h}$ giving rise to the last identity are
\begin{align*}
	x \pl [y,z]_{\mathfrak{h}}
	&= [x \pl y,z]_{\mathfrak{h}} + [y,x \pl z]_{\mathfrak{h}}\\
	[x,y]_{\mathfrak{h}} \pl z
	&= (x \pl y) \pl z - x \pl (y \pl z) - (y \pl x) \pl z + y \pl (x \pl z).
\end{align*}
They define the notion of post-Lie algebra. Note that if the Lie algebra $\mathfrak{h}$ is trivial, then the product $\pl$ reduces to a pre-Lie product. For more details, we refer the reader to Section \ref{sec: crossed to post-Lie}. See also \cite{CEFMK} for a brief review. 
 
In line with Agrachev's and Gamkrelidze's identification of the pre-Lie nature of the classical Magnus expansion, the post-Lie Magnus expansion, see Formula \eqref{eq:postLieMagIntro} below, as well as \eqref{eq:postLieMagorder3} and \eqref{eq:preLieMag in hbar}, permits to express solutions of matrix-valued finite difference initial value problems as a proper matrix exponential \cite{KM14}. The underlying summation by parts calculus is described in terms of noncommutative quasi-shuffle (or tridendriform) algebra. See \cite{CEFMK,CEFO,EFPatras2021b} for more details, including links to stochastic as well as geometric integration theory. However, a detailed description of the Magnus expansion within the framework of post-Lie algebras is more complex and was provided in \cite{AEFM22,EFLMMK,EFLM,MQS,MQ-Crossed}. As an initial comparison, we present the post-Lie Magnus expansion up to the third order in the parameter~$h$.
\begin{equation}
\label{eq:postLieMagIntro}
	\chi(hx) 
	= hx - \frac{h^2}{2} x \pl x + \frac{h^3}{4} (x \pl x)\pl x + \frac{h^3}{12} x \pl (x \pl x) + \frac{h^3}{12} [x \pl x,x]_{\mathfrak{h}} + \cdots ,	
\end{equation}
indicating how a trivial Lie algebra $\mathfrak{h}$  ($[-,-]_{\mathfrak{h}}=0$) leads to the pre-Lie Magnus expansion \eqref{preLieMagnusseries}.

\smallskip


\subsection{The Magnus expansion and integration of post-Lie algebras}

In this survey we revisit the work \cite{MQ-Crossed} by showing that the viewpoint of crossed morphisms permits to understand the (pre-) post-Lie Magnus expansion as a crucial tool in the integration theory of (pre-) post-Lie algebras. Proposition \ref{lem: varPhi crossed morph} identifies the post-Lie Magnus expansion as the inverse of a crossed morphism between two formal groups naturally associated to a post-Lie algebra. Consequently, the post-Lie Magnus expansion can be identified as a relative Rota--Baxter operator, see Remark \ref{rmk:relRB}. 
 
\smallskip\smallskip

\noindent We close this introduction by highlighting the main points of each section.


\subsubsection{Plan of the present work}

We have divided the material of this note into six main sections followed by an appendix. Each section is devoted to the presentation of the (definitions and properties of each of the) main characters of our narrative, i.e.~crossed morphisms (Sections \ref{sec:crossedmorphism} and \ref{sec: crossed to post-Lie}), post-Lie algebras (Sections \ref{sec: crossed to post-Lie}, \ref{sec:exainf} and \ref{sec:unipo}) and the Magnus expansion (Section \ref{sec:Magnus}). We also explain how they interact with each other. In the appendix, we collect several of the relevant algebraic and differential geometric concepts used in Sections \ref{sec:crossedmorphism}-\ref{sec:Magnus}.

In more detail, in Section \ref{sec:crossedmorphism} we first introduce the notion of crossed morphism both in the case of Lie groups, see Subsection \ref{sec: subsec cross morph gp}, and of Lie algebras, see Subsection \ref{ssec:liealgebras} Then we recall how they relate to each other, see Subsubsection \ref{sssec:group2lie}. Moreover, in Lemma \ref{lem:crosMC} and Proposition \ref{pro:37} we enclose the relation of a crossed morphism between the Lie algebras $\g$ and $\h$ and the Maurer--Cartan elements of the differential graded Lie algebras $\Omega^\bullet(G,\h)^G$ and $C^\bullet(\g,\h)$, whose definition and main properties are collected in the appendix.

In Section \ref{sec: crossed to post-Lie}, we analyze post-Lie structures from the viewpoint of crossed morphisms. More precisely, after borrowing the notion of post-Lie group from \cite{Bai-Guo-Sheng-Tang-post-groups}, we prove that one can form a category isomorphic to the category of the so-called split special extensions of Lie groups, see Proposition \ref{prop: iso SE PL groups}. Then we introduce the notion of post-Lie algebra $(\h,\pl)$ and we show that the corresponding category is isomorphic to the category of special split extensions of Lie algebras, see Proposition \ref{pro:equiv}, which is the infinitesimal analogue of the category of special split extensions of Lie groups. Finally in Proposition \ref{prop:postgroups}, we explain how post-Lie algebras and post-Lie groups relate to each other.
In Section \ref{sec:exainf} we collect a few examples of post-Lie algebras having a differential geometric origin. 

In Section \ref{sec:unipo} we make an in-depth analysis of the universal enveloping algebra of a post-Lie algebra, aiming to introduce the so-called Grossman--Larson product, see Proposition \ref{pro:69}. Along the way, we introduce the notion of D-bialgebra and we discuss its role in the theory of post-Lie algebra, see Proposition \ref{pro:68}. 

In the sixth and final section, we explain how one can integrate a post-Lie algebra $(\h,\pl)$. More precisely, first we present the notion of completion of a post-Lie algebra with respect to a descending filtration. Then we show that a suitable twist of the map $\Theta$ introduced in Section \ref{sec:unipo}, Formula \eqref{eq:starmon}, defines an invertible crossed morphism between two formal groups $\mathcal G$ and $\mathcal H$, stemming from the completions of the Lie algebras $\g$ and $\h$ naturally associated to the post-Lie algebra $(\h,\pl)$.
The inverse of this crossed morphism is a remarkable Lie series called post-Lie Magnus expansion which, in our context, is the main ingredient to define the post-Lie group integrating the original post-Lie algebra, see Proposition \ref{pro:79} and Theorem \ref{thm:81}. Our approach to post-Lie algebras and their integration discloses an interesting (and novel) relation between the Grossman--Larson product and the post-Lie Magnus expansion summarized in Lemma \ref{lem:glth} and in Theorem \ref{thm:glth}. 

\vspace{0.2cm}


\subsection*{Acknowledgements}

We are grateful to the anonymous
referees, whose suggestions helped us to substantially improve the content and the presentation of our
manuscript.
KEF is supported by the Research Council of Norway through project 302831 {\it{Computational Dynamics and Stochastics
on Manifolds}} (CODYSMA). He would also like to thank the Centre for Advanced Study (CAS) in Oslo for its warm hospitality and financial support during the research project {\it{Signatures for Images}} (SFI).

\vspace{0.2cm}


\subsection{Notations and conventions} 
\label{sssec:notations}

\label{ssec:notation}
Hereafter, $M$ will denote a real smooth manifold, $C^\infty(M)$ is the ring of smooth functions on $M$ and $\mathfrak X(M)$ the (left) $C^\infty(M)$-module of vector fields on $M$. As usual, every $X\in\mathfrak X(M)$ will be identified with a derivation on $C^\infty(M)$, defined by $X(f)(m)=X_mf$, for all $m\in M$, and $\mathfrak X(M)$ will be endowed with the Lie bracket defined by $[X_1,X_2](f)=X_1(X_2(f))-X_2(X_1(f))$, for all $f\in C^\infty(M)$. The term ``smooth'' will be used synonymously with ``differentiable'' of class $C^\infty$ and if $\phi:M_1\rightarrow M_2$ is a smooth application, its differential, at the point $m$, will be denoted by $\phi_{\ast,m}$, where the latter is a linear map between the tangent spaces $T_mM_1$ and $T_{f(m)}M_2$, defined by the formula $\phi_{\ast, m}(v)(f)=v(f\circ\phi)$, for all $v\in T_mM_1$ and $f\in C^\infty(M_2)$.

\vspace{0.2cm}

	
\section{Crossed morphisms}
\label{sec:crossedmorphism}

In this section we will introduce the notion of crossed morphism both in the case of Lie groups and of Lie algebras. We will prove how crossed morphisms of Lie groups relate to equivariant one-forms and to Maurer--Cartan (MC)-elements. In particular, it will be proven that to every crossed morphism between Lie groups corresponds a crossed morphism between the corresponding Lie algebras. For more information about crossed morphisms and related topics, we refer the reader to the monograph \cite{Hilgert-Neeb}. We refer to Appendix \ref{sec:diffformV} for a collection of background material.


\subsection{Groups}
\label{sec: subsec cross morph gp}


\subsubsection{Crossed morphisms}

Assume $G$ and $H$ to be Lie groups and let $\phi: G\rightarrow \text{Aut}(H)$ be a Lie group morphism. 
Hereafter we will denote by $\phi_g$ the image of $g \in G$ via $\phi$, i.e.~$\phi_g:=\phi(g)$ for all $g\in G$. 
In particular, for all $g\in G$, $\phi_g:H \rightarrow H$ is an isomorphism of Lie groups.

\begin{lemma}\label{lem:semdirG}
$G \times H$ with the composition law 
\begin{equation}
\label{eq:prophi}
	(g_1,h_1)(g_2,h_2)=(g_1g_2,h_1\phi_{g_1}(h_2)),
\end{equation}
is a Lie group, called the semi-direct product of $G$ and $H$ and usually denoted by $G\ltimes_\phi H$.
\end{lemma}

\begin{proof} 
The identity for the product is $(e_G,e_H)$ and, for every $(g,h) \in G\times H$, $(g,h)^{-1}:=(g^{-1},\phi_{g^{-1}}(h^{-1}))$ is an inverse. 
Finally, associativity and smoothness of \eqref{eq:prophi} can be easily checked.
\end{proof}

The Lie group $G\ltimes_\phi H$ is an extension of $G$ by $H$. In other words, one has the following short exact sequence of Lie groups 
\begin{equation}
\label{eq: cano. extension groups}
	1\to H\xrightarrow{i} G\ltimes_\phi H \xrightarrow{\pi} G \to 1,
\end{equation}
where $i$ is the canonical inclusion and $\pi$ is the projection onto the first factor.  Now, consider a smooth map  $s\co G\to  G\times H$   such that  $\pi\circ s = \id$. Clearly, $s(g) = (g,f(g))$ for some smooth map $f\co G \to H$. 

\begin{definition}[Crossed morphism of groups]
A smooth map $f:G\rightarrow H$ is a crossed morphism relative to $\phi$ if for all elements $g_1,g_2\in G$
\begin{equation}
\label{eq:crossmor}
	f(g_1g_2)= f(g_1)\phi_{g_1}(f(g_2)).
\end{equation}
\end{definition}

The next result follows.

\begin{lemma}
$f:G\rightarrow H$ is a crossed morphism relative to $\phi$ if and only if $(\mathrm{id}_G,f):G\rightarrow G\ltimes_\phi H$ is a group morphism.
\end{lemma}

\begin{proof} 
It follows from a direct computation that
\[
	(g,f(g))(h,f(h))\stackrel{\eqref{eq:prophi}}{=}\big(gh,f(g)\phi_g(f(h))\big)
\]
which equals $(gh,f(gh))$ if and only if $f$ satisfies \eqref{eq:crossmor}.
\end{proof}

The previous result entails the following

\begin{corollary}
$f:G\rightarrow H$ is a crossed morphism relative to $\phi$ if and only if its graph is a subgroup of $G\ltimes_\phi H$.
\end{corollary}

\begin{proof} 
It suffices to note that the graph of $f$ coincides with the image of $(\mathrm{id}_G,f)$.
\end{proof}

\begin{remark}\label{rem:RRB}
\begin{enumerate}

\item \label{rem:RRB1}
Note that if $f:G \rightarrow H$ is a crossed morphism relative to a $G$-action on $H$, then $f(e_G)=e_H$. 

\item \label{rem:RRB2} Moreover, suppose that $f:G \rightarrow H$ is an invertible crossed morphism. Then relation \eqref{eq:crossmor}, read on $f^{-1}$, becomes for all $h_1,h_2\in H$
\begin{equation}
\label{relRBop}
	f^{-1}(h_1)f^{-1}(h_2)=f^{-1}(h_1\phi_{f^{-1}(h_1)}(h_2)) .
\end{equation}
This follows at once from \eqref{eq:crossmor}, if $h_i=f(g_i)$, $i=1,2$. Such an $f^{-1}$ is called a Rota--Baxter operator relative to $\phi$. See \cite[Def.~3.1]{JSZ2021} for the definition of the notion of relative Rota--Baxter operator. 
\end{enumerate}
\end{remark}

	
\subsubsection{Logarithmic derivative}	
\label{sssec:logderiv}

With the aim of introducing the infinitesimal analogue of a crossed morphism, we recall the following notion. Hereafter, $M$ will denote a smooth manifold, $H$ is a Lie group and $\mathfrak h$ its Lie algebra. Moreover, unless differently specified, the Maurer--Cartan (MC) form considered will be always left-invariant. We refer to Appendix \ref{sec:diffformV} for a collection of some background material.
	
\begin{definition}[Logarithmic derivative]
The logarithmic derivative of $f\in C^\infty(M,H)$ is the element $\delta f\in\Omega^1(M,\mathfrak h)$ defined for all $m \in M$ and $v \in T_mM$ by 
\begin{equation*}
	\delta f_m(v)=(L_{f(m)^{-1}})_{\ast,f(m)}f_{\ast,m}(v).
\end{equation*}
\end{definition}
	
Note that the formula defining the notion of logarithmic derivative is, by chain-rule, equivalent to 
\begin{equation}
	\delta f_m(v)=\theta_{f(m)}(f_{\ast,m}v)=(f^\ast\theta)_m(v), \qquad \forall m\in M,\,v\in T_mM.\label{eq:mci}
\end{equation}
In other words, the logarithmic derivative of $f:M\rightarrow H$ is the pull-back to $M$ via $f$ of the MC-element $\theta$ of $H$, that is
\begin{equation}
\label{eq:MC2}
	\delta f=f^\ast\theta.
\end{equation}
See Example \ref{ex:MCform} in Appendix \ref{sec:diffformV} and the comments that follow it.

\begin{lemma}\label{lem: log der MC element}
If $f\in C^\infty(M,H)$, then  $\delta f$ is an $\mathrm{MC}$-element of $\Omega^1(M,\mathfrak h)$.
\end{lemma}
	
\begin{proof}
Since $\delta f=f^*\theta$ and $\theta$ is a MC-element, we have that:
\begin{equation*}
	df^*\theta = f^*d\theta = -\frac{1}{2}f^*[\theta\wedge\theta] = -\frac{1}{2}[f^*\theta\wedge f^*\theta].
\end{equation*}
\end{proof}

Note that $C^\infty(M,H)$ is a group, whose binary operation, inverse and identity are induced by those of $H$, i.e.~for all $m\in M$ and $f,f_1,f_2\in C^\infty(M,H)$, we have $f_1f_2(m):=f_1(m)f_2(m)$, $f^{-1}(m):=f(m)^{-1}$ and the identity $1\in C^\infty(M,H)$ is the constant map $1:M\rightarrow H$  which associates to every $m \in M$ the identity element $e_H$. The logarithmic derivative interacts with the group structure of $C^\infty(M,H)$ as explained in the following proposition. See \cite{Neeb} and references therein.

\begin{proposition}
For all $f_1,f_2\in C^\infty(M,H)$, the following Leibniz rule holds true
\begin{equation}
\label{eq:leib}
	\delta (fg)=\delta f+\text{Ad}_{f^{-1}}\delta g.
\end{equation}
\end{proposition}

\begin{remark}
Note that Formula \eqref{eq:leib} means that for all $m\in M$
\[
	\delta (fg)_m=\delta f_m+\text{Ad}_{f^{-1}(m)}\delta g_m,
\]
is an identity between two applications from $T_mM$ to $\mathfrak h$, i.e., for all $v\in T_mM$ one has that
\[
	\delta (fg)_m(v)=\delta f_m(v)+\text{Ad}_{f^{-1}(m)}\delta g_m(v),
\]
is an identity in $\mathfrak h$.
\end{remark}

The above formula implies that 
\begin{equation}
\label{eq:one}
	\delta 1=0.
\end{equation}
In fact, taking $f=1=g$ in Formula \eqref{eq:leib} one gets $\delta 1=\delta (1\cdot1)=\delta 1+\delta 1$. Formulas \eqref{eq:leib} and \eqref{eq:one} imply for all $f \in C^\infty(M,H)$ that
\begin{equation}
	\delta f^{-1}=-\text{Ad}_f\delta f.\label{eq:inverse}
\end{equation}
In fact, for all $f\in C^\infty(M,H)$ the following holds
\[
	0\stackrel{\eqref{eq:one}}{=}\delta 1=\delta (f^{-1}f)\stackrel{\eqref{eq:leib}}{=}\delta f^{-1}+\text{Ad}_{f}\delta f.
\]
After these remarks we can go back to our setting where $G$ and $H$ are two Lie groups and let $\mathfrak g$ respectively $\mathfrak h$ be their Lie algebras. A given action $\phi:G\rightarrow\text{Aut}(H)$ induces a $G$-action on $\mathfrak h$, via the formula
\begin{equation}
	g.x=(\phi_g)_{\ast,e_H}x, \qquad \forall x\in\mathfrak h.\label{eq:Ghmor}
\end{equation} 
Under this assumption, $\alpha\in\Omega^1(G,\mathfrak h)$ will be called $G$-equivariant, see Definition \ref{def:Gequiv} in Appendix \ref{sec:diffformV}, if	 
\begin{equation}
	L_{g'}^\ast\alpha_g=g'.\alpha_g, \qquad \forall g,g'\in G, \label{eq:equiform}
\end{equation}
where the left-hand side of the previous formula should be read as follows. For all $v\in T_gG$
\[
g'.\alpha_g(v)\stackrel{\eqref{eq:Ghmor}}{=}(\phi_g)_{\ast,e_H}\alpha_g(v).
\]

It is worth recalling that the previous identity is equivalent to
\begin{equation}
\label{eq:equiform1}	
	\alpha_{g'g}\big((L_{g'})_{\ast,g}v\big)=g'.\alpha_g(v),	
\end{equation}
for all $g',g\in G$ and $v\in T_gG$, see Formula \eqref{eq:equivK1} in Appendix \ref{sec:diffformV}, which, in turn, implies that for every $\xi\in\mathfrak g$
\begin{equation}
\label{eq:equiform2}
	\alpha_{g'}\big(X_{\xi}(g')\big)=g'.\alpha_{e_G}(\xi),	
\end{equation}
where $X_\xi$ is the left-invariant vector field on $G$ defined by $\xi$. One can prove the following 
	
\begin{theorem}\label{th: crossed morph gp G-inv}
If $f\in C^\infty(G,H)$ is a crossed morphism relative to the $G$-action $\phi$, then
\begin{enumerate}
	\item $f(e_G)=e_H$;
	\item $\delta f\in\Omega^1(G,\mathfrak h)$ is $G$-equivariant.
\end{enumerate}
\end{theorem}
	
\begin{proof} 
In Remark \ref{rem:RRB}, it was already observed that if $f:G\rightarrow H$ is a crossed morphism, then $f(e_G)=e_H$. We show now that such an $f$ satisfies the second condition in the statement, i.e.~that 
\begin{equation}
\label{eq:identity}
	(\delta f)_{g'g}\big((L_{g'})_{\ast,g}v\big)=g'.(\delta f)_g(v),  \qquad  \forall g',g\in G,\;v\in T_gG,	
\end{equation}
see Formula \eqref{eq:equiform1}. To this end, let us first compute the left-hand side of the previous identity:
\begin{eqnarray*}
	(\delta f)_{g'g}\big((L_{g'})_{\ast,g}v\big)&=&\theta_{f(g'g)}\big(f_{\ast,g'g}(L_{g'})_{\ast,g}v\big)\\
	&=&\theta_{f(g'g)}\big((f\circ L_{g'})_{\ast,g}v\big)\\
	&=&\theta_{f(g')\phi_{g'}(f(g))}\big((f\circ L_{g'})_{\ast,g}v\big).
\end{eqnarray*}
Now one can observe that for every $g'\in G$
\[
	f \circ L_{g'}=L_{f(g)}\circ\phi_{g'}\circ f,
\]
which entails
\begin{eqnarray*}
	\theta_{f(g')\phi_{g'}(f(g))}\big((f\circ L_{g'})_{\ast,g}v\big)
	&=&\theta_{f(g')\phi_{g'}(f(g))}\big((L_{f(g')}\circ\phi_{g'}\circ f)_{\ast,g}v\big)\\
	&=&\big(L_{(f(g')\phi_{g'}(f(g)))^{-1}}\big)_{\ast,f(g')\phi_{g'}(f(g))}\big((L_{f(g')}\circ\phi_{g'}\circ f)_{\ast,g}v\big)\\
	&=&\big(L_{(\phi_{g'}(f(g)))^{-1}}\circ\phi_{g'}\circ f\big)_{\ast,g}v.
\end{eqnarray*}
On the other hand, first one observes that the right-hand side of \eqref{eq:identity} is 
\begin{eqnarray*}
	g'.(\delta f)_g(v)
	&=&(\phi_{g'})_{\ast,e_H}(\delta f)_g (v)=(\phi_{g'})_{\ast,e_H}\big(\theta_{f(g)}(f_{\ast,g}v)\big)\\
	&=& (\phi_{g'})_{\ast,e_H}(L_{f(g)^{-1}})_{\ast,f(g)}(f_{\ast,g}v)\\
	&=& (\phi_{g'}\circ L_{f(g)^{-1}}\circ f)_{\ast,g}(v).
\end{eqnarray*}
Now observe that 
\[
	(\phi_{g'}\circ L_{f(g)^{-1}}\circ f)(k)
	=\phi_{g'}(f(g)^{-1}f(k))=\phi_{g'}(f(g)^{-1})\phi_{g'}(f(k)),  \qquad \forall k\in G
\]
i.e.
\[
	\phi_{g'}\circ L_{f(g)^{-1}}\circ f
	=L_{\phi_{g'}(f(g)^{-1})}\circ\phi_{g'}\circ f=L_{(\phi_{g'}(f(g)))^{-1}}\circ\phi_{g'}\circ f,
\]
which implies the following equality
\[
	g'.(\delta f)_g(v)=(\phi_{g'})_{\ast,e_H}(\delta f)_g (v)=(L_{(\phi_{g'}(f(g)))^{-1}}\circ\phi_{g'}\circ f)_{\ast,g}v,
\]
proving the statement. 
\end{proof}

\begin{remark}
If $G$ is connected, the previous conditions suffice to characterize the crossed morphisms among the smooth maps between $G$ and $H$. More precisely, one can prove that if $G$ is connected, then $f:G\rightarrow H$ is crossed morphism (relative to a given $G$-action on $H$) if and only if 1) $f(e_G)=e_H$ and 2) $\delta f$ is a $G$-equivariant one-form with values in $\mathfrak h$, see Lemma II 3.5 in \cite{Neeb}. 
\end{remark}


\subsection{Lie algebras}
\label{ssec:liealgebras}


\subsubsection{Crossed morphisms}
\label{sssec:crossedmorph}

Let $\h$ and $\g$ be two Lie algebras and  $\alpha\co \g \to \der(\h)$ be a morphism of Lie algebras. Note that the existence of such a morphism is equivalent to the existence of a $\g$-action on $\h$. In a similar way to the group case, the direct sum $\h\oplus \g$ can be endowed with a structure of a Lie algebra. 
Explicitly, for all $y,y'\in \h$ and for all $x,x'\in \g$ the  bracket 
\begin{equation*}
		[(y,x),(y',x')] =  \Big( [y,y']_\h + \alpha_x(y')-\alpha_{x'}(y),  [x,x']_\g \Big),
\end{equation*}
satisfies the Jacobi identity. We denote by $\h\rtimes_\alpha \g=(\h\oplus\g,[-,-])$ the resulting Lie algebra. 
It is an extension of $\g$ by $\h$, usually called the semi-direct product of $\g$ and $\h$, see also Lemma \ref{lem:semdirG}. In other words, one has the following short exact sequence of Lie algebras 
\begin{equation}\label{eq: cano. extension}
	0\to \h\xrightarrow{i} \h\rtimes_\alpha \g \xrightarrow{\pi} \g \to 0 ,
\end{equation}
where $i$ is the canonical inclusion and $\pi$ is the projection on the second factor. Now, consider a section of the underlying sequence of vector spaces, that is a linear map  $s\co\g\to  \h\rtimes_\alpha \g$ such that  $\pi\circ s = \id$. Clearly, $s(x) = (\psi(x),x)$ for some linear map $\psi\co \g \to \h$. 

\begin{definition}[Crossed morphism of Lie algebras]
A linear map $\psi\co \g \to \h$ is called a \emph{crossed morphism} relative to the $\g$-action on $\h$ defined by Lie morphism $\alpha:\g\rightarrow\text{Der}(\h)$ if, for all $x,y\in \g$, one has 
\begin{equation}\label{eq:cmi}
	\psi([x,y]_\g) = \alpha_x(\psi(y)) - \alpha_{y}(\psi(x))+[\psi(x),\psi(y)]_\h.
\end{equation}
\end{definition}
 		
\begin{lemma}
The section $s$ is a morphism of Lie algebras if and only if $\psi$ is a crossed morphism.  
 \end{lemma}

\begin{example}
If $\mathfrak h$ is abelian, i.e.~if $[-,-]_{\mathfrak h}\equiv 0$, then the above map $\psi$ is a crossed morphism relative to $\alpha\co \g \to \operatorname{End}_{\mathbb K}(\h)$ if and only if for all $x,y \in \mathfrak g$
\[
	\psi([x,y]_{\mathfrak g})=\alpha_x(\psi(y))-\alpha_y(\psi(x)).
\]
\end{example}

\begin{example}
If $f\in\operatorname{Hom}_{\text{Lie}}(\mathfrak g,\mathfrak h)$ then $\alpha_f:\mathfrak g\rightarrow\operatorname{Der}(\mathfrak h)$, defined for all $y\in\mathfrak h$ by
\[
	\alpha_{f}(x)(y)=[f(x),y]_{\mathfrak h},
\]
is a morphism of Lie algebras. 
Moreover, $\psi\co \g\to \h$ is a crossed morphism relative to $\alpha_{f}$  if and only if $f+\psi\in\operatorname{Hom}_{\text{Lie}}(\mathfrak g,\mathfrak h)$.
\end{example}

\begin{example}\label{ex: phi is id}
If $\g$ and $\h$ share the same underlying vector space, say $V$, and if $\psi$ is the identity endomorphism of $V$, then \eqref{eq:cmi} becomes
\begin{equation}\label{eq: h bar }
	[ x,y ]_{\g} = \alpha_{x}(y) - \alpha_y (x) + [x,y]_{\h} \qquad \text{ for all } x,y \in V. 
\end{equation}
In other words, $V$ has two related Lie algebra structures. The Lie algebra $\h$ sometimes is dubbed the \emph{initial} Lie algebra and hereafter we will write $\llbracket -,-\rrbracket$ to denote $[-,-]_\g$, i.e. $\g:=(V,\llbracket-,-\rrbracket)$.
\end{example}
		
\begin{example}\label{ex:invR}
Suppose	$\psi\co \g \to \h$ is an invertible crossed morphism and let $R:=\psi^{-1}$. Then
\begin{equation*}
	[R(x),R(y)]_\g = R\big(  [x,y]_\h + \alpha_{R(x)}(y) - \alpha_{R(y)}(x) \big).
\end{equation*}
Moreover, if $\h$ and $\g$ share the same underlying vector space and if $\alpha_h(g)= [h,g]_{\h}$, i.e.~$\alpha$ is the adjoint action, then $R$ is called a Lie Rota--Baxter operator of weight 1 (or of weight 0 if $\h$ was abelian). In other words, the \emph{invertible} crossed morphisms relative to the adjoint action coincide with the invertible Lie Rota--Baxter operators (of weight $1$ or $0$). However, there are Lie Rota--Baxter operators (of weight 0,1) which do not come from crossed morphisms. 
\end{example}

We will prove now the following result whose Lie group analogue is Proposition \ref{prop:invcross}. 

\begin{proposition}\label{prop: adjun}
Let $\phi \co \g \to \h$ be an \emph{invertible} crossed morphism relative to $\alpha \co \g\to \der(\h)$. 
The following bracket 
\begin{equation}
\label{eq:crosalgLie}
	\llbracket x,y\rrbracket:=\alpha_{\phi^{-1}(x)}(y)-\alpha_{\phi^{-1}(y)}(x)+[x,y]_\h,\qquad \forall x,y \in \h.
\end{equation}
defines another structure of Lie algebra, denoted by $\overline\h$,  on the underlying vector space of $\h$. 
Moreover, the identity map  $\id:\overline\h\rightarrow\h$ is a crossed morphism relative to the action $\alpha\circ\phi^{-1}$ of $\h$ over $\h$.
\end{proposition}

\begin{proof}
First note that $\alpha_{\phi^{-1}(\llbracket x,y\rrbracket)} = \alpha_{\phi^{-1}(x)} \circ \alpha_{\phi^{-1}(y)} - \alpha_{\phi^{-1}(y)} \circ \alpha_{\phi^{-1}(x)}$, as one can check by recalling that for $\xi=\phi^{-1}(x)$, $\eta=\phi^{-1}(y)\in\g$, $\phi([\xi,\eta]_\g)=\alpha_\xi(\phi(\eta))-\alpha_\eta(\phi(\xi))+[\phi(\xi),\phi(\eta)]_\h$. The proof that \eqref{eq:crosalgLie} is a Lie bracket on $V$, i.e.~that it satisfies the Jacobi identity, is a tedious though straightforward computation which we prefer not to present. Finally, the last part of the proposition follows at once from the definition of crossed morphism.
\end{proof}

As a consequence of the above proposition, if one is interested in studying invertible crossed morphisms, it may be  enough to consider the setting of Example \ref{ex: phi is id}. For a more precise statement of this fact, see \cite[Prop.~6]{MQS}. 
	
\begin{remark}
Let 
\begin{equation*}
	 0\to \h \to  \mathfrak{e}  \to \g \to 0 
\end{equation*}
be an extension  of Lie algebras. 
A \emph{splitting} of the above sequence is a map $s\co \g\to \e$ that is both a section and a morphism of Lie algebras.
 One can show that a short exact sequence admits a splitting if and only if it is isomorphic to the extension \eqref{eq: cano. extension} and there exists a crossed morphism $\phi\co \g \to \h$ relative to $\alpha$ (given by \eqref{eq: cano. extension}). 
\end{remark}

A morphism between two short exact sequences 
\begin{equation*}\
	0\to \h \xrightarrow{i} \h \rtimes_\alpha \g \xrightarrow{\pi} \g \to 0 
\end{equation*}
and 
\begin{equation*}
		0 \to \h'\xrightarrow{i'} \h'\rtimes_{\alpha'} \g' \xrightarrow{\pi'} \g' \to 0 
\end{equation*}
that admits a splitting, is a pair $(f,g)$ of morphisms of Lie algebras, $f \co \h\to \h'$ and $g\co \g\to \g'$, such that 
	\begin{enumM}
		\item\label{M1}$f\circ\phi=\phi'\circ g$ and
		\item\label{M2}$f(\alpha_{x}(y))=\alpha'_{g(x)}(f(y)) \text{ for all } x \in \g \text{ and } y\in\h$. 
	\end{enumM}	
In particular, let us consider the subcategory of sequences as in Example \ref{ex: phi is id}:  $\h$ and $\g$ share the same underlying vector space, say $V$, and $\id_V$ is a crossed morphism. A morphism in this category is a morphism of Lie algebras $f \co \h\to \h'$ such that  for all $x,y \in V$
\begin{equation*}
	f(\alpha_{x}(y))=\alpha'_{f(x)}(f(y)). 
\end{equation*}
We denote by $\cat{SE}$ (for split special extensions) this category.


\subsubsection{From groups to Lie algebras}
\label{sssec:group2lie}

Let $\alpha\co \g\to \der(\h)$ be a morphism of Lie algebras and recall the dgla $(C^*(\g,\h),[-,-])$ from Proposition \ref{prop: Lie on CE} in Appendix \ref{sec:diffformV}.

\begin{lemma}\label{lem:crosMC}
A linear map $\psi\co \g \to \h$ is a crossed morphism relative to $\alpha$ if and only if it is an ${\mathrm{MC}}$-element of $C^\bullet(\g,\h)$. 
\end{lemma}  

\begin{proof} For all $x,y\in \g$, one has:
\begin{equation}
\label{eq:MCpl}
	\big(d_{CE}\psi +\frac{1}{2}\{\psi,\psi\}\big)(x,y) 
	= \alpha_x(\psi(y)) - \alpha_{y}(\psi(x)) - \psi([x,y]_\g) + [\psi(x),\psi(y)]_\h.
\end{equation}
\end{proof}

\begin{proposition}\label{pro:37}
Every crossed morphism of Lie groups induces a crossed morphism of Lie algebras.
\end{proposition}
	
\begin{proof}
By Theorem \ref{th: crossed morph gp G-inv}, for any crossed morphism $f\co G\to H$, its logarithmic derivative $\delta f$ is an  element of $\Omega^1(G,\h)^G$, which turns out to be also an MC-element (Lemma  \ref{lem: log der MC element}). Finally, by Theorem \ref{th: DeRham and CE} and Proposition \ref{prop: de Rham CE iso of Lie}, $\delta f$ corresponds to an MC-element of $C^\bullet(\g,\h)$. 
We conclude by means of the previous Lemma. Note that, if $f:G\rightarrow H$ is a crossed morphism relative to the $G$-action on $H$ defined by $\phi: G\rightarrow\text{Aut}(H)$, the induced crossed morphism $\psi:\g\rightarrow\h$ will be relative to the $\g$-action on $\h$ defined by the Lie algebra morphism $\alpha:\g\rightarrow\text{Der}(\h)$ that is defined by the formula 
\begin{equation}
\label{eq:derx}
	\alpha_x(\xi):=\left.\frac{d}{dt}\right\vert_{t=0}(\phi_{\exp(tx)})_{\ast,e_G}(\xi), \qquad \forall x\in\g,\xi\in\h.
\end{equation}
The latter is obtained by recalling that the $G$-action $\phi$ on $H$ defines a morphism of Lie groups $(\phi_{\cdot})_{\ast,e_G}:G\rightarrow\text{Aut}(\h)$, $g\mapsto (\phi_g)_{\ast,e_G}$, see also \eqref{eq:Ghmor}. Then $\alpha$ in the formula above is obtained differentiating this Lie group morphism at the identity of $G$. We close the proof noticing that the crossed morphism of Lie algebras so obtained is the differential at the identity of $f$. In fact, for all $g\in G$ and $v\in T_gG$
\[
	\delta f_g(v)=\theta_{f(g)}(f_{\ast,g}(v))
\]
see \eqref{eq:mci}, where $\theta$ is the (left) Maurer--Cartan 1-form of $H$. Since $f$ is a crossed morphism, after evaluating both sides of the previous identity at $e_G$, one deduces that for all $x\in\g$ 
\[
	\delta f_{e_G}(x)=\theta_{e_H}(f_{\ast,e_G}(x)),
\]
see (1) in Theorem \ref{th: crossed morph gp G-inv}. To conclude it suffices observe that $\theta_{e_H}=\text{id}\in\text{End}(\h)$.
\end{proof}
 
\begin{remark}
In general, it is not possible to \emph{integrate} crossed morphisms of Lie algebras to crossed morphisms of Lie group, see  \cite{Neeb} and \cite{Sharpe} for a detailed discussion about this point.
\end{remark}


\section{Crossed morphisms and post-Lie structures}
\label{sec: crossed to post-Lie}

In this section, we recall the notion of post-Lie algebra and of post-(Lie) group. We show that both amount to a certain type of extensions (split special extensions); these correspond to a certain type of crossed morphisms. Therefore, in essence, post(-Lie) groups and post-Lie algebras amount to the same concept, i.e.~that of split special extension, but in different categories: the one of (Lie) groups, the other one of Lie algebras.


\subsection{Post-Lie groups \cite{Bai-Guo-Sheng-Tang-post-groups}}

In this subsection, we will consider a special class of crossed morphisms of Lie groups, see Subsection \ref{sec: subsec cross morph gp}. We start our analysis with the following result whose infinitesimal analogue is Proposition \ref{prop: adjun}.

\begin{proposition}
\label{prop:invcross} 
Let $G,H$ be two Lie groups. If $f:G\rightarrow H$ is an invertible crossed morphism relative to the action $\phi:G\rightarrow\text{Aut}(H)$, then $\mathrm{id}:\overline{H}\rightarrow H$ is a crossed morphism relative to the $G$-action $\phi\circ f^{-1}:\overline{H}\rightarrow\text{Aut}(H)$ where $\overline{H}=(H,\star)$ is the Lie group whose group operation is defined by
\begin{equation}
	h_1\star h_2=h_1\phi_{f^{-1}(h_1)}(h_2),\qquad \forall h_1,h_2\in H.\label{eq:staract}
\end{equation}
Note that $H$ and $\overline H$ share the same underlying manifold and for this reason one can consider the identity map between $\overline{H}$ and $H$.
\end{proposition}

\begin{proof}
First, note that 
\begin{equation}
	\phi_{f^{-1}(h_1\star h_2)}=\phi_{f^{-1}(h_1)}\circ\phi_{f^{-1}(h_2)},\label{eq:idf}
\end{equation} 
for all $h_1,h_2\in H$. In fact since $f:G\rightarrow H$ is an invertible crossed morphism relative to the $G$-action $\phi$ on $H$, $f^{-1}$ satisfies 
\[
	f^{-1}(h_1\star h_2)=f^{-1}(h_1)f^{-1}(h_2), \qquad \forall h_1,h_2\in H,
\] 
see Remark \ref{rem:RRB}.
Then, since $\phi$ is a $G$-action on $H$
\begin{eqnarray*}
	\phi_{f^{-1}(h_1\star h_2)}=\phi_{f^{-1}(h_1)f^{-1}(h_2)}=\phi_{f^{-1}(h_1)}\circ\phi_{f^{-1}(h_2)}.
\end{eqnarray*} 
Observe now that \eqref{eq:staract} is associative. In fact
\begin{eqnarray*}
	h_1\star (h_2\star h_3)
	&=&h_1\star (h_2\phi_{f^{-1}(h_2)}(h_3))=h_1\phi_{f^{-1}(h_1)}\big(h_2\phi_{f^{-1}(h_2)}(h_3)\big)\\
	&=&h_1\phi_{f^{-1}(h_1)}(h_2)\phi_{f^{-1}(h_1)}\big(\phi_{f^{-1}(h_2)}(h_3)\big)\\
	&\stackrel{\eqref{eq:idf}}{=}&h_1\phi_{f^{-1}(h_1)}(h_2)(\phi_{f^{-1}(h_1\star h_2)}(h_3))\\
	&=&(h_1\star h_2)\star h_3.
\end{eqnarray*}
Note that for all $h\in H$, $h\star e_H=h=e_H\star h$ and that every $h$ has a $\star$-inverse, i.e.~the element $\phi_{(f^{-1}(h))^{-1}}(h^{-1})$. These comments entail that $\overline{H}=(H,\star)$ is a group. Now it suffices to notice that all the operations involved in the definition of $\star$ and of the inverse are smooth to conclude that $\overline{H}=(H,\star)$ is a Lie group. The statement that $\mathrm{id}:\overline{H}\rightarrow H$ is a crossed morphism relative to $\phi\circ f^{-1}$ follows from the definition of crossed morphism. 
\end{proof}

Consider two Lie groups, $G$ and $H$, that share the same underlying set. To emphasize the latter fact, let us denote $G$ by $\overline{H}$ and by $\star$ its product. Suppose there is an action $\phi\co \overline{H} \to \Aut(H)$ and that the identity map $\id \co \overline{H} \to H$ is a crossed morphism relative to $\phi$. We call such a tuple $(H,\overline{H},\phi,\id)$ a split special extension. We let $\cat{SE(gp)}$ denote the category of split special extensions; a morphism between two split special extensions $(H,\overline{H},\phi,\id)$ and $(H',\overline{H}',\phi',\id)$ is a morphism of Lie groups $p\co H\to H'$ such that 
\begin{equation*}
	p(\phi_{g}(\tilde{g}))=\phi'_{p(g)}(p(\tilde{g})) \qquad \text{ for all } g,\tilde{g}\in H. 
\end{equation*}

As a consequence of the fact that $\id\co \overline{H} \to H$ is a crossed morphism, one has for all $g,\tilde{g}\in H$
\begin{equation*}
 	g\star \tilde{g} =  \id(g\star \tilde{g}) = g\phi_{g}(\tilde{g}).
\end{equation*}

Let us set 
\begin{equation}
g\blacktriangleright \tilde{g} := \phi_{g}(\tilde{g}).\label{eq:dblacktria}
\end{equation}
That $\phi$ is a morphism of groups means that for all $g,\tilde{g},h\in H$
\begin{equation*}
	\phi_{g\star \tilde{g}}(h) = \phi_{g}( \phi_{\tilde{g}}(h)).
\end{equation*}
The left-hand side equals $\big( g\cdot (g\blacktriangleright \tilde{g})\big) \blacktriangleright h$ and the right-hand side is $g\blacktriangleright (\tilde{g}\blacktriangleright h)$. Putting these properties together, one obtains the following definition, which appeared in \cite{Bai-Guo-Sheng-Tang-post-groups}\footnote{We invite the reader to look into the earlier reference \cite{MQ-Crossed}, where the problem of integration of post-{L}ie algebras appeared in the framework of crossed morphisms.}. 

\begin{definition}\cite{Bai-Guo-Sheng-Tang-post-groups}
\label{def:post-LieG}
A \emph{post-group} is a group $(H,\cdot)$ equipped with a binary operation $\blacktriangleright\co H\times H\to H$ satisfying the following relations:
\begin{enumerate}
	
	\item for all $g\in H$, the map $\phi_g \co H \to H$ defined by $\phi_{g}(\tilde{g}) = g\blacktriangleright \tilde{g}$
	is an automorphism of $(H,\cdot)$.  This means that $g\blacktriangleright (g'\cdot h) = (g\blacktriangleright \tilde{g})\cdot(g\blacktriangleright h)$ and $g\blacktriangleright e_H = e_H$ for all $g,\tilde{g},h\in H$. 
	
	\item for all $g,\tilde{g},h\in H$, one has  $(g\cdot (g\blacktriangleright \tilde{g}))\blacktriangleright h = g\blacktriangleright (\tilde{g}\blacktriangleright h)$. 
\end{enumerate}
A \emph{post-Lie group} is a post-group $(H,\cdot,\blacktriangleright)$ such that $(H,\cdot)$ is a Lie group and $\blacktriangleright$ is smooth. 
\end{definition}

We let $\cat{PostLie(gp)}$ be the category of (finite dimensional, connected) post-Lie groups. Morphisms are morphisms of Lie groups that are compatible with the product $\blacktriangleright$. 

\begin{proposition}\label{prop: iso SE PL groups}
	The two categories $\cat{SE(gp)}$ and $\cat{PostLie(gp)}$ are isomorphic. 
\end{proposition}

\begin{proof}
We already saw that a split  special extension provides a post-Lie group. Let us see the converse. Let $(H,\cdot,\blacktriangleright)$ be a post-Lie group. Define, for all $g,h\in H$, 
\begin{equation}
	g\star h:=g\phi_g(h),\label{eq:pg1}
\end{equation}	
where we set $\phi_g(h):=g\blacktriangleright h$, see \eqref{eq:dblacktria}.
Let $\overline{H}$ denote the pair $(H,\star)$, i.e.~the manifold underlying the Lie group $H$ endowed with the binary operation \eqref{eq:pg1}. 
Note that \eqref{eq:pg1} is a smooth application from $H\times H$ to $H$. 

Recall that the second axiom of a post-Lie group is equivalent to $\phi_{g \star g'} = \phi_{g}\circ  \phi_{g'}$. Therefore, it remains to prove the following lemma.

\begin{lemma}\label{prop:postg}
	If the map $\phi:\overline{H}\rightarrow\text{Aut}(H)$ satisfies
	\begin{equation}
		\phi(g\star h)=\phi_g\circ\phi_h,\qquad \forall g,h\in H,\label{eq:pg2}
	\end{equation}
	then $\overline{H}$ is a Lie group.
\end{lemma}

\begin{proof}
First one checks associativity:
\begin{eqnarray*}
	h_1\star(h_2\star h_3)
	&=&h_1\star(h_2\phi_{h_2}(h_3))
	=h_1\phi_{h_1}(h_2\phi_{h_2}(h_3))
	=h_1(\phi_{h_1}(h_2)\phi_{h_1}(\phi_{h_2}(h_3)))\\
	&=&(h_1\phi_{h_1}(h_2))(\phi_{h_1}\circ\phi_{h_2}(h_3))\\
	&=&(h_1\star h_2)(\phi_{h_1}\circ\phi_{h_2}(h_3)),
	\end{eqnarray*}
which follows from the associativity of the product $\cdot$ of $H$ and from the definition of $\star$. Assuming \eqref{eq:pg2}, the last term is equal to $(h_1\star h_2)(\phi_{h_1\star h_2}(h_3))$ which implies the associativity of $\star$. To identify the identity as well as inverses, one observes
\[
	h\star e_H=h\phi_h(e_H)=he_H=h,\qquad \forall h\in H,
\]
i.e.~$e_H$ is a right-identity. This implies that given $h\in H$
\[
	\phi_h\circ\phi_{e_H}=\phi_{h\star e_H}=\phi_h,
\]
which yields
\begin{equation}
	\phi_{e_H}=\text{Id}\in\text{Aut}(H).\label{eq:id*}
\end{equation}
On the other hand, $e_H\star h=e_H\phi_{e_H}(h)\stackrel{\eqref{eq:id*}}{=}h$, showing that $e_H$ is also a left-identity of $\star$. Finally, given $h\in H$, one computes
\begin{eqnarray}\label{eq:pg3}
	h\star\phi_h^{-1}(h^{-1})=h\phi_h(\phi_h^{-1}(h^{-1}))=e_H.
\end{eqnarray}
On the other hand,
\begin{eqnarray*}
	\phi_h^{-1}(h^{-1})\star h
	=\phi_h^{-1}(h^{-1})\phi_{\phi_h^{-1}(h^{-1})}(h)
	&=&\phi_h^{-1}\big(h^{-1}\phi_h(\phi_{\phi_h^{-1}(h^{-1})}(h))\big)\\
	&\stackrel{\eqref{eq:pg2}}{=}&\phi_h^{-1}\big(h^{-1}\phi_{h\star\phi_h^{-1}(h^{-1})}(h)\big)\\
	&\stackrel{\eqref{eq:pg3}}{=}&\phi_h^{-1}\big(h^{-1}\phi_{e_H}(h)\big)\\
	&\stackrel{\eqref{eq:id*}}{=}&\phi_h^{-1}(h^{-1}h))\\
	&=&e_H,
\end{eqnarray*}
which, together with \eqref{eq:pg3}, shows that $\phi_h^{-1}(h^{-1})$ is the $\star$-inverse of $h$.
\end{proof}

We conclude, by observing that the identity map $\id\co \overline{H}\rightarrow H$ is a crossed homomorphism relative to the $\overline{H}$-action on $H$ defined by $\phi$. Finally, it is straightforward to check that these assignments are inverse to each other and are functorial. 
\end{proof}


\subsection{Post-Lie algebras}
\label{subsec:postLie}

In this section we recall the definition of post-Lie algebra \cite{BV2007}. See also \cite{MKW}. The reader is referred to \cite{CEFMK} for a short review.

\begin{definition}
A \emph{post-Lie algebra} is a pair $(\h,\triangleright)$ consisting of a Lie algebra $\h=(V,[-,-])$ and a binary operation
$\pl\co V\ot V \to V$ such that for all $x,y,z\in V$: 
\begin{enumPL}
	\item\label{PL Property 1} $x\triangleright [y,z]=[x\triangleright y,z]+[y,x\triangleright z]$, and
	\item\label{PL Property 2} $[x,y]\triangleright z={\rm a}_{\triangleright}(x,y,z)-{\rm a}_{\triangleright}(y,x,z)$.
\end{enumPL}
On the right-hand side of \ref{PL Property 2},  ${\rm a}_\triangleright$ denotes the associator
\begin{equation}
\label{eq:ass}
	{\rm a}_\triangleright(x,y,z):=x\triangleright(y\triangleright z)
	-(x\triangleright y)\triangleright z,  
\end{equation}
for the bilinear product $\triangleright$. The latter is called the \emph{post-Lie product} of the post-Lie algebra.
\end{definition}

Note that \ref{PL Property 1} is equivalent to the existence of a linear map $\alpha:V\rightarrow\text{Der}(\h)$, $x\mapsto\alpha_x$, where $\alpha_x(y):=x\triangleright y$, for all $y\in V$. On the other hand, \ref{PL Property 2} implies that:
\begin{enumerate}
\item on the underlying vector space $V$ of a post-Lie algebra $(\h,\triangleright)$ the binary operation $\llbracket-,-\rrbracket:V\times V\rightarrow V$ defined for all $x,y\in V$ by
\begin{equation}
\label{eq:secbra}
	\llbracket x,y\rrbracket :=x\triangleright y-y\triangleright x+[x,y],
\end{equation}
gives a second Lie algebra $\g:=(V,\llbracket-,-\rrbracket)$ and 
\item  $\alpha:\g\rightarrow\text{Der}(\h)$ is a morphism of Lie algebras.
\end{enumerate}
On the other hand, let $\h=(V,[-,-])$ be a Lie algebra and let $\alpha:V\rightarrow\text{Der}(\h)$ be a linear map. 
Now define $\llbracket-,-\rrbracket:V\times V\rightarrow V$ for all $x,y \in V$ by
\begin{equation}
	\llbracket x,y\rrbracket=\alpha_x(y)-\alpha_y(x)+[x,y].\label{eq:plc}
\end{equation}
 
\begin{lemma}\label{lem:pl}
If $\alpha_{\llbracket x,y\rrbracket}=\alpha_x\circ\alpha_y-\alpha_y\circ\alpha_x$ for all $x,y\in V$, then \eqref{eq:plc} is a Lie bracket. 
\end{lemma}

\begin{proof}
It suffices to prove that if $\alpha_{\llbracket x,y\rrbracket}=\alpha_x\circ\alpha_y-\alpha_y\circ\alpha_x$ holds true for all $x,y\in V$, then $\llbracket\llbracket x,y\rrbracket,z\rrbracket+\llbracket\llbracket z,x\rrbracket,y\rrbracket+\llbracket\llbracket y,z\rrbracket,x\rrbracket=0$ for all $x,y,z\in\ V$ and this can be checked by computation.
\end{proof}

\begin{remark}
	Note that the previous lemma gives a sufficient but not necessary condition for \eqref{eq:plc} being a Lie bracket. In fact, if on the Lie algebra $\h$ one defines $\alpha_x=\text{ad}_x=[x,-]$ for all $x\in\h$, then $\llbracket-,-\rrbracket=3[-,-]$, implying that $\llbracket-,-\rrbracket$ is a Lie bracket and, at the same time, that $\alpha_{\llbracket x,y\rrbracket}=3\alpha_{[x,y]}$ for all $x,y\in\h$. On the other hand, $\alpha_x\circ\alpha_y-\alpha_y\circ\alpha_x=\alpha_{[x,y]}$, showing that this choice of $\alpha$ does not fulfill the hypothesis of the lemma. 
\end{remark}

These comments entail the following  

\begin{proposition}
\label{prop:pl}
A Lie algebra $\h=(V,[-,-])$ carries the structure of a post-Lie algebra $(\h,\triangleright)$ if and only if a linear map $\alpha:V\rightarrow\text{Der}(\h)$ is defined such that $\alpha_{\llbracket-,-\rrbracket}=\alpha_x\circ\alpha_y-\alpha_y\circ\alpha_y$ for all $x,y\in V$, where $\llbracket-,-\rrbracket$ is as in \eqref{eq:plc}. 
\end{proposition}

\begin{proof}
It suffices noticing that, writing $\alpha_x:=x\;\triangleright$, \ref{PL Property 2} is equivalent to the condition $\alpha(\llbracket x,y\rrbracket)=\alpha_x\circ\alpha_y-\alpha_y\circ\alpha_x$. Furthermore, observe that, if this condition is fulfilled, $\alpha:\g=(V,\llbracket-,-\rrbracket)\rightarrow \text{Der}(\h)$ is (automatically) a Lie algebra morphism.
\end{proof}

\begin{remark}
Observe that under the hypothesis of Lemma \ref{lem:pl}, $\text{id}:\g\rightarrow\h$ is a crossed morphism relative to the $\g$-action on $\h$ defined by $\alpha:\g\rightarrow\text{Der}(\h)$.
\end{remark}

\begin{example}[Pre-Lie algebras]
\label{ex:preLie}
An abelian post-Lie algebra, i.e.~a post-Lie algebra whose Lie algebra $\h$ is abelian, is called a pre-Lie algebra. In this case \ref{PL Property 1} is an empty relation and  \ref{PL Property 2} says that the associator ${\rm a}_\triangleright$ is symmetric in the first two variables. This property suffices to conclude that $\llbracket x,y\rrbracket:=x\pl y-y\pl x$ is a Lie bracket defined on the underlying vector space $V$ of $\h$.
\end{example}

A morphism of post-Lie algebras is a morphism of the underlying Lie algebras that is compatible with the post-Lie products. 
We let $\cat{PostLie}$ denote the resulting category. 

\begin{proposition}
\label{pro:equiv}
	The two categories $\cat{SE}$ and $\cat{PostLie}$ are isomorphic. 
\end{proposition}  

\begin{proof}
To an exact sequence of $\cat{SE}$, i.e., a tuple $(\overline{\h},\h,\alpha,\id_V)$, one may associate the post-Lie algebra $(\h,\pl)$, where for all $x,y\in \h$
\begin{equation}
\label{eq: postlie from hh,v ,id}
	x\pl y := \alpha_{x}(y).  
\end{equation}
Indeed, \ref{PL Property 1} is clear since $\alpha_x$ is a derivation  of $\h$, and  \ref{PL Property 2} results from the fact that $\alpha\co \overline{\h}\to \Der(\h)$ is a Lie morphism: for all $x,y$ and $z$ in $\h$, one has 
\begin{equation*}
	[x,y]\pl z
	=\alpha_{[x,y]}(z)
	=\alpha_{\llbracket x,y \rrbracket - \alpha_{x}(y) + \alpha_{y}(x)}(z)\\
	=\alpha_{x}(\alpha_{y}(z))-\alpha_{y}(\alpha_{x}(z))-\alpha_{\alpha_{x}(y)}(z)+\alpha_{\alpha_{y}(x)}(z).
\end{equation*}
By a straightforward computation one verifies that this assignment induces an isomorphism of categories. 
\end{proof}


\subsection{From post-Lie groups to post-Lie algebras: the crossed morphism approach}
\label{ss:postLie}

Here we show how the setting of Section \ref{sec:crossedmorphism} forces, at the infinitesimal level, the existence of a post-Lie algebra structure. \\

Let $(H,\cdot,\blacktriangleright)$ be a post-Lie group or, equivalently, see Proposition \ref{prop: iso SE PL groups}, let $(H,\cdot)$ and $\overline{H}=(H,\star)$ be two Lie groups, $\phi\co \overline{H} \to \Aut(H)$ a morphism of Lie groups such that the identity map $\id \co \overline{H} \to H$ is a crossed morphism relative to $\phi$. The latter fact entails that $\delta(\text{id}) = \id^\ast\theta \in \Omega^1(\overline{H},\h)^{\overline{H}}$, see Theorem \ref{th: crossed morph gp G-inv}. Note that $(\text{id}^\ast\theta)_h=\theta_h$, for all $h\in \overline{H}$.

Using the isomorphism between $\Omega^\bullet(\overline{H},\h)^{\overline{H}}$ and $C^\bullet(\overline{\h},\mathfrak h)$ induced by the evaluation map at $e_{\overline{H}}=e_{H}$, see Proposition \ref{prop: de Rham CE iso of Lie} in the Appendix, one obtains the MC-element $ev_{e_H}(\text{id}^\ast\theta)\in C^{1}(\overline{\h},\mathfrak h)$. More explicitly,
\[
	ev_{e_H}(\text{id}^\ast\theta)
	=(\text{id}^\ast\theta)_{e_H}=\theta_{e_H}
	=\text{id}\in C^1(\overline{\h},\mathfrak h).
\]

Since $\id\in C^1(\overline{\h},\mathfrak h)$ is a MC-element, for all $x,y\in\overline{\h}$
\begin{eqnarray*}
	0
	&=&\big(d_{CE}\id +\frac{1}{2}\{\id,\id\}\big)(x,y) 
			= \alpha_x(\id(y)) - \alpha_{y}(\id(x)) - \id([x,y]_{\overline{\h}}) 
			+ [\id(x),\id(y)]_{\h}\\
	&=&\alpha_x(y)-\alpha_y(x)-\llbracket x,y\rrbracket+[x,y],		.
\end{eqnarray*}
that is
\begin{equation}
	\llbracket x,y\rrbracket=\alpha_x(y)-\alpha_y(x)+[x,y],\label{eq:la}
\end{equation}
where, for every $x$, $\alpha_x\in\text{Der}(\h)$ is defined as in \eqref{eq:derx}. The previous computation says that the Lie bracket of the Lie algebra $\overline{\h}$ of $\overline{H}$ is as in \eqref{eq:la} and the comments enclosed in the proof of Lemma \ref{lem:crosMC} entails that $\alpha:\overline{\h}\rightarrow\text{Der}(\h)$ is a Lie morphism. Applying now the result in Proposition \ref{prop:pl} one can conclude with the following

\begin{proposition}\label{prop:postgroups}
Let $H$ be a connected Lie group and $\phi:H \rightarrow \Aut(H)$ a smooth map such that \eqref{eq:pg2} holds true. Then $\h$ carries a structure of a post-Lie algebra whose post-Lie product is defined by the formula $x\triangleright y=\alpha_x(y)$ where $\alpha_x$ is as in Formula \eqref{eq:derx}.
\end{proposition}

\begin{remark}
Lemma \ref{prop:postg} and the discussion following it, relates the approach to post-Lie algebras via crossed morphism presented in this note (following \cite{MQ-Crossed}), with the notion of post-(Lie) group, see for example \cite{Bai-Guo-Sheng-Tang-post-groups}. In particular, Proposition \ref{prop:postgroups} gives an alternative proof of the result contained in the latter reference stating that the Lie algebra of a post-(Lie) group carries a structure of a post-Lie algebra. In particular, the $\star$-product defined in \eqref{eq:pg1}, when associative, is (an instance of) the so-called \emph{Grossman--Larson product}, see \cite{Grossman-Larson,Munthe-Kaas-Lundervold-PL, MKW} and Subsection \ref{ss:GL} here below. \end{remark}


\section{Further examples of post-Lie algebras}
\label{sec:exainf}

In this section we will present two examples of post-Lie algebras having a differential-geometric origin. These examples will not play any role in the sequel of the paper and, for this reason, the reader more algebraically inclined can safely skip this section and move to the next one.

\begin{example}[Infinite dimensional post-Lie algebra, see \cite{B-N},\cite{Munthe-Kaas-Lundervold-PL}]
\label{ex:crossedLiealg}
Let $K$ be a Lie group and let $\mathfrak k$ be its Lie algebra. Consider $\g=\mathfrak X(K)$ the Lie algebra of the vector fields on $K$ and $\h=C^\infty(K,\mathfrak k)$, endowed with the Lie bracket defined by 
\begin{equation}
	\llceil \xi,\eta\rrceil(k)=[\xi(k),\eta(k)]_\mathfrak k,\label{eq:bracketbundle}
\end{equation}
for all $\xi,\eta\in\h$ and $k\in K$. 
First note that $\mathfrak h=\Omega^0(K,\mathfrak k)$, see Section \ref{sec:diffformV}, and that it can be identified with $C^\infty(K)\otimes\mathfrak k$. More precisely, if $\{e_i\}_{i=1,\ldots,\text{dim} K}$ is a basis of $\mathfrak k$, every $\xi\in\mathfrak h$ can be written as $\xi=\sum_{i=1}^{\text{dim} K}f_i \otimes e_i$, for some $f_i\in C^\infty(K)$. Under this identification, for all $X\in\mathfrak g$, one has $X(\xi)=\sum_{i}(Xf_i)\otimes e_i$, i.e.
\begin{equation}
	(X\xi)(k)=X_k\xi=\sum_{i}X_k(f_i) e_i,\qquad \forall k\in K.\label{eq:derAlge}
\end{equation}
Furthermore, observe $\theta\in\Omega^1(K,\mathfrak k)$, the left-invariant Maurer--Cartan form of $K$, see Example \ref{ex:MCform}, induces the application $\phi\co \g \to  \mathfrak h$, defined for all $X\in\g$ by
\begin{equation}
	\phi(X)=i_X\theta, \label{eq:phi}
\end{equation}
where $i_X\theta$ is the element in $\mathfrak h$ such that for all $k\in K$
\[
	i_X\theta(k)=\theta_k(X(k))=(L_{k^{-1}})_{\ast,k}(X(k)).
\] 
Computing $i_{[X,Y]}\theta=\mathcal L_X(i_Y\theta)-i_Y(\mathcal L_X\theta)$, where $\mathcal L_X$ denotes the operation of Lie derivative along $X\in\mathfrak X(K)$, since $d\theta+\frac{1}{2}[\theta,\theta]=0$, one obtains for all $X,Y\in\g$
\[
	i_{[X,Y]}\theta=i_X(di_Y\theta)-i_Y(di_X\theta)+\llceil i_X\theta,i_Y\theta\rrceil,
\]
i.e.~$\phi\co \g\to \h$ is an invertible crossed morphism, where $\upsilon:\mathfrak g\rightarrow\text{Der}(\mathfrak h)$ is the Lie algebra morphism defined for all $X\in\g$, $\xi\in\h$ by
\begin{equation}
	\upsilon_X(\xi)(k):=X_k\xi, \label{eq:up}
\end{equation}
where the right-hand-side was defined in \eqref{eq:derAlge}.

$\phi$ is an invertible $C^\infty(K)$-linear map, such that, for all $x\in\mathfrak k$, $\phi(X_x)$ is the $\mathfrak k$-valued constant function on K equal to $x$, where $X_x$ is the left-invariant vector field defined by the element $x\in\mathfrak k$.
\end{example}

The previous example, which considers a crossed morphism between infinite dimensional Lie algebras, can be recast in the framework of the theory of Lie algebroids as observed in the following comments.


\subsubsection{Action Lie algebroids and bundle of Lie algebras}
\label{sssec:algebroid}

Recall that a Lie algebroid over a manifold $M$ is triple $(E,[-,-],a)$ of a vector bundle $E\stackrel{\pi}{\rightarrow} M$ whose space of global section $\Gamma(E)$ is endowed with a $\mathbb R$-bilinear Lie bracket $[-,-]$ and a $C^\infty(M)$-linear map $a:\Gamma(E)\rightarrow\mathfrak X(M)$, called the anchor of the Lie algebroid, such that
\begin{equation}
	[s_1,fs_2]=f[s_1,s_2]+a(s_1)(f)s_2,\label{eq:LeibnizAlg}
\end{equation}
for all $s_1,s_2\in\Gamma(E)$ and $f\in C^\infty(M)$. There are many examples of Lie algebroids, among them we recall the tangent bundle of a manifold $M$ and every (finite dimensional) Lie algebra. In the first case $[-,-]$ is the usual bracket between vector fields and the anchor is the identity map of $\Gamma(TM)=\mathfrak X(M)$. In the latter case one can think of the Lie algebra as a vector bundle over a zero dimensional manifold, whose global sections are the elements of the Lie algebra endowed with the Lie bracket of the underlying Lie algebra. 

Another two examples of Lie algebroids, of particular importance for us, are defined starting from a finite dimensional Lie algebra $\mathfrak k$ and a manifold $M$. The first one the Lie algebroid whose vector bundle is $E=M\times\mathfrak k$ and whose anchor is the zero map. Since $E$ is trivial, $\Gamma(E)\simeq C^\infty(M,\mathfrak k)$ and since $a=0$ the Lie bracket is $C^\infty(M)$-bilinear and for all $s_1,s_2\in C^\infty(M,\mathfrak k)$, one has $[s_1,s_2](m)=[s_1(m),s_2(m)]_\mathfrak k$ for all $m\in M$. In this case, $E\stackrel{\pi}{\rightarrow}M$ is a \emph{bundle of Lie algebras}, and this is how such a Lie algebroid is called. To define the second Lie algebroid which is relevant for us, we suppose that $M$ carries an action of $\mathfrak k$, i.e.~it is defined a morphism of Lie algebras $\mathfrak k\stackrel{\rho}{\rightarrow}\mathfrak X(M)$. Now we suppose that, as in the previous case, $E=M\times\mathfrak k$. Furthermore, define $a_{\text{act}}:\Gamma(E)\simeq C^\infty(M,\mathfrak k)\rightarrow\mathfrak X(M)$ as the $C^\infty(M)$-linear map that to every $x\in\mathfrak k$, identified with the constant section $\xi_x\in C^\infty(M,\mathfrak k)$ $\xi_x(m)=x$ for all $m\in M$, associates $X_x=\rho(x)$, the \emph{fundamental vector field} defined by $x$ and define the bracket
$$
	\Gamma(E) \times \Gamma(E) \stackrel{[-,-]_{\text{act}}}{\xrightarrow{\hspace{1cm}}}\Gamma(E)
$$
via the following formula
\begin{equation}
\label{eq:actibracket}
	[\xi,\eta]_{\text{act}}(m)
	=a_{{\text{act}}}(\xi)_m(\eta)-a_{{\text{act}}}(\eta)_m(\xi)+[\xi(m),\eta(m)]_\mathfrak k,\qquad \forall m \in M.
\end{equation}
Since the bracket \eqref{eq:actibracket} satisfies \eqref{eq:LeibnizAlg}, $(E=M\times\mathfrak k,[-,-]_{\text{act}},a_{\text{act}})$ is a Lie algebroid called \emph{action Lie algebroid}. Taking $M=K$, one sees that the \emph{inverse} of application $\phi$ defined in Formula \eqref{eq:phi} is the anchor of the action Lie algebroid defined on the tangent bundle of $K$ (trivialized via left-translations) acted upon by $\mathfrak k$ via the morphism of Lie algebras $\mathfrak k\stackrel{\rho}{\rightarrow}\mathfrak X(K)$ which, to every $x\in\mathfrak k$ associates the corresponding left-invariant vector field $X_x$, where $X_x(k)=(L_k)_{\ast,e}(x)$ for all $k\in K$. Note that on $K\times\mathfrak k\stackrel{\pi}{\rightarrow}\mathfrak k$, together with the action Lie algebroid defined by the left-translations, it is also defined a bundle of Lie algebras, see Formula \eqref{eq:bracketbundle}.

Before moving to the next topic we want to link what was discussed in Example \ref{ex:crossedLiealg} and in the comments coming after that to the results in \cite{M-KSV}, where the authors discuss a relation between post-Lie algebras and Lie algebroids.
	
\begin{example} 
In what follows we will use the notations introduced in Example \ref{ex:crossedLiealg}. In particular, $K$ is a Lie group, $\mathfrak k=\text{Lie}(K)$,  $\g=\mathfrak X(K)$ and $\h=C^\infty(K,\mathfrak k)$ endowed with the Lie bracket $\llceil-,-\rrceil$ defined in \eqref{eq:bracketbundle}. Under these assumptions one can introduce on $\h$ a bilinear product $\triangleright$ via the following formula
\begin{equation}
	f\triangleright g:=\phi^{-1}(f)g,\qquad \forall f,g\in\h,\label{eq:Kpl}
\end{equation}
which which makes $(\h, \llceil-,-\rrceil, \triangleright)$ into a post-Lie algebra whose associated Lie algebra $\overline\h$ has as underlying vector  space $C^\infty(K,\mathfrak k)$ and whose Lie bracket is 
\[
	\llfloor f,g\rrfloor=\phi^{-1}(f)g-\phi^{-1}(g)f+\llceil f,g\rrceil,\qquad \forall f,g\in C^\infty(K,\mathfrak k).
\]
In other words, the post-Lie algebra $(\h, \llceil-,-\rrceil, \triangleright)$ introduced above, corresponds to the canonical action Lie algebroid defined on $K$ acting on itself by left translations and it represents an instance of the theory contained in \cite{M-KSV}.
Pulling back \eqref{eq:Kpl} to $\g$, one obtains a $\mathbb K$-linear application $\triangledown:\g\otimes\g\rightarrow\g$, defined by
\begin{equation}
\triangledown(X\otimes Y)=\triangledown_X Y=\phi^{-1}(\phi(X)\triangleright\phi(Y)),\qquad \forall X,Y\in\g,
\end{equation}
which is $C^\infty(K)$-linear with respect to the first entry and such that
\[
	\triangledown_X(\xi Y)
	=X(\xi)Y+\xi \triangledown_X Y,\,\forall\xi\in C^\infty(K),\qquad X,Y\in\mathfrak X(K).
\]
This last identity, together with \eqref{eq:Kpl}, implies that $\triangledown_X Y=0$ for all $X\in\mathfrak X(K)$ and all $Y$ left-invariant. In other words, $\triangledown$ defines a flat linear connection on $TK$, whose flat sections are the left-invariant vector fields, and whose torsion is easily shown to be parallel since $T(X_x,X_y)=-X_{[x,y]_\mathfrak t}$, for all $x,y\in\mathfrak k$. In this way one recovers the post-Lie algebra on $\mathfrak X(K)$ defined by the \emph{Cartan connection}, see for example Corollary 4 in \cite{EFMpostLiealgebra-factorization}.
\end{example}

\begin{remark}
Following the suggestions contained in \cite{M-KSV}, it could be interesting to investigate the possibility to relate the approach to post-Lie algebras via crossed morphisms proposed in this note to homogeneous manifolds more general than Lie groups.
\end{remark}

\begin{example}
In the study of homogeneous manifolds is the following remarkable example of a post-Lie algebra, which was recently found by Grong, Munthe-Kaas and Stava in \cite{GMKS}. Before explaining its geometric origin and motivating its study, let us state it in algebraic terms. 	
	
Let $\mathfrak{h}=(V,[-,-])$ be a Lie algebra and let $\mathfrak{e}:=\End(V)$ be the Lie algebra of the linear endomorphisms of $V$ and $\mathfrak{d}$ its sub-Lie algebra of derivations of $V$. Let $\alpha\co \h \to  \mathfrak{d}$ be a mere linear map. 
Set $x \pl_\h y := \alpha_x(y)$ for all $x,y\in V$. 
Set $T\co \h \ot \h \to \h$ 
\begin{equation*}
	T(x,y) = x\pl y - y\pl x -[x,y]_\h
\end{equation*}
be the \emph{torsion}, which measures for $[-,-]_\h$ its lack of being the anti-symmetrization of $\pl$, let and $R\co \h\ot \h \to \e$ be 
\begin{equation*}
	R(x,y) = \alpha_x\circ \alpha_y - \alpha_y\circ \alpha_x  - \alpha_{[x,y]_\h}
\end{equation*}
be the \emph{curvature}, which measures for $\alpha$ its lack of being a morphism of Lie algebras. For any map $F\co \h^{\ot n} \to \h$ (or $F\co \h^{\ot n} \to \e$) one has 	
\begin{equation*}
	(\alpha_xF)(x_1,\ldots,x_n) = \alpha_x(F(x_1,\ldots,x_n)) - \sum_{i=1}^{n} F(x_1,,\ldots, \alpha_x(x_i),\ldots,x_n). 
\end{equation*} 
We write $\alpha F=0$ if one has $\alpha_xF =0$ for all $x\in \h$. 
	
Consider the space $\h\oplus \e$ and define the following operations $[-,-]$ and $\pl$: for all $x,y\in \h$ and $E,F\in \e$, let:
\begin{align*}
		x\pl y &=  x\pl_\h y		&& [x,y] = -T(x,y) + R(x,y)\\
		x\pl E &= \alpha_xE			&& [x,E] = E(x) \\
		E\pl x &= E(x)				&& \\
		E\pl F &= [E,F]_{\e}			&& [E,F] = -[E,F]_{\e}.
\end{align*}

\begin{theorem}[\cite{GMKS}]
If $\alpha T = 0$ and $\alpha R = 0$, then the triple $(\h\oplus \e,[-,-],\pl)$ is a post-Lie algebra.  
\end{theorem}	

The verification of this fact is straightforward but tedious. Note that the Jacobi identity for the bracket follows from $\alpha T = 0$ and $\alpha R = 0$.	This may be easily seen if one remarks that $T$ and $R$ satisfy the \emph{Bianchi identities}: 
\begin{align*}
	\sum_{\circlearrowright}T(T(x,y),z) + (\alpha_xT)(y,z) - R(x,y)(z) = 0 \\
	\sum_{\circlearrowright} (\alpha_xR)(y,z) +  R(T(x,y),z) = 0,
\end{align*}
where $\sum_{\circlearrowright}$ denotes the sum over the cyclic permutations of $(x,y,z)$. 
	
The example from \cite{GMKS} is as follows. Consider a connected manifold $M$ equipped with an affine connection $\nabla$ (which, in our example, is $\alpha$). It is known since \cite{Munthe-Kaas-Lundervold-PL} that if $\nabla$ is flat and with a constant torsion (\ie $R=0$ and $\nabla T =0$), then the space of vector fields on $M$, $\Gamma(TM)$, is a post-Lie algebra; its Lie bracket is the usual Jacobi bracket and the post-Lie product is given by $x\pl y = \nabla_x(y)$. Such a result allows us to perform the Runge--Kutta--Munthe-Kaas integration methods. 
	
A natural question to ask is whether such a theory is available for more general manifolds. The recent example obtained by Grong, Munthe-Kaas and Stava shows that manifolds equipped with an affine connection that has both a constant curvature and torsion (\ie $\nabla R= 0$ and $\nabla T=0$) enter this framework. The main ingredient consists in considering, along with the Lie algebra  $\h=\Gamma(TM)$, the infinitesimal transformations of the holonomy group, say $\mathfrak{hol}\subset \e$.  The proof is based on the study of the framed bundle $FM$ of $M$. It has a natural flat connection, let us denote it by $\tilde{\nabla}$. On the sub-bundle $N\subset FM$ of the frames that are obtained by parallel transport of some original one, the connection $\tilde{\nabla}$ turns out to have constant torsion if and only if the connection of $M$ is such that $\nabla R= 0$ and $\nabla T=0$. Therefore, $\Gamma(TN)$ is a post-Lie algebra. The post-Lie algebra $\h\oplus \mathfrak{hol}$ is obtained as being the natural decomposition of a sub-post-Lie algebra of $\Gamma(TN)$. 
\end{example}


\section{On the universal enveloping algebras of post-Lie algebras}
\label{sec:unipo}


\subsection{Free pre-Lie and free post-Lie algebras}
\label{ssec:freeprepostlie}
	
Free pre- and post-Lie algebras are well-understood in terms of non-planar respectively planar rooted trees. See for instance \cite{ChapotonLivernet2000,MQS,Munthe-Kaas-Lundervold-PL,Guin-Oudom}. We will review an operadic approach, which provides a convenient framework to understand how (free) pre- and post-Lie algebras relate with other structures such as associated universal enveloping algebras.

	
\subsubsection{Free pre-Lie algebras}
\label{sssec:freeprelie}
	
We briefly review the operadic model $\SB$ of the operad $\ca{P}re\ca{L}ie$ that controls pre-Lie algebras in terms of (non-planar) rooted trees. See Chapoton and Livernet \cite{ChapotonLivernet2000} for more details. 
For each $n\geq1$, the vector space $\SB(n)$ is generated by the set of pairs $(T,f)$, where: 
	\begin{itemize}
		\item $T$ is an isomorphism class of non-planar rooted trees with $n$ vertices;
		\item $f\co \{1,\ldots,n\} \to \text{Vert}(T)$ is a bijection between $\{1,\ldots,n\}$ and the set of vertices of $T$. 
	\end{itemize} 
It is naturally endowed with an action of the symmetric group $\Sigma_{n}$. Since we consider non-planar rooted trees, the action is not free: 
	\begin{equation*}
		\begin{tikzpicture}
			[baseline, my circle/.style={draw, fill, circle, minimum size=3pt, inner sep=0pt}, level distance=0.4cm, 
			level 2/.style={sibling distance=0.6cm}, sibling distance=0.6cm,baseline=1ex]
			\node [prelie,label=left:\tiny{$3$}] {} [grow=up]
			{
				child {node [prelie,label=left:\tiny{$2$}]  {}} 
				child {node [prelie,label=left:\tiny{$1$}] {}} 
			};
		\end{tikzpicture}
		\cdot (2,1,3) 
		=
		\begin{tikzpicture}
			[baseline, my circle/.style={draw, fill, circle, minimum size=3pt, inner sep=0pt}, level distance=0.4cm, 
			level 2/.style={sibling distance=0.6cm}, sibling distance=0.6cm,baseline=1ex]
			\node [prelie,label=left:\tiny{$3$}] {} [grow=up]
			{
				child {node [prelie,label=left:\tiny{$1$}]  {}} 
				child {node [prelie,label=left:\tiny{$2$}] {}} 
			};
		\end{tikzpicture}
		=
		\begin{tikzpicture}
			[baseline, my circle/.style={draw, fill, circle, minimum size=3pt, inner sep=0pt}, level distance=0.4cm, 
			level 2/.style={sibling distance=0.6cm}, sibling distance=0.6cm,baseline=1ex]
			\node [prelie,label=left:\tiny{$3$}] {} [grow=up]
			{
				child {node [prelie,label=left:\tiny{$2$}]  {}} 
				child {node [prelie,label=left:\tiny{$1$}] {}} 
			};
		\end{tikzpicture}
		\in \SB(3).
	\end{equation*}
	The operadic structure is as follows. 	For any two trees $\tau_1 \in \SB(m)$ and $\tau_2 \in \SB(n)$, 
	let $v$ be the vertex of $\tau_1$ that is labeled by $i$; let $k$ be the number of its incoming edges. 
	For a map $\phi \co \{1,\ldots,k\} \to V(\tau_2)$, where $V(\tau)$ denotes the vertex set of the tree $\tau$, let $\tau_1 \circ_i^{\phi} \tau_2$ be the tree obtained by 
	substituting the vertex labeled by $i$ by the tree $\tau_2$, and then grafting the incoming edges of the vertex $i$ 
	to the labeled vertices of $\tau_2$ following the map $\phi$. The labelling of $\tau_1 \circ_i^{\phi} \tau_2$ 
	is given by classical re-indexing.  
	
	The partial composition of $\tau_1$ and $\tau_2$ at vertex $i$ is:  
	\begin{equation}\label{eq: explicit partial compo partial planar tree}
		\tau_1 \circ_i \tau_2 = \sum_{\phi} \tau_1 \circ_i^{\phi} \tau_2,
	\end{equation} 
	where $\phi$ runs through the set of maps from $\{1,\ldots,k\}$ to $V(\tau_2)$. 
	For instance: 
	\begin{equation*}
		\begin{tikzpicture}
			[baseline, my circle/.style={draw, fill, circle, minimum size=3pt, inner sep=0pt}, level distance=0.4cm, 
			level 2/.style={sibling distance=0.6cm}, sibling distance=0.6cm,baseline=1ex]
			\node [prelie,label=left:\tiny{$1$}] {} [grow=up]
			{
				child {node [prelie,label=left:\tiny{$3$}]  {}} 
				child {node [prelie,label=left:\tiny{$2$}] {}} 
			};
		\end{tikzpicture}
		\circ_1 
		\begin{tikzpicture}
			[baseline, my circle/.style={draw, fill, circle, minimum size=3pt, inner sep=0pt}, level distance=0.4cm, 
			level 2/.style={sibling distance=0.6cm}, sibling distance=0.6cm,baseline=1ex]
			\node [prelie,label=left:\tiny{$1$}] {} [grow=up]
			{
				child {node [prelie,label=left:\tiny{$2$}]  {}} 
			};
		\end{tikzpicture}
		=
		\begin{tikzpicture}
			[baseline, my circle/.style={draw, fill, circle, minimum size=3pt, inner sep=0pt}, level distance=0.4cm, 
			level 2/.style={sibling distance=0.6cm}, sibling distance=0.6cm,baseline=1ex]
			\node [prelie,label=left:\tiny{$1$}] {} [grow=up]
			{
				child {node [prelie,label=left:\tiny{$4$}]  {}} 
				child {node [prelie,label=left:\tiny{$3$}] {}} 
				child {node [prelie,label=left:\tiny{$2$}] {}}
			};
		\end{tikzpicture}
		+
		\begin{tikzpicture}
			[baseline, my circle/.style={draw, fill, circle, minimum size=3pt, inner sep=0pt}, level distance=0.4cm, 
			level 2/.style={sibling distance=0.6cm}, sibling distance=0.6cm,baseline=1ex]
			\node [prelie,label=left:\tiny{$1$}] {} [grow=up]
			{
				child {node [prelie,label=left:\tiny{$2$}]  {} 
					child {node [prelie,label=left:\tiny{$4$}] {}} 
				}
				child {node [prelie,label=left:\tiny{$3$}] {} }
			};
		\end{tikzpicture}
		+
		\begin{tikzpicture}
			[baseline, my circle/.style={draw, fill, circle, minimum size=3pt, inner sep=0pt}, level distance=0.4cm, 
			level 2/.style={sibling distance=0.6cm}, sibling distance=0.6cm,baseline=1ex]
			\node [prelie,label=left:\tiny{$1$}] {} [grow=up]
			{
				child {node [prelie,label=left:\tiny{$2$}]  {} 
					child {node [prelie,label=left:\tiny{$3$}] {}} 
				}
				child {node [prelie,label=left:\tiny{$4$}] {} }
			};
		\end{tikzpicture}
		+
		\begin{tikzpicture}
			[baseline, my circle/.style={draw, fill, circle, minimum size=3pt, inner sep=0pt}, level distance=0.4cm, 
			level 2/.style={sibling distance=0.6cm}, sibling distance=0.6cm,baseline=1ex]
			\node [prelie,label=left:\tiny{$1$}] {} [grow=up]
			{
				child {node [prelie,label=left:\tiny{$2$}]  {} 
					child {node [prelie,label=left:\tiny{$4$}] {}} 
					child {node [prelie,label=left:\tiny{$3$}] {}}
				}
			};
		\end{tikzpicture}
		.
	\end{equation*}
	
	\begin{theorem}[\cite{ChapotonLivernet2000}]
		The operad $\SB$ is isomorphic to $\ca{P}re\ca{L}ie$.  
	\end{theorem}
	
	For instance one observes that: 
		\begin{equation*}
		\begin{tikzpicture}
			[  level distance=0.4cm, level 2/.style={sibling distance=0.6cm}, sibling distance=0.6cm,baseline=1ex,level 1/.style={level distance=0.4cm}]
			\node [] {} [grow'=up]
			{	
				{node [prelie, label=left:\small{$3$}]  {}
					child {node [prelie,label=left:\small{$1$}]  {}} 
					child {node [prelie,label=left:\small{$2$}]  {}} 
				}
			};
		\end{tikzpicture}  
		= 
		\begin{tikzpicture}
			[  level distance=0.4cm, level 2/.style={sibling distance=0.6cm}, sibling distance=0.6cm,baseline=1ex,level 1/.style={level distance=0.4cm}]
			\node [] {} [grow'=up]
			{	
				{node [prelie, label=left:\small{$2$}]  {}
					child {node [prelie,label=left:\small{$1$}]  {}} 
				}
			};
		\end{tikzpicture}  
		\circ_2 
		\begin{tikzpicture}
			[  level distance=0.4cm, level 2/.style={sibling distance=0.6cm}, sibling distance=0.6cm,baseline=1ex,level 1/.style={level distance=0.4cm}]
			\node [] {} [grow'=up]
			{	
				{node [prelie, label=left:\small{$2$}]  {}
					child {node [prelie,label=left:\small{$1$}]  {}} 
				}
			};
		\end{tikzpicture}  
		- 
		\begin{tikzpicture}
			[  level distance=0.4cm, level 2/.style={sibling distance=0.6cm}, sibling distance=0.6cm,baseline=1ex,level 1/.style={level distance=0.4cm}]
			\node [] {} [grow'=up]
			{	
				{node [prelie, label=left:\small{$2$}]  {}
					child {node [prelie,label=left:\small{$1$}]  {}} 
				}
			};
		\end{tikzpicture}  
		\circ_1 
		\begin{tikzpicture}
			[  level distance=0.4cm, level 2/.style={sibling distance=0.6cm}, sibling distance=0.6cm,baseline=1ex,level 1/.style={level distance=0.4cm}]
			\node [] {} [grow'=up]
			{	
				{node [prelie, label=left:\small{$2$}]  {}
					child {node [prelie,label=left:\small{$1$}]  {}} 
				}
			};
		\end{tikzpicture}  .
	\end{equation*}
	In particular, if $(V,\pl)$ is a pre-Lie algebra, one has  
		\begin{equation*}
		\begin{tikzpicture}
			[  level distance=0.4cm, level 2/.style={sibling distance=0.6cm}, sibling distance=0.6cm,baseline=1ex,level 1/.style={level distance=0.4cm}]
			\node [] {} [grow'=up]
			{	
				{node [prelie, label=left:\small{$2$}]  {}
					child {node [prelie,label=left:\small{$1$}]  {}} 
				}
			};
		\end{tikzpicture}  
		\left( x \ot y \right)
	=	x\pl y
	~~~\text{ and }~~~
		\begin{tikzpicture}
			[  level distance=0.4cm, level 2/.style={sibling distance=0.6cm}, sibling distance=0.6cm,baseline=1ex,level 1/.style={level distance=0.4cm}]
			\node [] {} [grow'=up]
			{	
				{node [prelie, label=left:\small{$3$}]  {}
					child {node [prelie,label=left:\small{$1$}]  {}} 
					child {node [prelie,label=left:\small{$2$}]  {}} 
				}
			};
		\end{tikzpicture}  
		\left( x \ot  y \ot z \right)
	= x\pl(y\pl z) - (x\pl y)\pl z.
	\end{equation*}
	That the corolla is non-planar corresponds to the symmetry in the first two variables of the associator 
	$$ 
		x\pl(y\pl z) - (x\pl y)\pl z =  y\pl(x\pl z) - (y\pl x)\pl z.
	$$
	
\begin{remark}
	For later use, observe that corollas 
	\begin{equation}\label{eq: braces}
		\begin{tikzpicture}
			[  level distance=0.4cm, level 2/.style={sibling distance=0.6cm}, sibling distance=0.6cm,baseline=2.5ex,level 1/.style={level distance=0.4cm}]
			\node [] {} [grow'=up]
			{	
				{node [prelie, label=right:\small{$n+1$}]  {}
					child {node [prelie,label=above:\small{$1$}]  {}} 
					child {node [prelie,label=above:\small{$2$}]  {}} 
					child {node [label=above:\small{$\dots$}]        {}} 
					child {node [prelie,label=above:\small{$n$}]  {}} 
				}
			};
		\end{tikzpicture}  
	\end{equation}
	define operations that are symmetric in the first $n$ variables. 
\end{remark}

From the above theorem, we can deduce a description of the free pre-Lie algebra on one generator. 
Recall that, given an operad $\mathcal{O}$ and a vector space $V$, the free $\mathcal{O}$--algebra generated by $V$  
is explicitly given by $\mathcal{O}(V):= \bigoplus_{n\geq 0} \mathcal{O}(n)\ot_{\mathbb{S}_n} V^{\ot n}$. 
Here: $\ca{P}re\ca{L}ie(\Bbbk) = \SB(\Bbbk)$ is therefore generated by the set
\begin{equation*}\label{eq: set gen of free prelie}
	\ca{G}_{pre} = 
	\bigg\{
	\begin{tikzpicture}
		[  level distance=0.5cm, level 2/.style={sibling distance=0.6cm}, sibling distance=0.6cm,baseline=1ex,level 1/.style={level distance=0.4cm}]
		\node [] {} [grow=up]
		{	
			{node [prelie]  {}
			}
		};
	\end{tikzpicture} 
	,
	\begin{tikzpicture}
		[  level distance=0.5cm, level 2/.style={sibling distance=0.6cm}, sibling distance=0.6cm,baseline=1ex,level 1/.style={level distance=0.4cm}]
		\node [] {} [grow=up]
		{	
			{node [prelie]  {}
				child {node [prelie]  {} 
				}
			}
		};
	\end{tikzpicture} 
	,
	\begin{tikzpicture}
		[  level distance=0.5cm, level 2/.style={sibling distance=0.6cm}, sibling distance=0.6cm,baseline=1ex,level 1/.style={level distance=0.4cm}]
		\node [] {} [grow=up]
		{	
			{node [prelie]  {}
				child {node [prelie]  {} }
				child {node [prelie]  {} 
				}
			}
		};
	\end{tikzpicture} 
	,
	\begin{tikzpicture}
		[  level distance=0.5cm, level 2/.style={sibling distance=0.6cm}, sibling distance=0.6cm,baseline=1ex,level 1/.style={level distance=0.4cm}]
		\node [] {} [grow=up]
		{	
			{node [prelie]  {}
				child {node [prelie]  {} 
					child {node [prelie]  {} 
				}}
			}
		};
	\end{tikzpicture} 
	,
	\begin{tikzpicture}
		[  level distance=0.5cm, level 2/.style={sibling distance=0.6cm}, sibling distance=0.6cm,baseline=1ex,level 1/.style={level distance=0.4cm}]
		\node [] {} [grow=up]
		{	
			{node [prelie]  {}
				child {node [prelie]  {} }
				child {node [prelie]  {} 
					child {node [prelie]  {} }
				}
			}
		};
	\end{tikzpicture} 
	,
	\begin{tikzpicture}
		[  level distance=0.5cm, level 2/.style={sibling distance=0.6cm}, sibling distance=0.6cm,baseline=1ex,level 1/.style={level distance=0.4cm}]
		\node [] {} [grow=up]
		{	
			{node [prelie]  {}
				child {node [prelie]  {} }
				child {node [prelie]  {} }
				child {node [prelie]  {} }
			}
		};
	\end{tikzpicture} 
	, 
	\begin{tikzpicture}
		[  level distance=0.5cm, level 2/.style={sibling distance=0.6cm}, sibling distance=0.6cm,baseline=1ex,level 1/.style={level distance=0.4cm}]
		\node [] {} [grow=up]
		{	
			{node [prelie]  {}
				child {node [prelie]  {} 
					child {node [prelie]  {} 
						child {node [prelie]  {} 
				}}}
			}
		};
	\end{tikzpicture} 
	,\, \ldots \bigg\}.
\end{equation*}

Its structure of pre-Lie algebra is given as follows. For any two trees $\tau_1$ and $\tau_2$ in $\ca{G}_{pre}$, their pre-Lie product 
\begin{equation*}
	\tau_1 \pl \tau_2  
\end{equation*}
is obtained by considering the sum of all graftings of the root of $\tau_1$ to the vertices of $\tau_2$, joining them by an new edge. 
For instance:
\begin{equation*}
		\begin{tikzpicture}
		[  level distance=0.5cm, level 2/.style={sibling distance=0.6cm}, sibling distance=0.6cm,baseline=0ex,level 1/.style={level distance=0.4cm}]
		\node [] {} [grow=up]
		{	
			{node [prelie]  {}
			}
		};
	\end{tikzpicture} 
	\pl 
	\begin{tikzpicture}
		[  level distance=0.5cm, level 2/.style={sibling distance=0.6cm}, sibling distance=0.6cm,baseline=1ex,level 1/.style={level distance=0.4cm}]
		\node [] {} [grow=up]
		{	
			{node [prelie]  {}
				child {node [prelie]  {} 
				}
			}
		};
	\end{tikzpicture} 
	=
		\begin{tikzpicture}
		[  level distance=0.5cm, level 2/.style={sibling distance=0.6cm}, sibling distance=0.6cm,baseline=1ex,level 1/.style={level distance=0.4cm}]
		\node [] {} [grow=up]
		{	
			{node [prelie]  {}
				child {node [prelie]  {} }
				child {node [prelie]  {} 
				}
			}
		};
	\end{tikzpicture} 
	+
	\begin{tikzpicture}
	[  level distance=0.5cm, level 2/.style={sibling distance=0.6cm}, sibling distance=0.6cm,baseline=1ex,level 1/.style={level distance=0.4cm}]
	\node [] {} [grow=up]
	{	
		{node [prelie]  {}
			child {node [prelie]  {} 
				child {node [prelie]  {} 
			}}
		}
	};
\end{tikzpicture} .
\end{equation*}


\subsubsection{Free post-Lie algebras}
\label{sssec:freepostlie}

The structure of post-Lie algebra was investigated in \cite{Munthe-Kaas-Lundervold-PL} where, in particular, a description of the free post-Lie algebra on one generator was given in terms of planar rooted trees. In a similar way as for $\ca{P}re\ca{L}ie$, there also exists a combinatorial model, $\PSB$, for the operad $\ca{P}ost\ca{L}ie$ that controls post-Lie algebras; see \cite{MQS}. While $\SB$ is based on non-planar rooted trees, the operad $\PSB$ is based on planar ones. 
	
Here, we briefly review the construction of $\PSB$. However, instead of giving a formal definition of $\PSB$, which is more involved than that of $\SB$, will we simply emphasize a few notable differences between $\PSB$ and $\SB$. We refer to \cite{MQS} for details. 
	
Note that in a post-Lie algebra $\h=(V,[-,-],\pl)$ the Lie bracket is an obstruction to the commutativity of the associator in the first two variables: 
	\begin{equation*}
		[x,y]\triangleright z={\rm a}_{\triangleright}(x,y,z)-{\rm a}_{\triangleright}(y,x,z). 
	\end{equation*}
	Therefore, instead of considering non-planar rooted trees as in $\SB$, one is led to consider planar ones.
	Moreover, since one has to implement the Lie bracket, it is also convenient to consider trees with a different type of vertices (unlabelled ones) and declare them to be anti-symmetric and satisfying the Jacobi relation. 
	Doing this, we end up with trees with two types of vertices: labeled and unlabelled. 
	The resulting ``trees" are in part planar (around labeled vertices) and in part anti-symmetric (around unlabelled vertices).

		If $\h=(V,[-,-],\pl)$ is a post-Lie algebra, one has  
	\begin{equation*}
		\begin{tikzpicture}
			[  level distance=0.4cm, level 2/.style={sibling distance=0.6cm}, sibling distance=0.6cm,baseline=1ex,level 1/.style={level distance=0.4cm}]
			\node [] {} [grow'=up]
			{	
				{node [my circle, label=left:\small{$2$}]  {}
					child {node [my circle,label=left:\small{$1$}]  {}} 
				}
			};
		\end{tikzpicture}  
		\left( x \ot y \right)
		=	x\pl y,
		~~~~~
		\begin{tikzpicture}
			[  level distance=0.5cm, level 2/.style={sibling distance=0.6cm}, sibling distance=0.6cm,baseline=1ex,level 1/.style={level distance=0.4cm}]
			\node [] {} [grow=up]
			{	
				{node [whitesq]  {}
					child {node [my circle,label=above:\small{$2$}]  {} }
					child {node [my circle,label=above:\small{$1$}]  {} 
					}
				}
			};
		\end{tikzpicture} 
		(x\ot y) = 
		[x,y]
		~~~\text{ and }~~~
		\begin{tikzpicture}
			[  level distance=0.4cm, level 2/.style={sibling distance=0.6cm}, sibling distance=0.6cm,baseline=1ex,level 1/.style={level distance=0.4cm}]
			\node [] {} [grow'=up]
			{	
				{node [my circle, label=right:\small{$3$}]  {}
					child {node [my circle,label=above:\small{$1$}]  {}} 
					child {node [my circle,label=above:\small{$2$}]  {}} 
				}
			};
		\end{tikzpicture}  
		\left( x \ot  y \ot z \right)
		= x\pl(y\pl z) - (x\pl y)\pl z.
	\end{equation*}
\begin{remark}
Here, the second tree has an unlabelled vertex as root, which stands for the Lie bracket. 
Another natural way to implement this bracket would have been by formally writing brackets of trees, like 
\begin{equation*}
\left[ 
\scalebox{.4}{
\begin{forest}
	baseline,for tree={%
		label/.option=content,
		content=,
		circle,
		thick,
		fill,
		minimum size=9pt,
		inner sep=0pt,
		l =.5cm,
		s sep= 8mm,
		grow=north,
		edge ={line width=1pt}
	}
	[
	[] 
	]
\end{forest}
}
	,
	\bullet
	\right] 
	:= 
\scalebox{.4}{
\begin{forest}
	baseline,for tree={%
		label/.option=content,
		content=,
		circle,
		thick,
		fill,
		minimum size=9pt,
		inner sep=0pt,
		l sep=5mm,
		l =.5cm,
		s sep= 8mm,
		grow=north,
		edge ={line width=1pt}
	}
	[ , fill=white, shape=rectangle, draw=black, minimum size=10pt, 
	[] 	[[]]
	]
\end{forest}
}
\end{equation*}
The choice of the "tree"-implementation of the Lie bracket was motivated by an easier description of the operadic structure of $\PSB$, as all elements are (classes) of trees. 
\end{remark}

	Note that the corolla
	$	\begin{tikzpicture}
		[  level distance=0.4cm, level 2/.style={sibling distance=0.6cm}, sibling distance=0.6cm,baseline=1ex,level 1/.style={level distance=0.4cm}]
		\node [] {} [grow'=up]
		{	
			{node [my circle, label=right:\small{$3$}]  {}
				child {node [my circle,label=above:\small{$1$}]  {}} 
				child {node [my circle,label=above:\small{$2$}]  {}} 
			}
		};
	\end{tikzpicture}$
	is a planar tree. The operad $\PSB$ is designed such that: 
	\begin{equation*}
				\begin{tikzpicture}
				[  level distance=0.4cm, level 2/.style={sibling distance=0.6cm}, sibling distance=0.6cm,baseline=1ex,level 1/.style={level distance=0.4cm}]
				\node [] {} [grow'=up]
				{	
					{node [my circle, label=left:\small{$2$}]  {}
						child {node [my circle,label=left:\small{$1$}]  {}} 
					}
				};
			\end{tikzpicture}  
			\circ_1 
			\begin{tikzpicture}
			[  level distance=0.5cm, level 2/.style={sibling distance=0.6cm}, sibling distance=0.6cm,baseline=1ex,level 1/.style={level distance=0.4cm}]
			\node [] {} [grow=up]
			{	
				{node [whitesq]  {}
					child {node [my circle,label=above:\small{$2$}]  {} }
					child {node [my circle,label=above:\small{$1$}]  {} 
					}
				}
			};
		\end{tikzpicture}
		=	\begin{tikzpicture}
			[  level distance=0.4cm, level 2/.style={sibling distance=0.6cm}, sibling distance=0.6cm,baseline=1ex,level 1/.style={level distance=0.4cm}]
			\node [] {} [grow'=up]
			{	
				{node [my circle, label=right:\small{$3$}]  {}
					child {node [my circle,label=above:\small{$1$}]  {}} 
					child {node [my circle,label=above:\small{$2$}]  {}} 
				}
			};
		\end{tikzpicture}  
		-
			\begin{tikzpicture}
			[  level distance=0.4cm, level 2/.style={sibling distance=0.6cm}, sibling distance=0.6cm,baseline=1ex,level 1/.style={level distance=0.4cm}]
			\node [] {} [grow'=up]
			{	
				{node [my circle, label=right:\small{$3$}]  {}
					child {node [my circle,label=above:\small{$2$}]  {}} 
					child {node [my circle,label=above:\small{$1$}]  {}} 
				}
			};
		\end{tikzpicture}  .
	\end{equation*}
	
		As for the pre-Lie case, one has the following result. 
	
	\begin{theorem}\label{th: iso}{\cite[Theorem 11]{MQS}}
		The operad $\PSB$ is isomorphic to $\ca{P}ost\ca{L}ie$. 
	\end{theorem}
	
	By Theorem \ref{th: iso}, one deduces that the free post-Lie algebra on one generator is generated by 
	\begin{equation*}\label{eq: set gen of free postlie}
		\mathcal{G}_{post}
		= \bigg\{
		\begin{tikzpicture}
			[  level distance=0.5cm, level 2/.style={sibling distance=0.6cm}, sibling distance=0.6cm,baseline=1ex,level 1/.style={level distance=0.4cm}]
			\node [] {} [grow=up]
			{	
				 {node [my circle]  {}
				}
			};
		\end{tikzpicture} 
		,
		\begin{tikzpicture}
			[  level distance=0.5cm, level 2/.style={sibling distance=0.6cm}, sibling distance=0.6cm,baseline=1ex,level 1/.style={level distance=0.4cm}]
			\node [] {} [grow=up]
			{	
			 {node [my circle]  {}
					child {node [my circle]  {} 
					}
				}
			};
		\end{tikzpicture} 
		,
		\begin{tikzpicture}
			[  level distance=0.5cm, level 2/.style={sibling distance=0.6cm}, sibling distance=0.6cm,baseline=1ex,level 1/.style={level distance=0.4cm}]
			\node [] {} [grow=up]
			{	
			 {node [my circle]  {}
					child {node [my circle]  {} }
					child {node [my circle]  {} 
					}
				}
			};
		\end{tikzpicture} 
		,
		\begin{tikzpicture}
			[  level distance=0.5cm, level 2/.style={sibling distance=0.6cm}, sibling distance=0.6cm,baseline=1ex,level 1/.style={level distance=0.4cm}]
			\node [] {} [grow=up]
			{	
			 {node [my circle]  {}
					child {node [my circle]  {} 
						child {node [my circle]  {} 
					}}
				}
			};
		\end{tikzpicture} 
		,
		\begin{tikzpicture}
			[  level distance=0.5cm, level 2/.style={sibling distance=0.6cm}, sibling distance=0.6cm,baseline=1ex,level 1/.style={level distance=0.4cm}]
			\node [] {} [grow=up]
			{	
				 {node [my circle]  {}
					child {node [my circle]  {} }
					child {node [my circle]  {} 
						child {node [my circle]  {} }
					}
				}
			};
		\end{tikzpicture} 
		,
		\begin{tikzpicture}
			[  level distance=0.5cm, level 2/.style={sibling distance=0.6cm}, sibling distance=0.6cm,baseline=1ex,level 1/.style={level distance=0.4cm}]
			\node [] {} [grow=up]
			{	
				 {node [my circle]  {}
					child {node [my circle]  {} 
						child {node [my circle]  {} } }
					child {node [my circle]  {} }
				}
			};
		\end{tikzpicture} 
		,
		\begin{tikzpicture}
			[  level distance=0.5cm, level 2/.style={sibling distance=0.6cm}, sibling distance=0.6cm,baseline=1ex,level 1/.style={level distance=0.4cm}]
			\node [] {} [grow=up]
			{	
				 {node [whitesq]  {}
					child {node [my circle]  {} }
					child {node [my circle]  {} 
						child {node [my circle]  {} }
					}
				}
			};
		\end{tikzpicture} 
		,
		\begin{tikzpicture}
			[  level distance=0.5cm, level 2/.style={sibling distance=0.4cm}, sibling distance=0.4cm,baseline=1ex,level 1/.style={level distance=0.4cm}]
			\node [] {} [grow=up]
			{	
				 {node [my circle]  {}
					child {node [my circle]  {} }
					child {node [my circle]  {} }
					child {node [my circle]  {} }
				}
			};
		\end{tikzpicture} 
		, 
		\begin{tikzpicture}
			[  level distance=0.5cm, level 2/.style={sibling distance=0.6cm}, sibling distance=0.6cm,baseline=1ex,level 1/.style={level distance=0.4cm}]
			\node [] {} [grow=up]
			{	
				 {node [my circle]  {}
					child {node [my circle]  {} 
						child {node [my circle]  {} 
							child {node [my circle]  {} 
					}}}
				}
			};
		\end{tikzpicture} 
		,\dots \bigg\}.
	\end{equation*}
	Here, the trees with black round-shaped vertices are planar trees, and the tree displayed with a white square-shaped vertex represents the Lie bracket $[x\pl x,x]$ on the variable $x$.

The post-Lie product 
\begin{equation*}
	\tau_1 \pl \tau_2  
\end{equation*}
of two trees $\tau_1$ and $\tau_2$ in $\ca{G}_{post}$ with only black round-shaped vertices is obtained by considering the sum of all graftings of the root of $\tau_1$ to the \emph{left-side} of the vertices of $\tau_2$, joining them by an new edge. 
	For instance:
	\begin{equation*}
			\begin{tikzpicture}
			[  level distance=0.5cm, level 2/.style={sibling distance=0.6cm}, sibling distance=0.6cm,baseline=0ex,level 1/.style={level distance=0.4cm}]
			\node [] {} [grow=up]
			{	
				{node [my circle]  {}
				}
			};
		\end{tikzpicture} 
		\pl 
		\begin{tikzpicture}
			[  level distance=0.5cm, level 2/.style={sibling distance=0.6cm}, sibling distance=0.6cm,baseline=1ex,level 1/.style={level distance=0.4cm}]
			\node [] {} [grow=up]
			{	
				{node [my circle]  {}
					child {node [my circle]  {} }
					child {node [my circle]  {} 
					}
				}
			};
		\end{tikzpicture} 
		= 
			\begin{tikzpicture}
			[  level distance=0.5cm, level 2/.style={sibling distance=0.6cm}, sibling distance=0.4cm,baseline=1ex,level 1/.style={level distance=0.4cm}]
			\node [] {} [grow=up]
			{	
				{node [my circle]  {}
					child {node [my circle]  {} }
					child {node [my circle]  {} }
					child {node [my circle]  {} }
				}
			};
		\end{tikzpicture} 
		+
		\begin{tikzpicture}
		[  level distance=0.5cm, level 2/.style={sibling distance=0.6cm}, sibling distance=0.6cm,baseline=1ex,level 1/.style={level distance=0.4cm}]
		\node [] {} [grow=up]
		{	
			{node [my circle]  {}
				child {node [my circle]  {} }
				child {node [my circle]  {} 
					child {node [my circle]  {} }
				}
			}
		};
	\end{tikzpicture} 
	+
	\begin{tikzpicture}
		[  level distance=0.5cm, level 2/.style={sibling distance=0.6cm}, sibling distance=0.6cm,baseline=1ex,level 1/.style={level distance=0.4cm}]
		\node [] {} [grow=up]
		{	
			{node [my circle]  {}
				child {node [my circle]  {} 
					child {node [my circle]  {} } }
				child {node [my circle]  {} }
			}
		};
	\end{tikzpicture} 
.
	\end{equation*}

	As for the operad $\SB$, in $\PSB$ one has distinguished elements: the corollas
	\begin{equation}\label{eq: braces PSB}
	\begin{tikzpicture}
		[  level distance=0.4cm, level 2/.style={sibling distance=0.6cm}, sibling distance=0.6cm,baseline=2.5ex,level 1/.style={level distance=0.4cm}]
		\node [] {} [grow'=up]
		{	
			{node [my circle, label=right:\small{$n+1$}]  {}
				child {node [my circle,label=above:\small{$1$}]  {}} 
				child {node [my circle,label=above:\small{$2$}]  {}} 
				child {node [label=above:\small{$\dots$}]        {}} 
				child {node [my circle,label=above:\small{$n$}]  {}} 
			}
		};
	\end{tikzpicture}  .
\end{equation}
However, they  are subject to the following equation:
\begin{equation*}
	\begin{tikzpicture}
		[baseline, my circle/.style={draw, fill, circle, minimum size=3pt, inner sep=0pt}, level distance=0.5cm, 
		level 2/.style={sibling distance=0.8cm}, sibling distance=0.5cm]
		\node [my circle,label=right:\tiny{$n+1$}] {} [grow=up]
		{child {node [my circle,label=above:\tiny{$n$}]  {}} 
			child {node [label=above:\tiny{$\cdots$}]  {}}
			child {node [my circle,label=above:\tiny{$i+1$}]  {}}
			child {node [my circle,label=above:\tiny{$i$}]  {}}
			child {node [label=above:\tiny{$\cdots$}]  {}}
			child {node [my circle,label=above:\tiny{$2$}]  {}}
			child {node [my circle,label=above:\tiny{$1$}]  {}}};
	\end{tikzpicture} - 
	\begin{tikzpicture}
		[baseline, my circle/.style={draw, fill, circle, minimum size=3pt, inner sep=0pt}, level distance=0.5cm, 
		level 2/.style={sibling distance=0.8cm}, sibling distance=0.5cm]
		\node [my circle,label=right:\tiny{$n+1$}] {} [grow=up]
		{child {node [my circle,label=above:\tiny{$n$}]  {}} 
			child {node [label=above:\tiny{$\cdots$}]  {}}
			child {node [my circle,label=above:\tiny{$i$}]  {}}
			child {node [my circle,label=above:\tiny{$i+1$}]  {}}
			child {node [label=above:\tiny{$\cdots$}]  {}}
			child {node [my circle,label=above:\tiny{$2$}]  {}}
			child {node [my circle,label=above:\tiny{$1$}]  {}}};
	\end{tikzpicture} 
	=
		\begin{tikzpicture}
		[baseline, my circle/.style={draw, fill, circle, minimum size=3pt, inner sep=0pt}, level distance=0.5cm, 
		level 2/.style={sibling distance=0.8cm}, sibling distance=0.5cm]
		\node [my circle,label=left:\tiny{$n$}] {} [grow=up]
		{child {node [my circle,label=above:\tiny{$n-1$}]  {}} 
			child {node [label=above:\tiny{$\cdots$}]  {}}
			child {node [my circle,label=above:\tiny{$2$}]  {}}
			child {node [my circle,label=above:\tiny{$1$}]  {}}};
	\end{tikzpicture} 
	\circ_i 
	\begin{tikzpicture} 
		[baseline, my circle/.style={draw, fill, circle, minimum size=3pt, inner sep=0pt}, level distance=0.4cm, 
		level 2/.style={sibling distance=0.8cm}, sibling distance=0.6cm]
		\node [whitesq] {} [grow=up]
		{child {node [my circle,label=above:\tiny{$2$}]  {}}
			child {node [my circle,label=above:\tiny{$1$}]  {}}};
	\end{tikzpicture} .
\end{equation*}


\subsection{On the Grossman--Larson product}
\label{ss:GL}

Let $(\h,\pl)$ be a pre-Lie algebra whose underlying vector space is $V$ and let $\g= (V,\llbracket-,-\rrbracket)$ be the corresponding Lie algebra, see Example \ref{ex:preLie}. There are two enveloping algebras that may be considered: $\ca{U}(\g)$ and $\ca{U}(\h)=S(V)$, the latter being the symmetric algebra of $V$ since $\h$ is abelian. Note that they are in fact bialgebras (even Hopf algebras) when endowed with the unshuffle coproduct. 

In the seminal work of Guin--Oudom \cite{Guin-Oudom}, the authors showed that a structure of pre-Lie algebra on $V$ corresponds to an associative product $\ast$ on the free symmetric coalgebra $S(V)$. This product is called the \emph{Grossman--Larson product}; see \cite{Grossman-Larson}. It endows $S(V)$ with another Hopf algebra structure which is isomorphic to $\ca{U}(\g)$.\\ 
The previous construction has a non-commutative analogue. More precisely, it can be extended to a (general) post-Lie algebra $(\h,\pl)$, where $\h=(V,[-,-])$ and $\g=(V,\llbracket-,-\rrbracket)$, with the Lie bracket $\llbracket-,-\rrbracket$ being the one defined in \eqref{eq:secbra}.
	
It was shown in \cite{EFLM} that, in analogy to the abelian case, one can define on $\mathcal U(\h)$ a new associative product $\ast$, still named Grossman--Larson product, such that $\mathcal U_\ast(\h)=(\mathcal U(\h),\ast)$ is a Hopf algebra isomorphic, as Hopf algebra, to $\mathcal U(\g)$. Let us present another way to construct the Grossman--Larson product. 
	First, we observe the following claims:
\begin{itemize} 
	\item (Claim 1)  $\ca{U}(\h)$ is a left $\ca{U}(\g)$--module in the category of 
		coalgebras (i.e.~the action is compatible with the coproducts);
	\item (Claim 2) There exists an isomorphism $\Theta \co  \ca{U}(\g) \to \ca{U}(\h)$ 
		both of $\ca{U}(\g)$--modules and  of coalgebras. 
\end{itemize} 

From this, the Grossman--Larson product on $\ca{U}(\h)$ is given by $a\ast b := \Theta^{-1}(a)\cdot b$, where $\cdot$ refers to the left action. This is due to the following general result. 

\begin{lemma}
Let $A$ be a bialgebra and $B$ be an $A$--module in the category of coalgebras. 
Let $\Theta\co A \to B$ be an isomorphism of $A$--modules in the category of coalgebras. 
The map $\ast\co B\ot B \to B$ given by 
\begin{equation*}
	b\ast b' = \Theta^{-1}(b)\cdot b'
\end{equation*} 
for all $b,b'\in B$ makes $B$ into a bialgebra. Moreover, $\Theta$ is an isomorphism of bialgebras from $A$ to $(B,\ast)$.  
\end{lemma}

\begin{proof}
We leave the details of the proof to the reader. 
\end{proof}

In other words, the existence of the above isomorphism $\Theta\co A\to B$ permits to transfer the structure of $A$ to $B$.  

\begin{remark}
\begin{enumerate}
\item Note that, since $\Theta$ is a linear isomorphism one may transfer the product of $A$ to $B$ in the classical way: $B$ can be equipped with the product $b\ast'b' := \Theta( \Theta^{-1}(b)\cdot \Theta^{-1}(b'))$ for all $b,b'\in B$. 
Since $\Theta$ is an isomorphism of algebras between $A$ and $(B,\ast)$, both products $\ast'$ and $\ast$ coincide. 
This way of defining the Grossman--Larson product was presented in (the first part of the talk) \cite{Manchon-video}. 

\item The present point of view on the Grossman--Larson product can be related to Gavrilov \cite{Gavrilov} in his study on covariant derivatives and the double exponential. The K-map of Gavrilov (in fact $K^{-1}$) was defined in a similar way as the present map $\Theta$; for more details see Subsection 2.5 of \cite{Al-K_EF_M_MK}.	
\end{enumerate}
\end{remark}

 Following \cite[Section 3]{MQ-Crossed}, let us show Claim 1 and 2. Consider a post-Lie algebra $\h=(V,[-,-],\pl)$ with second Lie algebra $\g:=(V,\llbracket-,-\rrbracket)$ and recall from Proposition \ref{pro:equiv} that 
\begin{equation}
\label{eq:strange}
	\phi=\id_V\co \g\to \h
\end{equation}
is a crossed morphism relative to $\alpha \co \g\to \der(\h)$. 

\begin{remark}\label{rem:strange}
In the following, we will use the notation $\phi$ to denote the identity map $\id_V$ for two main reasons. First, we frequently need to consider the same element of $V$ as both an element of $\h$ and $\g$. Second, the construction presented here, which is based on the identity crossed morphism, can be generalized to include invertible crossed morphisms that differ from the identity. For results that suggest such generalizations, see Example \ref{ex:invR}, Proposition \ref{prop: adjun} and Proposition \ref{prop:invcross}. 
\end{remark}

We begin by the construction of the action 
\begin{equation*}
	M\co \ca{U}(\g) \to \End_{\mathbb K} (\ca{U}(\h)).
\end{equation*} 
First, let us observe that, since $\alpha_x=x\pl$ is a derivation of $\h$, it can be extended as a derivation for the associative algebra $\ca{U}(\h)$. 
Also, let 
\begin{equation*}
	L_\phi   \co \g    \rightarrow    \operatorname{End}_{\mathbb K}(\mathcal U(\h))
\end{equation*} 
 be the linear application defined by 
\begin{equation*}
	L_{\phi(x)}(X) =  \phi(x)\cdot X  \text{ for all }  x\in\g \text{ and } X\in\mathcal U(\mathfrak h),
\end{equation*} 
and let $M\co \mathfrak g\rightarrow\operatorname{End}_\mathbb K(\mathcal U(\mathfrak h))$ be the linear map defined by
\begin{equation}
	M_x=\alpha_x+L_{\phi(x)}, \text{ for all } x\in\mathfrak g.\label{eq:eqM}
\end{equation}
The following lemma shows that $M$ extends to a morphism of associative algebras $M \co \ca{U}(\g) \to \End_{\mathbb K} (\ca{U}(\h))$, providing the action map. 

\begin{lemma}
	For all $x,y\in\mathfrak g$, one has 
	\begin{equation}
		M_{[x,y]}=[M_x,M_y]. \label{eq:comm}
	\end{equation}
	In other words, $\mathcal U(\mathfrak h)$ carries a structure of $\g$--module defined by $M \co \g\rightarrow \End_\mathbb K(\ca{U}(\h))$. 
\end{lemma}

\begin{proof}
For every $x,y\in\mathfrak g$ and $a\in\mathfrak h$, it suffices to compare $M_{[x,y]}(a)$ with $[M_x,M_y](a)$. 
\end{proof}

\begin{lemma}
	The map $M$ defines an action $\ca{U}(\g)\ot \ca{U}(\h)\to \ca{U}(\h)$ in the category of coalgebras.
\end{lemma}
\begin{proof}
	 We denote by $\Delta_{sh}$ the unshuffle coproduct on $\ca{U}(\h)$ (or on $\ca{U}(\g)$). 
	 We use Sweedler's notation: $\Delta_{sh}(X)=X^1\ot X^2$ (omitting the sum). \\
	Observe that:
\begin{enumerate}
	\item for any derivation $f$ of $\ca{U}(\h)$ and for all $Y\in \ca{U}(\h)$, one has $(f\ot1 + 1\ot f)\Delta_{sh}(Y) =  \Delta_{sh}(f(Y))$; 
	\item for any $a\in \g$ and $Y\in \ca{U}(\h)$, one has $(L_{\phi(a)}\ot1 + 1\ot L_{\phi(a)})\Delta_{sh}(Y) = \Delta_{sh}(L_{\phi(a)}(Y))$  since the coproduct is a morphism of algebras.
\end{enumerate}
As a consequence, one has $(M_a\ot1 + 1\ot M_a)\Delta_{sh}(Y) =  \Delta_{sh}(M_a(Y))$ for all $a\in \g$ and  $Y\in \ca{U}(\h)$.
To conclude, we proceed by an induction on the length of the monomials of $\ca{U}(\g)$. 
For $a\in \g$ and  $X\in \ca{U}(\g)$, one has $(aX)^1\ot (aX)^2 = aX^1\ot X^2 + X^1\ot aX^2$. 
Supposing $M_{X}(Y)^1 \ot  M_{X}(Y)^2 = \Delta_{sh}(M_X(Y))$ for all $Y\in \ca{U}(\h)$, the inductive step is  
\begin{align*}
	M_{(aX)^1}(Y^1) \ot  M_{(aX)^2}(Y^2)
	&= M_{aX^1}(Y^1) \ot  M_{X^2}(Y^2) + M_{X^1}(Y^1) \ot  M_{aX^2}(Y^2)\\
	&= (M_{a}\ot \id) (M_{X^1}(Y^1) \ot  M_{X^2}(Y^2)) +   (\id\ot M_a) (M_{X^1}(Y^1) \ot  M_{X^2}(Y^2))\\
	&= (M_{a}\ot \id +  \id\ot M_a) (M_{X}(Y)^1 \ot  M_{X}(Y)^2)\\
	&= \Delta_{sh}(M_{a}(M_X(Y))). 
\end{align*}
\end{proof}

The map $\Theta \co \mathcal U(\mathfrak g)\rightarrow\mathcal U(\mathfrak h)$ is defined on every monomial $X\in\mathcal U(\mathfrak g)$ by
\begin{equation}
\label{eq:Theta}
	\Theta (X)=M_X(1)
\end{equation}
and is  extended to all $\mathcal U(\mathfrak g)$ by linearity. It is a morphism of coalgebras as well as of left $\mathcal U(\mathfrak g)$--modules; see \cite[Proposition 28]{MQS} and also Subsection 3.1 in \cite{EFLM}. 

Let us observe that, in particular, for all $x,y,z\in \g$, one has 
\begin{equation}\label{eq: Theta via phi}
\begin{split}
	\Theta(x) 		&= \phi(x) \\
	\Theta(xy) 	&= M_x(\phi(y)) = \phi(x)\phi(y) + \alpha_{x}(\phi(y)) \\
	\Theta(xyz) 	&= M_x( \phi(y)\phi(z) + \alpha_{y}(\phi(z)) )\\
				&= \phi(x)\phi(y)\phi(z) +  \alpha_x(\phi(y))\phi(z)   + \phi(y)\alpha_x(\phi(z)) 
					+ \phi(x)\alpha_{y}(\phi(z))  +\alpha_{x}( \alpha_{y}(\phi(z))) .
\end{split}
\end{equation}
To simplify notation, we recall that $\phi = \id_V$. Denoting $\alpha$ in terms of $\pl$, i.e., $\alpha_x=x \pl$, one obtains
\begin{align*}
	\Theta(x) 		&= x \\
	\Theta(xy) 	&= xy + x\pl y \\
	\Theta(xyz) 	&= xyz +  (x\pl y)z + y(x\pl z) + x( y\pl z) + x\pl(y\pl z). 
\end{align*}

From this, we can describe the Grossman--Larson product explicitly for small monomials:
\begin{align*}
	x\ast Y &= M_{\Theta^{-1}(x)}(Y) = M_x(Y) = xY + x\pl Y \\
	(xy)\ast Z &= M_{\Theta^{-1}(xy)}(Z) = M_{xy - x\pl y} (Z) \\
	&= xyZ + x(y\pl Z) + x\pl(yZ)+  x\pl(y\pl Z) -  (x\pl y)Z - (x\pl y)\pl Z \\
	&= xyZ + x(y\pl Z) + y(x\pl Z)+  x\pl(y\pl Z)  - (x\pl y)\pl Z.
\end{align*}

Note that, for all monomials $x_1\cdots x_n\in\ca{U}(\g)$ one has
\begin{equation}
\label{eq:starmon}
	\Theta(x_1\cdots x_n)=x_1\ast\cdots\ast x_n.
\end{equation}

In \cite[Proposition 6]{Al-K_EF_M_MK}, an interesting formula for Gavrilov's K-map has been given in terms of set partitions (see also \cite{EFMMK}). It is almost immediate to see that this result holds true for $\Theta$. It is as follows.  For each $n\geq1$, consider a partition $\pi=(B_1,\ldots,B_k)$ of $\{1,\ldots,n\}$ into $k$--blocks. The blocks are ordered according to their maximum, \ie $\max B_1< \max B_2< \cdots < \max B_k=n$.  For every block $B=\{b_1<b_2<\cdots <b_{\ell}\}$ of length $\ell$, we set 
\begin{equation}
\label{orderedblocks}
	x_B :=  x_{b_1}\pl (\cdots  (x_{b_{\ell-2}} \pl (x_{b_{\ell-1}}\pl x_{b_{\ell}}))\cdots ). 
\end{equation}
For each monomial $x_1\cdots x_n$, we let
\begin{equation}
\label{orderedblocksexpansion}
	(x_1 x_2\cdots x_n)^\pi := x_{B_1}x_{B_2} \cdots x_{B_k}
\end{equation}

\begin{proposition}
\label{prop:cumulants}
For each $n\geq 1$, one has 
\begin{equation}
\label{eq:setpartitions}
	\Theta(x_1\cdots x_n) 
	= \sum_{\pi} 	(x_1x_2\cdots x_n)^\pi,
\end{equation}
where the sum is taken over all set partitions $\pi$ of the strictly ordered set $\{ 1,\ldots,n\}$. 
\end{proposition}

\begin{proof}
First observe that $\Theta(xX)= \alpha_{x}(\Theta(X)) + x\Theta(X)$ for all $x\in \g$ and $X\in \ca{U}(\mathfrak g)$. 
From this the proof can be performed as in  \cite[Proposition 6]{Al-K_EF_M_MK},  by induction on the length of a monomial. 
\end{proof}


\subsection{D-bialgebras}
\label{ssec:Dalgebra}

The result of Guin--Oudom \cite{Guin-Oudom} establishes a correspondence between pre-Lie algebra structures on $V$ and associative products on the Hopf algebra $S(V)$ that make it into another Hopf algebra. In fact, this correspondence ``factorizes" through another type of structure; see  \cite[Proposition 2.7]{Guin-Oudom}. A similar result holds true in the context of post-Lie algebras; the resulting  structure is called a $D$--bialgebra. The notion of $D$--bialgebras is a slight enrichment of the initial notion of $D$--algebra that arises in the study of numerical integration methods on manifolds; see \cite{MKW}. The relation between post-Lie algebras and $D$--algebras was first established in \cite{Munthe-Kaas-Lundervold-PL}.  

\begin{definition}\label{def:dbial}
A $D$--bialgebra $D$ is a bialgebra $(D,\cdot,1,\Delta,\epsilon)$ endowed with an exhaustive increasing filtration
\begin{equation*}
	\Bbbk\cdot 1 =D^0\subset D^1 \subset \cdots \subset D^n \subset \cdots 
\end{equation*}
such that $D^i\cdot D^j \subset D^{i+j}$ and such that its set of primitive elements is given by $\ker(\epsilon)\cap D^1 = \text{Prim}(D)$ and generates $D$ as an algebra. 
Moreover, $D$ is  endowed with a map $\pl\co D\ot D \to D$ such that: 
\begin{enumD}
	\item\label{D bial item1} $1\triangleright X=X$ for all $X\in D$;
	\item\label{D bial item3} $\Delta(X\triangleright Y)=(X^{(1)}\triangleright Y^{(1)})\ot (X^{(2)}\triangleright Y^{(2)})$ for all $X,Y\in D$; 
	\item\label{D bial item4} $X\triangleright(Y\cdot Z)=(X^{(1)}\triangleright Y)\cdot(X^{(2)}\triangleright Z)$ for all $X,Y$ and $Z$ in $D$;  and,  
	\item\label{D bial item5} $(x\cdot X)\triangleright y=x\triangleright (X\triangleright y)-(x\triangleright X)\triangleright y$ for all $X\in D$ and $x$ and $y$ in $\text{Prim}(D)$. 
\end{enumD}
\end{definition} 

\begin{remark}
From \ref{D bial item1} and \ref{D bial item4} one obtains that 
\begin{equation}\label{eq: item 1 bis}
	X\pl 1= \epsilon(X)    \text{ for all } X\in D. 
\end{equation}
In a $D$-bialgebra, the primitive elements are closed under $\pl$, that is  
\begin{equation}\label{eq: item 3 bis}
	\pl \Big( \text{Prim}(D) \ot \text{Prim}(D) \Big) \subset \text{Prim}(D).
\end{equation} 
This is a consequence of \ref{D bial item3} and \eqref{eq: item 1 bis}. In fact, in a D-bialgebra, \ref{D bial item3} and \eqref{eq: item 3 bis} are equivalent properties, that is, one may define a D-bialgebra as before by substituting \ref{D bial item3} by \eqref{eq: item 3 bis}. 
To see this, one may use the fact that  $\text{Prim}(D)$ generates $D$ and the properties \ref{D bial item4} and \ref{D bial item5}.
	
Note also that $\text{Prim}(D)$ is closed under antisymmetrization  of the product of $D$. 
Moreover, since $\text{Prim}(D)$ generates $D$, the coproduct is in fact the shuffle coproduct. 
\end{remark}

In what follows, we recall these correspondences in the more general context of post-Lie algebras. 

As shown in the previous section, for all primitive $z \in \g$ and all monomials $X= x_1x_2\cdots x_n\in  \ca{U}(\g)$,  the Grossman--Larson product $X\ast z$ has a non-trivial term in $\g$, which is given in terms of iterations of the post-Lie product $\pl$. For instance, the term in $\g$ of the product $(xy)\ast z$ is $x\pl(y\pl z)  - (x\pl y)\pl z= a_{\pl}(x,y,z)$. If $V$ is a pre-Lie algebra, this term is symmetric in the first two variables. In fact, it turns out that there is a natural way to extend the pre-Lie product $\pl\co V\ot V \to V$ to an operation $\pl\co S^{n}(V) \ot V \to V$ for all $n\geq 1$ with no loss of information: these operations are given by the corollas described in \eqref{eq: braces}. A similar fact holds true for  post-Lie algebras, but considering the corollas \eqref{eq: braces PSB}. 
This operadic point of view turns out to be convenient to prove  the following result. 

\begin{proposition}\label{pro:68}
There is a one-to-one correspondence  between the structures of post-Lie algebra $(V,\triangleright,[-,-])$ and the structure of $D$-bialgebras on $\ca{U}(\h)$, where $\h=(V,[-,-])$.
\end{proposition}

\begin{proposition}\label{pro:69}
Given a $D$--bialgebra $(D,\cdot,\Delta,\pl)$, the  product
\begin{equation}
\label{eq:GrossmanLarson}
	A\ast B = A^{(1)}\cdot (A^{(2)}\pl B) \qquad \forall A,B\in D,
\end{equation}
is associative. Moreover, if $D=\ca{U}(\h)$ as above, then $\ast$ is the Grossman--Larson product. 	
\end{proposition}

Hereafter $\ca{U}_\ast(\h)$ will denote coalgebra $\ca{U}(\h)$ endowed with Grossman--Larson product \eqref{eq:GrossmanLarson}. 
In particular, we have that 
\begin{equation}\label{eq: theta2: morphism}
	\Theta:\ca{U}(\g)\rightarrow\ca{U}_\ast(\h)
\end{equation} 
defined in \eqref{eq:starmon} is an isomorphism of Hopf algebras, see Section 3 in \cite{EFLM} and also Subsection 5.1 in \cite{EFMpostLiealgebra-factorization}.

	
\section{Magnus expansions and integration of post-Lie algebras}
\label{sec:Magnus}
	
	\def\Cross{\operatorname{Cross}}
	\def\Der{\operatorname{Der}}
	\def\End{\operatorname{End}}
	\def\Aut{\operatorname{Aut}}
	\def\CM{\cat{CM}}
	\def\CMinv{\cat{CM_{inv}}}
	\def\CMid{\cat{CM_{id}}}
	\def\CMGp{\cat{CMGp}}
	\def\CMGpinv{\cat{CMGp_{inv}}}
	\def\CMGpid{\cat{CMGp_{id}}}
	\def\CMpl{\cat{CM_{pl}}}
	\def\PB{\cat{Pbialg}}
	\def\IsoPB{\cat{Pbialg_{inv}}}
	\def\IsoPBid{\cat{PBialg_{id}}}
	\def\U{\mathbb{U}}
	\def\BCH{\operatorname{BCH}}
	\def\locGp{\cat{CMGp^{loc}}}
	\def\Exp{\operatorname{Exp}}
	\def\ImU{\cat{Im}(\widehat{\ca{U}})} 
	\def\CMinvfin{\cat{CM^{fin}_{inv}}} 
	\def\PBcomp{\cat{\widehat{Pbialg}_{inv}}} 
	\def\GpSplit{\cat{SplitGp}}
	\def\Split{\cat{Split}}
	
In this final section, first we introduce the so-called post-Lie Magnus expansion and then use it to address the problem of integration in post-Lie algebra. Along the way, we will also comment on the relation of this remarkable formal series with the Grossman--Larson product. Our approach will be based on the notion of crossed morphisms as presented in the first part of this note.


\subsection{Preliminaries} 

First, let us set up the framework. 
For more detailed information about what follows we refer the reader to \cite{Bai-Guo-Sheng-Tang-post-groups} and the references therein. 
Since hereafter we will need to handle (infinite) series of elements of post-Lie algebras and of their universal enveloping algebras, it will be convenient to work with post-Lie algebras endowed with an additional structure described here below. 
Recall that descending filtration $\mathcal F_\bullet V$ on a vector space $V$ is a chain 
\[
	V\supset\mathcal F_1 V\supset \mathcal F_2 V\supset \mathcal F_3 V\supset\cdots\supset \mathcal F_n V\supset\cdots
\] 
of (nested) vector subspaces of $V$ such that, for all $i$, $\mathcal F_{i+1} V$ is vector subspace of $\mathcal F_i V$. A vector space endowed with such a filtration will be called filtered. It is worth mentioning that filtered vector spaces form a (linear symmetric monoidal) category, whose morphisms are filtration preserving linear maps and whose tensor product is defined by a Cauchy-like formula, i.e.~for all $n$, $\mathcal F_n(V\otimes W)=\oplus_{i+j=n}\mathcal F_i V\otimes \mathcal F_j W$. Given a filtered vector space one defines its \emph{completion} as the inverse limit
\[
	\hat V:=\mathop{\lim_{{\longleftarrow}}}_{i}V/\mathcal F_i V
	=\Big\{(\overline{v}_1,\overline{v}_2,\dots)\in\prod_{i}V/\mathcal F_i V\, \vert\;v_j\equiv v_i\text{mod}\;(\mathcal F_i V),\;j>i\Big\},
\]
which is a vector space endowed with a (natural) filtration defined by 
\[
	\mathcal F_i\hat V=\big\{(\overline{v}_1,\overline{v}_2,\dots)\in\hat V\vert\; \overline{v}_j=0,\;j\leq i\big\}.
\]
Note that there is a \emph{canonical} linear map $\gamma:V \rightarrow \hat V$, defined by $\gamma(v)=(\overline{v},\overline{v},\dots)$. We say that a vector space $V$ is complete with respect to a descending filtration $\mathcal F_\bullet V$, if the map $\gamma$ is an isomorphism of  (filtered) vector spaces. Note that also completed vector spaces form a linear symmetric monoidal category, whose \emph{completed} tensor product is defined as 
\[
	\hat V\hat\otimes\hat U=\mathop{\lim_{{\longleftarrow}}}_{i,j}V/\mathcal F_i V\otimes U/\mathcal F_j U.
\]

\begin{remark}
In the previous comments one can trade the category of finite dimensional vector spaces with the one of rings, modules over rings, groups, and other structures. For example, in what follows we will be interested in the category of completed (finite) dimensional post-Lie algebras. In all of the possible cases, the construction sketched above has the following rationale. The existence of the filtration permits to define a topology on the underlying set so that the notions of convergence and of Cauchy sequence make sense. The operation of completion with respect to a filtration defines another object, belonging to the same category as the original one and having, among others, the following two properties: 1) there is a canonical map from the original object to the completed one which, under some further (mild) assumptions, identifies the former with a sub-object of the latter, 2) In the completed object all Cauchy sequences are convergent. 
\end{remark}

We can now introduce the following 

\begin{definition}
A complete post-Lie algebra is a post-Lie algebra $(\h,\pl)$ whose underlying vector space $V$ is endowed with a descending filtration $\mathcal F_\bullet V$ such that:
\begin{enumerate}
\item For all $i,j$, $[\mathcal F_i V,\mathcal F_j V]\subset\mathcal F_{i+j} V$ and $\mathcal F_i V\pl\mathcal F_j V\subset\mathcal F_{i+j} V$;
\item $V$ is complete with respect to this filtration.
\end{enumerate}
If, furthermore, $\mathcal F_1 V=V$ one says that the filtered post-Lie algebra is connected.
\end{definition}

Now let $(\h,\pl)$ be a post-Lie algebra and define recursively the following vector subspaces of its underlying vector space $V$
\begin{equation}
\label{eq:filt}
	\mathcal F_1:=V,\; 
	\mathcal F_2:=[V,V]+V\pl V\quad\text{and} \quad 
	\mathcal F_n=\oplus_{j=1}^{n-1}[\mathcal F_{j} V,\mathcal F_{n-j} V]+\oplus_{j=1}^{n-1}\mathcal F_{j} V\pl\mathcal F_{n-j} V,
	\; \forall n\geq 3.
\end{equation}
Then the completion of $(\h,\pl)$ with respect to the previous filtration is a connected, complete post-Lie algebra which, by abusing notation, will be denoted by $(\hat h,\pl)$. The filtration \eqref{eq:filt} was tailored to guarantee that some relevant Lie series are defined by Cauchy sequences and, for this reason, are convergent in $\hat h$. The first of these is the Baker--Campbell--Hausdorff series, BCH-series from now on, which is defined for all $x,y\in V$, by

\begin{equation}
	x\bullet_{\h}y := x+y + \frac{1}{2}[x,y]_{\h} + \frac{1}{12} \big(  [x,[x,y]_{\h}]_{\h} 
											+[y,[y,x]_{\h}]_{\h} \big) + \cdots. \label{eq:BCHG}
\end{equation}
Classically, when convergence issues are relevant, this series defines the structure of a local Lie group on a suitable neighbourhood of the neutral element of $V$. On the other hand, in our present context, the previous formula defines the structure of \emph{formal group} on $V$, denoted hereafter by $\mathcal H$, whose identity is the neutral element of $V$ and such that the inverse of $x\in \mathcal H$ is $-x$. Note that the underlying set of this group is $V$, the same as the one of the Lie algebra $\h$. Furthermore, using conventions and notations adopted in Subsection \ref{subsec:postLie}, together with the post-Lie algebra $(\h,\pl)$ one considers the Lie algebra $\g=(V,\llbracket-,-\rrbracket)$, which, since $\llbracket x,y\rrbracket=x\pl y-y\pl x+[x,y]$, can be completed to $\hat\g$ using the filtration \eqref{eq:filt}. Together with the two completed Lie algebra $\hat\h$ and $\hat\g$, one may consider their corresponding completed universal enveloping algebras, hereafter denoted by $\hat{\mathcal U}(\h)$ respectively $\hat{\mathcal U}(\g)$. For an exhaustive and in-depth analysis of the relations between a Lie algebra, its completion with respect to a filtration and their corresponding universal enveloping algebras we refer the reader to Chapter 7 and Chapter 8 of the monograph \cite{Fresse}. Here we simply remark that $\hat{\mathcal U}(\h)$ is a complete Hopf algebra, whose coproduct $\Delta$ is defined in terms of the completed tensor product $\hat\otimes$, see above in this section. Moreover, an element of $\hat{\mathcal U}(\h)$ is primitive if and only if it belongs to $\hat\h$, i.e.~$\xi\in\operatorname{Prim}(\hat{\ca{U}}(\h))$ if and only if $\Delta\xi=\xi\hat\otimes 1+1\hat\otimes\xi$.

\begin{remark}
Recall that together with $\mathcal U(\g)$ and $\mathcal U(\h)$ one can introduce $\mathcal U_\ast(\h)=(\mathcal U(\h),\ast)$ defined in terms of the Grossman--Larson product, see \eqref{eq:GrossmanLarson} in Section \ref{sec:unipo}, whose corresponding completion will be denoted by $\hat{\mathcal U}_\ast(\h)$. 
\end{remark}

For these completed universal enveloping algebras one can define the so-called \emph{exponential} and \emph{logarithm} maps. Considering, for example, $\hat{\mathcal U}(\h)$, we let 
\begin{equation}
\label{eq:primel}
	\exp_{\h} \co \operatorname{Prim}(\hat{\ca{U}}(\h)) \to \operatorname{Gp}(\hat{\ca{U}}(\h))
\end{equation}
be the \emph{exponential map}, where $\operatorname{Gp}(\hat{\ca{U}}(\h))$ is the set of group-like elements of $\hat{\ca{U}}(\h)$, that is
$\eta\in \hat{\ca{U}}(\h)$ is group-like if and only if $\Delta(\eta)=\eta\hat\otimes\eta$. Furthermore, recall that the exponential map \eqref{eq:primel} is bijection defined by $\exp_{\h}(x)=\sum_{j=0}^\infty\frac{x^n}{n!}$, where $x^n=x\cdots x\in\ca{U}(\h)$, whose inverse, denoted by $\log_{\h}$,  is called the \emph{logarithm map}. Similarly, one can consider $\exp_{\g}$ and $\log_{\g}$ as well as 
\begin{equation*}
 	\exp_{\ast} \co \operatorname{Prim}(\hat{\ca{U}}_\ast(\h)) \to \operatorname{Gp}(\hat{\ca{U}}_\ast(\h))
\end{equation*}
together with its inverse  $\log_{\ast}$.


\subsection{The post-Lie Magnus expansion and integration of post-Lie algebras} 

Note that, since $\Theta$ is an isomorphism of Hopf algebras, see \eqref{eq: theta2: morphism}, it restricts to an isomorphism $\Theta:\operatorname{Gp}(\hat{\ca{U}}(\g))\rightarrow\operatorname{Gp}(\hat{\ca{U}}_\ast(\h))$, such that for all $x\in\g$  
\begin{equation}
\label{eq: theta and phi}
 	\Theta (\exp_{\g}(x)) = \exp_{\ast}(\phi(x)).
\end{equation}
Recall that $\phi \co \g \to \h$, see \eqref{eq:strange} and also Remark \ref{rem:strange}.
Before moving on, we would like to make the following observations.
  
\begin{remark}
\label{rem:iso}
\begin{enumerate}
\item If the post-Lie algebra $(\mathfrak h,\pl,[-,-]_{\mathfrak h})$ has trivial Lie bracket, $[-,-]_{\mathfrak h}=0$, that is if it is a pre-Lie algebra, then 
\begin{equation}
\label{eq:primelpreLie}
	\exp_{\h} \co \operatorname{Prim}(\hat{\ca{S}}(\h)) \to \operatorname{Gp}(\hat{\ca{S}}(\h)).
\end{equation}

\item More importantly for what follows, in particular with respect to Formula \eqref{eq: varPhi def} below, we note that the exponential map $\exp_\h$ is a bijection between  $\operatorname{Prim}(\hat{\ca{U}}(\h))$ and $\operatorname{Gp}(\hat{\ca{U}}(\h))$. The group law \eqref{eq:BCHG} is defined via the BCH formula
\[
	\exp_\h x\, \exp_\h y=\exp_\h(x\bullet_\h y).
\] 
From this observation it follows that the formal group $\mathcal H$ is isomorphic, via $\exp_h$, to $\operatorname{Gp}(\hat{\ca{U}}(\h))$. The same remark applies to the Lie algebra $\g$ and the formal group $\mathcal G$.
\end{enumerate}
\end{remark}

\medskip

Given formal groups, $\mathcal G$ and $\mathcal H$, defined as above, we will argue that $\mathcal G$ acts (non-trivially) on $\mathcal H$ via automorphisms and then we will show the existence of a crossed morphism relative to this action. To this end, first observe that the action  $\alpha\co \g \to \Der(\h)$, $\alpha_x:=x \pl$, naturally induces an action at the level of groups, $\varUpsilon \co \ca{G} \to \Aut(\ca{H})$, given by 
\begin{equation}
\label{eq: varUpsil def}
	\varUpsilon= \Exp(\alpha).
\end{equation}
Explicitly, 
\begin{equation}
\label{eq: varUpsil}
	\varUpsilon_x(y) = \Exp(\alpha_x)(y) 
		= y + \sum_{n\geq 1} \frac{1}{n!}\alpha_{x}^n(y).
\end{equation}

\begin{lemma}
\label{lem: varUps crossed morph}
The map $\varUpsilon$ is a morphism of formal groups. 
\end{lemma}

\begin{proof}
Note that if $d \in \Der(\g)$, then $\Exp(d) \in \Aut(\g) \subset \Aut(\ca{G})$. Furthermore, if $\g$ and $\h$ are two (finite dimensional) Lie algebras and $\alpha \co \g \to \Der(\h)$ a morphism of Lie algebras, then one has for all $x,y \in \mathfrak g$
\begin{equation}\label{eq:idUP}
	\operatorname{Exp}(\alpha_x)\operatorname{Exp}(\alpha_y)
	=\operatorname{Exp}( \alpha_x \bullet_{\operatorname{End}(\h)}\alpha_y)
	=\operatorname{Exp}(\alpha_{ x\bullet_{\g} y }),
\end{equation}
which gives the result. Observe that, in the previous formula, the second equality holds true since $\alpha_{\cdot}:\g\rightarrow\text{End}(\h)$ is a morphism of Lie algebras, see Lemma \ref{lem:pl}, where the Lie algebra structure on $\text{End}(\h)$ is the one obtained skew-symmetrizing its (standard) associative product.
\end{proof}

The  isomorphism $\Theta:\ca{U}(\g)\rightarrow\ca{U}_\ast(\h)$ defined in \eqref{eq:starmon} induces a map $\varPhi\co \ca{G} \to \ca{H}$  via 	
\begin{equation}
\label{eq: varPhi def}
	\varPhi = \log_{\h} \circ\ \Theta \circ \exp_{\g},
\end{equation}
see Remark \ref{rem:iso}. Recalling identity \eqref{eq:setpartitions}, we see immediately that \eqref{eq: theta and phi} can be refined.

\begin{proposition}
\label{prop:cumulants}
If the post-Lie algebra $(\mathfrak h,\pl)$ has trivial Lie bracket, that is if it is a pre-Lie algebra, then
\begin{equation}
\label{eq: preLieMagnusgeneral}
	\Theta(\exp_{\mathfrak g} x) 
	= \exp_*(\phi(x)) 
	= \exp_{\mathfrak h} \Big(\sum_{n \ge 1}  \frac{\phi_n(x)}{n!}\Big),
\end{equation}
where the n-fold pre-Lie iteration is defined to be 
$$
	\phi_n(x) :=  \phi(x) \pl (\phi(x) \pl (\cdots  (\phi(x) \pl \phi(x)))\cdots ).
$$ 
This implies that 
\begin{equation}
\label{eq:inversepreLieMag}
	\varPhi(x) = \sum_{n \ge 1}  \frac{\phi_n(x)}{n!} = \frac{\varUpsilon_x - 1}{\mathrm{Ln}(\varUpsilon_x)} (x).
\end{equation}
\end{proposition}

\begin{proof} 
The first equality in \eqref{eq: preLieMagnusgeneral} follows from $\Theta:\operatorname{Gp}(\hat{\ca{U}}(\g))\rightarrow\operatorname{Gp}(\hat{\ca{S}}_\ast(\h))$ being an isomorphism. The second equality follows from identity \eqref{eq:setpartitions} 
$$
	\Theta(x^n) 
	= \phi(x) * \cdots * \phi(x)
	= \sum_{\pi} (\phi(x)^n)^\pi 
$$
and the fact that the ordering in \eqref{orderedblocksexpansion} is irrelevant in the commutative algebra $\hat{\ca{S}}(\h)$, which implies the well-known relation between exponential generating series
$$
	\exp_*(\phi(x)) = \exp_{\mathfrak h} \Big(\sum_{n \ge 1}  \frac{\phi_n(x)}{n!}\Big).
$$
The first equality in \eqref{eq:inversepreLieMag} follows from
$$
	\exp_{\mathfrak h} \varPhi(x) = \exp_*(\phi(x)), 
$$
whilst the second equality follows from \eqref{eq: varUpsil} and 
\begin{equation}
\label{eq:inversepreLieMagnew}
	 \sum_{n \ge 1}  \frac{\phi_n(x)}{n!} 
	 = \phi(x) + \frac{1}{2!} \phi(x) \pl \phi(x) 
	 	+ \frac{1}{3!} \phi(x) \pl (\phi(x) \pl \phi(x)) + \cdots
	 = \frac{\Exp(\alpha_x) - 1}{\alpha_x}(x).	 
\end{equation}
\end{proof}

Since $\varPhi$ is a composition of invertible maps, it is itself invertible.  From Proposition \ref{prop:cumulants} we deduce

\begin{corollary}
\label{cor:inversecumulants}
If the post-Lie algebra $(\mathfrak h,\pl, [-,-]_{\mathfrak h})$ has trivial Lie bracket, that is if it is a pre-Lie algebra, then the pre-Lie Magnus expansion is
\begin{equation}
\label{eq:preLieMagabstract}
	\varPhi^{-1} (x) = \frac{\mathrm{Ln}( \varUpsilon_{\varPhi^{-1}(x)} )}{\varUpsilon_{\varPhi^{-1}(x)} - 1}(x).
\end{equation}
\end{corollary}	

This is a direct consequence of the last equality of \eqref{eq:inversepreLieMagnew}, which implies that 
$$
	\varPhi^{-1} (x) 
	= \frac{\alpha_{\varPhi^{-1}(x)}}{\Exp(\alpha_{\varPhi^{-1}(x)}) - 1}(x)	
	= \sum_{n \ge 0} \frac{B_n}{n!} \alpha^n_{\varPhi^{-1} (x)}(x).
$$
\begin{remark}\label{rem:classvsmod}
Note that by definition $\alpha_{\varPhi^{-1}(x)} = L_{{\varPhi^{-1}(x)} \pl}$, which permits to connect Corollary \ref{cor:inversecumulants} back to Formula \eqref{eq:preLieMag1a}, where we introduce the concept of pre-Lie Magnus expansion in the classical context of Agrachev and Gamkrelidze.
\end{remark}

Returning to identity \eqref{eq: varPhi def}, we have seen that its inverse reduces to the pre-Lie Magnus expansion on a post-Lie algebra with trivial Lie bracket. It is therefore natural to identify the inverse of \eqref{eq: varPhi def} as the post-Lie Magnus expansion in the following.

\begin{definition}
The \emph{post-Lie Magnus expansion}, $\chi \co  \ca{H} \rightarrow \ca{G}$, defined via the post-Lie algebra $(\h,\pl,[-,-]_{\h})$ is the inverse of the map $\varPhi= \log_{\h} \circ\ \Theta \circ \exp_{\g}$. Explicitly, 
\begin{equation}
\label{eq:postLieMag}
	\chi :=\varPhi^{-1}= \log_{\g}\circ\ \Theta^{-1} \circ \exp_{\h}. 
\end{equation}
\end{definition}

We note that computing the terms in the post-Lie Magnus expansion and its inverse \eqref{eq: varPhi def} is far more involved than their pre-Lie cousins. We will return to this problem briefly in Subsection \ref{ssec:comput}.

The following result characterizes the algebraic nature of  the post-Lie Magnus expansion and will be crucial for what follows.

\begin{lemma}\label{lem:intpostLie}
For all $a,b\in \h$ one has 
\begin{equation}
\label{eq:bch1}
	\chi \big(a \bullet_{\h} \varUpsilon_{\chi(a)}(b)\big)
	=\chi(a) \bullet_{\g} \chi(b).  
\end{equation}
\end{lemma} 

\begin{proof}[Sketch of proof]
On the one hand, for all $a,b\in \h$, one has 
\begin{equation}
\label{eq: side 1}
	a  \bullet_{\h}\varUpsilon_{\chi(a)}(b)
 	= a  \bullet_{\h} (\exp_\h(a) \triangleright b) 
	=  \log_\h(\exp_\h(a)\ast\exp_\h(b)). 
\end{equation}
The first equality of \eqref{eq: side 1} is a direct consequence of the fact that 
\begin{equation*}
	\exp_\h(a) \triangleright b=\varUpsilon_{\chi(a)}(b),
\end{equation*} 
for all $a,b\in\h$, which was shown in \cite{MQS}, inspired by a similar result in the pre-Lie case; see \cite{Manchon2009}. The second equality  of \eqref{eq: side 1} was established by Fl{\o}ystad and Munthe-Kaas in \cite{F-MK} by introducing the \emph{composition product} 
$\sharp \co \h\times \h\rightarrow \h$ 
defined for all $a,b\in \h$ by 
\begin{equation*}
	a \sharp b=\log_\h(\exp_\h(a)\ast\exp_\h(b)). 
\end{equation*}
Their approach was motivated by an algebro-geometric perspective on post-Lie algebra and associated constructions in the enveloping algebras in terms of formal fields and flows. 
		
On the other hand, for all $a,b\in \h$, one has 
\begin{equation}\label{eq: side 2}
	\log_\h(\exp_\h(a)\ast\exp_\h(b)) 
	= \log_\h \Theta \big( \exp_\g(\chi(a)) \exp_\g(\chi(b)) \big)
	= \log_\h \exp_\ast \phi  \big( \chi(a) \bullet_\g \chi(b) \big).
\end{equation}
The first equality of \eqref{eq: side 2} follows from the definition of $\chi$ and the fact that $\Theta$ is a morphism of algebras. 
The second equality follows from the definition of $\bullet_{\g}$ and \eqref{eq: theta and phi}.  
		
The result follows by combining \eqref{eq: side 1} and \eqref{eq: side 2}, using the log version of \eqref{eq: theta and phi}. 
\end{proof}

Note that \eqref{eq:bch1} implies that 
\[
	\chi (a \bullet_{\h} b )
	=\chi(a) \bullet_{\g} \chi\big((\varUpsilon^{-1}_{\chi(a)}(b)\big)
	=\chi(a) \bullet_{\g} \chi\big(\varUpsilon_{-\chi(a)}(b)\big). 
\]
	
We can now state the following important
	
\begin{proposition}\label{lem: varPhi crossed morph}
The map $\varPhi$ is a crossed morphism of formal groups relative to the action $\varUpsilon$ defined in \eqref{eq: varUpsil def}. 	
\end{proposition}

\begin{proof}
It is a direct consequence of the above Lemma \ref{lem:intpostLie}. Indeed,  Formula \eqref{eq:bch1} above can be written as 
\begin{equation*}\label{eq:bch2}
	\varPhi(x \bullet_{\g}y)= \varPhi(x) \bullet_{\h} \varUpsilon_{x}(\varPhi(y)). 
\end{equation*}
\end{proof}

In summary, one has

\begin{corollary}
\label{cor:quadrupel}
To every post-Lie algebra $(\h,\pl)$ it is associated the quadruple $(\mathcal G,\mathcal H,\varUpsilon,\varPhi)$, where $\mathcal H,\mathcal G$ are the formal groups defined above, and $\varPhi:\mathcal G\rightarrow\mathcal H$ is an invertible crossed morphism relative to the action $\varUpsilon:\mathcal G\rightarrow\text{Aut}(\mathcal H)$ defined in \eqref{eq: varPhi def} respectively in \eqref{eq: varUpsil def}. \end{corollary}

\begin{remark}
\label{rmk:relRB}
Formula \eqref{eq:bch1} identifies the inverse of $\varPhi$, i.e., the post-Lie Magnus expansion as a relative Rota--Baxter operator from $\mathcal H$ to $\mathcal G$, in the sense of \cite[Def.~3.1]{JSZ2021}.  
\end{remark}
 
Let $V$  be the underlying vector space of the Lie algebra $\h$, see the beginning of this subsection, and let $\star: V \times V\rightarrow V$ be defined by 
\begin{equation}
	a\star b := a\bullet_\h\varUpsilon_{\chi(a)}(b). \label{eq:starProH}
\end{equation}

\noindent One has the following

\begin{proposition}\label{pro:79}
The binary operation $\star$ defined in \eqref{eq:starProH} employs a structure of formal group $\overline{\mathcal H}$ on $V$ such that $\mathrm{id}:\overline{\mathcal H}\rightarrow\mathcal H$ is a crossed morphism relative to the $\overline{\mathcal H}$-action on $\mathcal H$ defined by $\varUpsilon\circ\chi$.
\end{proposition}

\begin{proof}
It follows directly from Proposition \ref{prop:invcross} applied to the present framework.
\end{proof}

From the isomorphism of categories $\cat{SE(gp)}\cong \cat{PostLie(gp)}$  (see Proposition \ref{prop: iso SE PL groups}) the quadruple  $(\overline{\mathcal H},\mathcal H,\varUpsilon\circ\chi,\mathrm{id})$ corresponds to the post-Lie group $(\mathcal H,\bullet_\h,\blacktriangleright)$, where  for all  $a,b\in \mathcal H$
\begin{equation*}
	a \blacktriangleright b:=\varUpsilon_{\chi(a)}(b).
\end{equation*}

Finally, we can naturally think of the post-Lie group $(\mathcal H,\bullet_\h,\blacktriangleright)$ as a formal integration of the post-Lie algebra $(\h,\pl)$. 

\begin{theorem}\label{thm:81}
Let $(\mathcal H,\bullet_\h,\blacktriangleright)$ be a formal integration of $(\h,\pl)$. By applying to $(\mathcal H,\bullet_\h,\blacktriangleright)$ the differentiation procedure of Subsection \ref{ss:postLie} (in particular, see Proposition \ref{prop:postgroups}) one recovers $(\h,\pl)$. 
\end{theorem}

\begin{proof}
This follows at once from the discussion in Subsection \ref{ss:postLie}.
\end{proof}

A couple of comments are in order.
\begin{remark}
\begin{enumerate}
\item First we would like to comment on the presence of $\g$ in the quadruple mentioned in Corollary \ref{cor:quadrupel}. To this end, it is worth noticing that 
\begin{align}
\label{eq:starprod}
\begin{aligned}
	a \star b
	&= a \bullet_\h\varUpsilon_{\chi(a)}(b)\\
	&=a+\mathrm{Exp}(\alpha_{\chi(a)})(b)+\frac{1}{2}[a,\mathrm{Exp}(\alpha_{\chi(a)})(b)]_\h+\cdots\\
	&= a + b + a \triangleright b + \frac{1}{2}[a,b]_\h+\cdots,
\end{aligned}
\end{align}
where the last equality comes from the expansion $\chi(a)=a+O(a\! \pl a)$, see Formula \eqref{eq:postLieMagorder3} below. This entails that the Lie algebra of the formal group $\overline{\mathcal H}$ is the one with its Lie bracket given by the anti-symmetrization of the linear form $B_{\overline{\mathcal H}}:V\times V\rightarrow V$ defined by $B_{\overline{\mathcal H}}(a,b)=a\! \pl b+\frac{1}{2}[a,b]_\h$. In other words, one has that $\g=\overline\h$. The latter equality should not surprise. In fact, the map $\Theta:\mathcal U(\g)\rightarrow\mathcal U(\h)$ defined in Subsection \ref{ss:GL} is an isomorphism of Hopf algebras between $\mathcal U(\g)$ and $\mathcal U_\ast(\h)$ which restricts to the identity on primitive elements.

\item The quadruple $(\g,\h,\alpha,\text{id})$ can be considered an avatar of the post-Lie algebra $(\h,\pl)$ and for this reason $(\overline{\mathcal H},\mathcal H,\varUpsilon\circ\chi,\text{id})$ can be thought of a global object integrating it.

\item The post-Lie Magnus expansion, $\chi\co \overline{\mathcal{H}} \to \mathcal{G}$,  is a morphism of groups. (See Lemma \ref{lem:intpostLie})
$$
	\chi(a\star b) = \chi(a) \bullet_{\g} \chi(b).  
$$
This group morphism property is consistent with the fact that $\chi$ is a relative Rota--Baxter operator from $\mathcal H$ to $\mathcal G$, see Remark \ref{rmk:relRB}. It therefore defines a group morphism from the Grossman--Larson goup \cite{Al-K_EF_M_MK23} to $\mathcal G$. 
\end{enumerate}
\end{remark}


\subsection{The Grossman--Larson product via the post-Lie Magnus expansion}

Before moving to our final topic we would like to add one more comment regarding the role played by the post-Lie Magnus expansion in the theory of post-Lie algebra.  As was discussed at length in Subsection \ref{ss:GL}, the post-Lie product $\pl$ admits an extension to a product on $\mathcal U(\h)$ which yields an associative product $\ast:\mathcal U(\h)\times\mathcal U(\h)\rightarrow\mathcal U(\h)$, see Formula \eqref{eq:GrossmanLarson}. The latter is known the Grossman--Larson product and endows $\mathcal U(\h)$ with a ``new" Hopf algebra structure which turns out to be isomorphic to the one defined (canonically) on $\mathcal U(\g)$. We will now show that the Grossman--Larson product, when restricted to group-like elements of $\hat{\mathcal U}(\h)$, is defined by the post-Lie Magnus expansion. To this end let us first prove that

\begin{lemma}\label{lem:glth}
The set $\text{Gp}(\hat{\mathcal U}(\h))$ is closed under the product \eqref{eq:GrossmanLarson}.
\end{lemma}

\begin{proof}
First, recall that $\mathcal U(\h)$ (as well as $\hat{\mathcal U}(\h)$) is a $D$-bialgebra, see Definition \ref{def:dbial}, whose set of primitive elements is $\h$. Then, for $a,b \in \h$, 
\begin{equation}
	\exp_\h(a)\ast\exp_\h(b)\stackrel{\eqref{eq:GrossmanLarson}}{=}
	\exp_\h(a)\cdot(\exp_\h(a)\pl\exp_h(b)).\label{eq:laG}
\end{equation}
Using induction on the length of monomials together with \ref{D bial item5}, one obtains
\[
	\exp_\h(a)\pl\underbrace{b\cdots b}_{n-\text{times}}
	=\underbrace{\exp_\h(a)\pl b\cdots\exp_\h(a)\pl b}_{n-\text{times}}.
\]
This implies that the right-hand side of \eqref{eq:laG} is equal to $\exp_\h(a)\exp_\h(\exp_\h(a)\pl b)$, which entails
\[
	\exp_\h(a)\ast\exp_\h(b)
	=\exp_\h\big(a\bullet_\h((\exp_\h a)\pl b)\big).
\]
\end{proof}

We can now state the following important result.

\begin{theorem}\label{thm:glth}
{\it The restriction of the Grossman--Larson product to $\text{Gp}(\hat{\mathcal U}(\h))$ is the binary operation, still denoted by $\ast$, such that, for all $a,b\in V$
\begin{equation}
	\log_\h(\exp_\h(a)\ast\exp_\h(b))=a\star b.\label{eq:GLvsStar}
\end{equation}
In particular, it is completely determined by the post-Lie Magnus expansion (and the map $\varUpsilon \co \ca{G} \to \Aut(\ca{H})$) and it defines a new group structure on the underlying set of $\text{Gp}(\hat{\mathcal U}(\h))$.}
\end{theorem}

\begin{proof}
This is essentially Equation \eqref{eq: side 1} and follows from the previous lemma as well as the definition of the $\star$-product.
\end{proof}


\subsection{Computations}
\label{ssec:comput}

The problem of computing the post-Lie Magnus expansion and its inverse is challenging. 
From \eqref{eq: theta and phi} we deduce that 
\begin{equation}
\label{eq:expequality}
	\exp_{\ast} \circ\ \phi \chi  = \exp_{\h},
\end{equation}
where, to alleviate notations, we wrote $\phi\chi := \phi\circ \chi$.
Despite being just the identity, we have made the presence of the invertible crossed morphism, $\phi=\id_V: \g \to \h$, defined relative to $\alpha \co \g \to \der(\h)$, visible. See Proposition \ref{prop: adjun} and recall that 
\begin{equation}
\label{phipostLie}
	\phi( [x,y]_\g)
	= \alpha_x(\phi(y)) - \alpha_y(\phi(x)) + [\phi(x),\phi(y)]_\h
	= x\pl y -y\pl x + [x,y]_\h.
\end{equation} 
Its peculiar role will become clear quickly. Following \cite[Eq.~(81)]{EFMpostLiealgebra-factorization}, we can add a formal parameter $t$ and consider the formal series
\begin{equation*}
	\phi\chi(xt) = \sum_{n\geq 1} \chi^{\phi}_n(x)t^n.  
\end{equation*}
By comparing  the coefficients of $t^n$ between $\exp_{\ast} \phi\chi(xt)$ and $\exp_{\h}(xt)$ one obtains for all  $x \in \h$ the highly non-linear recursion
\begin{equation*}
	\chi^{\phi}_n(x)= \frac{x^n}{n!} - \sum_{\stackrel{k\geq 2, p_i>0}{p_1+\cdots +p_k=n}} 
	\frac{1}{k!} \chi^{\phi}_{p_1}(x)\ast \cdots \ast \chi^{\phi}_{p_k}(x) . 
\end{equation*} 	

On the other hand, from \eqref{eq:postLieMag} and the fact that $\phi\chi(x) \in \h$, we deduce 
\begin{align}
\label{postLierecursion}
\begin{aligned}
	\Theta \exp_{\g}\chi(tx) 
	&= 1 + \sum_{n>0} \Theta(\chi(tx)^n) \frac{1}{n!}\\
	&= 1 + \phi\chi(tx) + \sum_{n>1} \frac{1}{n!} \sum_{\pi} ((\phi\chi(tx))^n)^\pi \\    
	&= \exp_{\h}(tx). 
\end{aligned}
\end{align}
In the second equality, the sum is taken over all set partitions $\pi$ of $\{ 1,\ldots,n\}$, see \eqref{orderedblocks} and  \eqref{orderedblocksexpansion}, as well as Formula \eqref{eq:setpartitions} in Proposition \ref{prop:cumulants}. From this equality, we deduce the recursion in $\hat{\mathcal U}(\h)$
\begin{equation}
\label{eq:postLieMagrecursion1}
	\phi\chi(tx) = xt +  \sum_{j>1} x^j \frac{t^j}{j!} - \sum_{n>1} \frac{1}{n!} \sum_{\pi} ((\phi\chi(tx))^n)^\pi .
\end{equation}
Since $\phi\chi$ can be seen as a deformation of the identity, we see that \eqref{eq:postLieMagrecursion1} permits to compute the terms in the sum $\sum_{n\geq 1} \chi^{\phi}_n(x)t^n$ recursively. We note that at each order $t^n$, the corresponding term in the first sum on the right-hand side of \eqref{eq:postLieMagrecursion1}, i.e., $x^n \frac{t^n}{n!} $ is cancelled. This can be seen from 
$$
	-\sum_{i>1}\frac{1}{i!} \sum_{\pi}  ((\phi\chi(tx))^i)^\pi 
	= - \sum_{j>1} (\phi\chi(tx))^j \frac{1}{j!} - \sum_{i>1}\frac{1}{i!} \sum_{\pi \atop |\pi| < i}  ((\phi\chi(tx))^i)^\pi .
$$
The order $t^n$ contribution from the right-hand side of the last equality is exactly $- x^n \frac{t^n}{n!} $. 
Note that the basic building blocks of this recursion are the n-fold iterated post-Lie products 
\begin{equation*}
	(\phi\chi(tx))_{n} :=  \phi\chi(tx) \pl (\cdots  (\phi\chi(tx) \pl (\phi\chi(tx) \pl \phi\chi(tx) ))\cdots ). 
\end{equation*}
Using this equation, one may compute the first few terms of $\phi\chi$ explicitly. 
For instance: 
\begin{equation}
\label{eq:postLieMagorder3}
	\phi\chi(x)= x 
				-\frac{1}{2}x\pl  x 
				+ \frac{1}{4}(x\pl x)\pl x 
				+\frac{1}{12}x\pl (x\pl x) 
				+ \frac{1}{12}[x\pl x,x]_\h 
				+ \cdots
\end{equation}

From $\phi=\id_V\co \g\to \h$ and \eqref{phipostLie}, we deduce that $\chi(x) \in \g$ is given by: 
\begin{multline} 
\label{eq:preLieMag in hbar}
\chi(x)= x 
		-\frac{1}{2}x\pl  x 
		+ \frac{1}{6}\left( 
		(x\pl x)\pl x 
		+  x\pl (x\pl x) 
		+ \frac{1}{2}[x\pl x,x]_\g 
		\right) 
		\\
	-\frac{1}{24}
	\bigg( ( (x\pl x) \pl x)\pl x
	+ ( x\pl (x \pl x))\pl x
	+  2(x\pl x)\pl (x\pl x)
	+  x\pl ((x\pl x)\pl x)
	\\
	+  x\pl (x\pl (x\pl x))
	- [x, x\pl(x\pl x)]_\g
	- [x,(x \pl x)\pl x]_\g 
	\bigg)
	+\cdots 
\end{multline}

Other methods, based on the combinatorial description of the free post-Lie algebra on one generator, are provided in \cite{MQ-Crossed}; see also \cite{Al-K-preLie} for the pre-Lie case. 
Using tree notation, the first four components of $\phi\chi(\bullet)$ are:   
\begin{align*}
\chi_1^{\phi}(\bullet) &= \bullet
	, ~~
\chi_2^{\phi}(\bullet) =  -\frac{1}{2}
\scalebox{0.4}{
	\begin{forest}
		baseline,for tree={%
			label/.option=content,
			content=,
			circle,
			fill,
			minimum size=9pt,
			inner sep=0pt,
			l =.5cm,
			s sep= 8mm,
			grow=north,
			edge ={line width=1pt}
		}
		[
		[] 	
		]
	\end{forest}
}
	,~~
\chi_3^{\phi}(\bullet) = 
\frac{1}{3}
\scalebox{0.4}{
	\begin{forest}
		baseline,for tree={%
			label/.option=content,
			content=,
			circle,
			fill,
			minimum size=9pt,
			inner sep=0pt,
			l =.5cm,
			s sep= 8mm,
			grow=north,
			edge ={line width=1pt}
		}
		[
		[[]] 	
		]
	\end{forest}
}
+
\frac{1}{12}
\scalebox{.4}{
\begin{forest}
	baseline,for tree={%
		label/.option=content,
		content=,
		circle,
		thick,
		fill,
		minimum size=9pt,
		inner sep=0pt,
		l =.5cm,
		s sep= 8mm,
		grow=north,
		edge ={line width=1pt}
	}
	[
	[] 	[]
	]
\end{forest}
}
+
\frac{1}{12}
\left(
	\scalebox{0.4}{
	\begin{forest}
		baseline,for tree={%
			label/.option=content,
			content=,
			circle,
			fill,
			minimum size=9pt,
			inner sep=0pt,
			l =.5cm,
			s sep= 8mm,
			grow=north,
			edge ={line width=1pt}
		}
		[
		[] 	
		]
	\end{forest}
	\begin{forest}
		baseline, for tree={%
			label/.option=content,
			content=,
			circle,
			fill,
			minimum size=9pt,
			inner sep=0pt,
			grow=north,
			l=0.1cm,
			edge ={line width=1pt}
		}
		[]
		\end{forest}}
	-
	\scalebox{0.4}{
		\begin{forest}
			baseline,for tree={%
				label/.option=content,
				content=,
				circle,
				fill,
				minimum size=9pt,
				inner sep=0pt,
				l =.5cm,
				s sep= 8mm,
				grow=north,
			}
			[]
		\end{forest}
		\begin{forest}
			baseline, for tree={%
				label/.option=content,
				content=,
				circle,
				fill,
				minimum size=9pt,
				inner sep=0pt,
				grow=north,
				l=0.1cm,
				edge ={line width=1pt}
			}
			[ [] ] 
	\end{forest}}
	\right) 
	\text{ and } 
	\\
	\chi_4^{\phi}(\bullet) &=
	-\frac{1}{4}
\scalebox{0.4}{
	\begin{forest}
		baseline,for tree={%
			label/.option=content,
			content=,
			circle,
			fill,
			minimum size=9pt,
			inner sep=0pt,
			l =.5cm,
			s sep= 8mm,
			grow=north,
			edge ={line width=1pt}
		}
		[
		[[[]]] 	
		]
	\end{forest}
}
	-
	\frac{1}{12}
\scalebox{.4}{
\begin{forest}
	baseline,for tree={%
		label/.option=content,
		content=,
		circle,
		thick,
		fill,
		minimum size=9pt,
		inner sep=0pt,
		l =.5cm,
		s sep= 8mm,
		grow=north,
		edge ={line width=1pt}
	}
	[
	[[] 	[]]
	]
\end{forest}
}
	-
	\frac{1}{12}
\scalebox{.4}{
\begin{forest}
	baseline,for tree={%
		label/.option=content,
		content=,
		circle,
		thick,
		fill,
		minimum size=9pt,
		inner sep=0pt,
		l =.5cm,
		s sep= 8mm,
		grow=north,
		edge ={line width=1pt}
	}
	[
	[] 	[[]]
	]
\end{forest}
}
	+
	\frac{1}{24}
	\left(
		\scalebox{0.4}{
		\begin{forest}
			baseline, for tree={%
				label/.option=content,
				content=,
				circle,
				thick,
				fill,
				minimum size=9pt,
				inner sep=0pt,
				grow=north,
				l=0.1cm,
			}
			[]
		\end{forest}
		\begin{forest}
			baseline,for tree={%
				label/.option=content,
				content=,
				circle,
				thick,
				fill,
				minimum size=9pt,
				inner sep=0pt,
				l =.5cm,
				s sep= 8mm,
				grow=north,
				edge ={line width=1pt}
			}
			[
			[] 	[]
			]
		\end{forest}
	}
	-
	\scalebox{0.4}{
		\begin{forest}
			baseline,for tree={%
				label/.option=content,
				content=,
				circle,
				fill,
				minimum size=9pt,
				inner sep=0pt,
				l =.5cm,
				s sep= 8mm,
				grow=north,
				edge ={line width=1pt}
			}
			[
			[] 	[]
			]
		\end{forest}
		\begin{forest}
			baseline, for tree={%
				label/.option=content,
				content=,
				circle,
				fill,
				minimum size=9pt,
				inner sep=0pt,
				grow=north,
				l=0.1cm,
			}
			[]
	\end{forest}}
	\right) 
	+
	\frac{1}{12}
	\left(
	\scalebox{0.4}{
	\begin{forest}
		baseline,for tree={%
			label/.option=content,
			content=,
			circle,
			fill,
			minimum size=9pt,
			inner sep=0pt,
			l =.5cm,
			s sep= 8mm,
			grow=north,
		}
		[]
	\end{forest}
	\begin{forest}
		baseline, for tree={%
			label/.option=content,
			content=,
			circle,
			fill,
			minimum size=9pt,
			inner sep=0pt,
			grow=north,
			l=0.1cm,
			edge ={line width=1pt}
		}
		[ [[]]]
\end{forest}}
	-
		\scalebox{0.4}{
		\begin{forest}
			baseline,for tree={%
				label/.option=content,
				content=,
				circle,
				fill,
				minimum size=9pt,
				inner sep=0pt,
				l =.5cm,
				s sep= 8mm,
				grow=north,
			}
			[ [[]] ] 
		\end{forest}
		\begin{forest}
			baseline, for tree={%
				label/.option=content,
				content=,
				circle,
				fill,
				minimum size=9pt,
				inner sep=0pt,
				grow=north,
				l=0.1cm,
			}
			[]
	\end{forest}}
	\right).
\end{align*}
As one may start to observe on the first terms of $\chi$, if the underlying Lie algebra of the post-Lie algebra is abelian, then one recovers the \emph{pre-Lie Magnus expansion}; for instance see \cite[Corollary 8]{EFMpostLiealgebra-factorization}.

We close this survey on the Magnus expansion by noting that the post-Lie Magnus expansion turns out to be useful in the analysis of iso-spectral type flows on some post-Lie algebras; see \cite{EFLM}. In \cite[Proposition 3.11]{EFLM}, it has been shown how to deduce from \eqref{eq:expequality} the differential equation of Magnus-type
\begin{equation}
\label{eq:postLieODE}
	\chi'(xt)=(d\exp_{\ast})^{-1}_{-\chi(xt)}\big(\exp_{\ast}(-\chi (xt))\triangleright x\big). 
\end{equation}
 Also, (see \cite{EFMpostLiealgebra-factorization}) for a post-Lie algebra $\h$, the equation $a'(t)  = - a(t) \pl a(t)$, where $a(t)\in \h[[t]]$, with initial condition $a(0)=a_0\in\h$ has a solution $a(t) = \exp_{\ast}(-\chi(a_0t)) \pl a_0.$ It seems to be an interesting problem to compare the two recursions \eqref{eq:postLieODE} respectively \eqref{eq:postLieMagrecursion1}.


\section{APPENDIX: Differential geometric and algebraic background}
\label{sec:diffformV} 

In this Appendix, we collect the necessary differential geometric concepts that have been employed in this survey. In particular, we recall the notion of a differential form with values in a vector space, the notion of a differential graded Lie algebra as well as that of an equivariant differential form. Some of the relations among these notions will be discussed and the relevant properties of these geometric concepts will be presented below.

	
\subsection{Differential forms with values in a vector space}


If $V$ is a finite dimensional vector space and $M$ is a smooth manifold, then for every $k\geq 0$ one can define $\Omega^k(M,V) = \Omega^k(M)\otimes V$, the vector space of $k$-forms on $M$ with values in $V$, whose elements are  finite sums $\sum_i \omega_i \otimes v^i$ where $\omega_i\in\Omega^k(M)$ and $v^i\in V$ for all $i$. Note that $\Omega^k(M,V)$ inherits from $\Omega^k(M)$ the structure of $C^\infty(M)$-module, i.e.~$f.(\omega\otimes v)=(f\omega)\otimes v$, for all $f \in C^\infty(M)$ and $\omega\otimes v\in\Omega^k(M)\otimes V$.

\begin{definition}
The graded $C^\infty(M)$-module of smooth forms on $M$ with values in $V$ is denoted by $\Omega^\bullet(M,V)=\oplus_{k\geq 0}\Omega^k(M,V)$, where $\Omega^k(M,V)$ holds the elements of degree $k$.
\end{definition}

\begin{remark}
If $X\in\mathfrak X(M)$, then one can extend the contraction operation $i_X$ as well as Lie derivative $\mathcal L_X$ in the obvious way from $\Omega^\bullet(M)$ to $\Omega^\bullet(M,V)$ such that $i_X$ is a $C^\infty(M)$-linear operator of degree $-1$ while $\mathcal L_X$ is $\mathbb R$-linear of degree $0$.
\end{remark}

One can \emph{twist} the de Rham complex $(\Omega^\bullet(M),d)$ of $M$ to obtain the complex of the differential forms on $M$ with values in $V$, i.e.~$(\Omega^\bullet(M,V),d_V)$, where $d_V:=d\otimes \mathrm{id}_V:\Omega^\bullet(M,V)\rightarrow\Omega^{\bullet+1}(M,V)$ is the \emph{twisted} Cartan differential.

By construction, the usual properties and operations on $\Omega^\bullet(M)$ extend to $\Omega^\bullet(M,V)$. Moreover, the latter is also functorial in $V$. Hereafter, we briefly recall these facts:

\begin{itemize}
	\item 
\textit{Functoriality in $M$}. 
To every smooth map $f\co M\to N$ one can associate the linear map $f^\flat:\Omega^k(N,V)\rightarrow\Omega^k(M,V)$, defined for each $k$ by 
\[
	f^\flat(\sum_i\omega_i\otimes v^i):=\sum_i(f^\ast\omega_i)\otimes v^i. 
\] 
Note that since the pull-back commutes with the Cartan differential, $f^\flat:(\Omega^\bullet (N,V),d_V)\rightarrow (\Omega^\bullet (M,V),d_V)$ is a morphism of complexes.

\item \textit{Functoriality in $V$.}   
If $V$ and $U$ are finite dimensional vector spaces, every $\phi\in\text{Hom}(V,U)$ defines $\phi^\sharp:\Omega^k(M,V)\rightarrow\Omega^k(M,U)$ via the following formula 
\begin{equation}
	\phi^\sharp\big(\sum_i\omega_i\otimes v^i\big):=\sum_i\omega_i\otimes\phi(v^i),\label{eq:func}
\end{equation}
which is a morphism of $C^\infty(M)$-modules and which extends uniquely to a morphism of graded $C^\infty(M)$-modules. Moreover, the previous formula implies
\[
	\phi^\sharp(d_V\omega)=d_U(\phi^\sharp\omega),
\]
for all $\omega\in\Omega^\bullet(M,V)$, i.e.~$\phi^\sharp:(\Omega^\bullet(M,V),d_V)\rightarrow(\Omega^\bullet(M,U),d_U)$ is a morphism of complexes. 

\item \textit{More generally:} observe that any bilinear map $B:V\times U\rightarrow W$ induces a graded $C^\infty(M)$-bilinear map $B^\sharp:\Omega^\bullet(M,V)\times\Omega^\bullet(M,U)\rightarrow\Omega^\bullet(M,W)$ via the following formula:
	\begin{equation}
		B^\sharp\big(\sum_{i}\omega_i\otimes v^i,\sum_j\omega'_j\otimes u^j\big)
		=\sum_{i,j}\omega_i\wedge\omega_j'\otimes B(v^i,u^j).\label{eq:dgla1}
	\end{equation}
Also note that for all $\omega\in\Omega^\bullet(M,V)$, $\omega'\in\Omega^\bullet(M,U)$
\begin{equation}
	d_WB^\sharp(\omega,\omega')=B^\sharp(d_V\omega,\omega')+(-1)^{\vert\omega\vert}B^\sharp(\omega,d_U\omega'),
	\label{eq:dgla2}
\end{equation}
and that $B^\sharp$ commutes with both the operation of Lie derivative and of contraction with respect to any vector field $X$ on $M$. 
\end{itemize}
\begin{remark}
	The Cartan formula $\mathcal L_X=d_V\circ i_X+i_X \circ d_V$ and more generally, the so-called Cartan calculus extends to this enhanced setting.
\end{remark}

Before continuing, let us recall the following 
	
\begin{definition}[Differential graded Lie algebra]
\label{def:dgla}
A triple $(\mathfrak G,[-,-],\delta)$, where $(\mathfrak G=\oplus_{i\geq 0}\mathfrak G_i,[-,-])$ is a graded Lie algebra and $\delta:\mathfrak G\rightarrow\mathfrak G$ is a differential (i.e.~a degree-one linear map squaring to zero) which is a graded derivation of $[-,-]$, is called a differential graded Lie algebra (dgla).
\end{definition}	

	We will specialize the above observations when $U=W=V=\mathfrak h$ is a finite dimensional Lie algebra. Under these assumptions, taking $B=[-,-]$, the Lie bracket on $\mathfrak h$, Equations \eqref{eq:dgla1} and \eqref{eq:dgla2} imply the existence of the structure of dgla on $\Omega^\bullet(M,\mathfrak h)$, whose Lie bracket and differential will be denoted by $[-\wedge -]$, respectively $d_\mathfrak h$,  (see Definition \ref{def:dgla}). Indeed, if $[-\wedge -]:\Omega^\bullet(M,\mathfrak h)\times\Omega^\bullet(M,\mathfrak h)\rightarrow\Omega^\bullet(M,\mathfrak h)$ is given by
\begin{equation}
	[\omega\wedge \omega']=\sum_{i,j}\omega_i\wedge\omega_j'\otimes [v^i,v^j],
	\quad
	\omega_i,\omega'_j\in\Omega^\bullet(M),
\end{equation}
then one can show that 
\begin{lemma}\label{lem: com times Lie}
	$(\Omega^\bullet(M,\mathfrak h),[-\wedge -],d_{\mathfrak h})$ is a dgla.
\end{lemma}

\begin{proof} It suffices to check that $[\omega_1\wedge\omega_2]=(-1)^{\vert\omega_1\vert\,\vert\omega_2\vert+1}[\omega_2\wedge\omega_1]$ and that

\[
	(-1)^{\vert\omega_1\vert\,\vert\omega_3\vert}[\omega_1\wedge[\omega_2\wedge\omega_3]]
	+(-1)^{\vert\omega_3\vert\,\vert\omega_2\vert}[\omega_3\wedge[\omega_1\wedge\omega_2]]
	+(-1)^{\vert\omega_2\vert\,\vert\omega_1\vert}[\omega_2\wedge[\omega_3\wedge\omega_1]]=0,
\]
both of which can be done by direct computations.
\end{proof}
\begin{remark}
	 Lemma \ref{lem: com times Lie} is a particular instance of a more general result: the tensor product of a Lie algebra with a commutative dg-algebra is a dgla. 
\end{remark}
	
An important notion is enclosed in the sequel

\begin{definition}[Maurer--Cartan element]
Let $(\mathfrak G,[-,-],\delta)$ be a dgla. An $\alpha\in\mathfrak G_1$ such that $\delta\alpha + \frac{1}{2}[\alpha,\alpha]=0$ is called a \emph{Maurer--Cartan} (MC) element of $(\mathfrak G,[-,-],\delta)$. The set of all MC-elements will be denoted by $\mathrm{MC}(\mathfrak G,[-,-],\delta)$ or simply $\mathrm{MC}(\mathfrak G)$ when this notation  does not cause any confusion.
\end{definition}
	
The prototypical example of MC-element is the so-called (left-invariant) Maurer--Cartan form, whose definition we recall in the sequel 	
	
\begin{example}[Maurer--Cartan form of $G$]
\label{ex:MCform}
		Let $G$ be a Lie group of dimension $n$ and, for every $g \in G$, let $L_g: G \rightarrow G$ denote the application of left-translation defined by $g$. Recall that $\omega\in\Omega^1(G)$ is called left-invariant if $L_g^\ast\omega=\omega$ for all $g\in G$, i.e.~if $\omega$ is completely specified by its value at the group unit $e_G$. 
		If $\omega_1, \ldots, \omega_n$ is a basis of $G$-invariant one forms such that $\omega_{i,e_G}=e_{i}$ (i.e.~$e^1,\ldots,e^n$ is a basis of $\mathfrak g$, the Lie algebra of $G$ and $e_1,\dots,e_n$ is its dual basis), then $\theta:=\sum_{i=1}^n\omega_i\otimes e^i\in\Omega^1(G,\mathfrak g)$ satisfies
\begin{equation}
\label{eq:MC}
	d_\mathfrak g\theta+\frac{1}{2}[\theta\wedge\theta]=0.
\end{equation}
In fact, for all $g\in G$ and all $u,v \in T_gG$, if $\xi, \eta \in \mathfrak g$ are the unique elements whose corresponding fundamental vector fields $X_\xi,X_\eta$ satisfy $X_\xi(g)=v$ and $X_\eta(g)=u$, then
\begin{eqnarray*}
	(d_\mathfrak g\theta)_g(v,u)
	&=&d_\mathfrak g\theta(X_\xi,X_\eta)(g)\\
	&=&\cancel{X_\xi\theta(X_\eta)(g)}-\bcancel{X_\eta\theta(X_\xi)(g)}-\theta([X_\xi,X_\eta])(g)
	=-\theta([X_\xi,X_\eta])(g),
\end{eqnarray*}
where the first two terms cancel (separately) due to $\theta(X_\xi)$ and $\theta(X_\eta)$ being constant $\mathfrak g$-valued functions on $G$. The surviving term can be computed as follows
\begin{eqnarray*}
	\theta([X_\xi,X_\eta])(g)
		&=&\theta(X_{[\xi,\eta]})(g)=\sum_{i=1}^n\langle{\omega_i},X_{[\xi,\eta]}\rangle(g) e^i
		=\sum_{i,j,k,l}\xi_j\eta_kC^{jk}_l\langle\omega_i,X_l\rangle(g)e^i\\
		&=& \sum_{i,j,k,l}\xi_j\eta_kC^{jk}_l\delta^l_ie^i\\
		&=&\sum_{i,j,k}\xi_j\eta_kC^{jk}_ie^i.
\end{eqnarray*}
On the other hand
\begin{eqnarray*}
	[\theta\wedge\theta]_g(v,u)
		&=&[\theta\wedge\theta](X_\xi,X_\eta)(g)
		=\sum_{r,s}(\omega_r\wedge\omega_s)(X_\xi,X_\eta)(g)[e^r,e^s]\\
		&=&\sum_{r,s,l}C^{rs}_l(\omega_r\wedge\omega_s)(X_\xi,X_\eta)(g)e^l\\
		&=&\sum_{r,s,l}C^{r,s}_l\sum_{j,k}\det 
		\begin{bmatrix}
		\delta^j_r 	& \delta^j_s\\ 
		\delta^k_r & \delta^k_s 
		\end{bmatrix}\xi_j\eta_ke^l
		=2\sum_{l,j,k}\xi_j\eta_kC^{jk}_le^l,
\end{eqnarray*}
		which together with the previous result yields \eqref{eq:MC}. 
\end{example}
	
The form $\theta$ defined above satisfies for all $g\in G$ and $v\in T_{g}G$
\begin{equation}
	\theta_g(v)=\langle\theta_g,v\rangle=(L_{g^{-1}})_{\ast,g}v.\label{eq:MC1}
\end{equation}
	
In fact for every $g\in G$ and $v\in T_gG$

\begin{eqnarray*}
	\theta_g(v)
	=\sum_{i=1}^n\langle\omega_{i,g},v\rangle e^i
	=\sum_{i=1}^n\langle (L_{g^{-1}})^\ast_{,e_G}e_i,v\rangle e^i
	=\sum_{i=1}^n\langle e_i,(L_{g^{-1}})_{\ast,g}v\rangle e^i
	=(L_{g^{-1}})_{\ast,g}v.
\end{eqnarray*}

The identity \eqref{eq:MC1} implies for all $\xi\in\g$
\begin{equation}
		\theta_{e_G}(\xi)=\xi.\label{eq:MC2}
\end{equation}
Moreover, one can prove that $\theta$ is left-invariant, i.e.
\begin{equation}
	L_g^\ast\theta=\theta,
	\quad \forall g \in G. \label{eq:MC3}
\end{equation}	
This property can be deduced from the ``coordinate representation" $\theta=\sum_{i=1}^n\omega_i\otimes e^i$ (being the $\omega_i$'s left-invariant) or proven using \eqref{eq:MC1}. In fact for every $h\in G$ and $v\in T_hG$

\begin{eqnarray*}
	\langle (L_g^\ast\theta)_h,v\rangle
	=\langle\theta_{gh},(L_g)_{\ast,h}v\rangle
	\stackrel{\eqref{eq:MC1}}{=}(L_{(gh)^{-1}})_{\ast,gh}(L_g)_{\ast,h}v
	=(L_{h^{-1}})_{\ast,h}v=\theta_h(v),
\end{eqnarray*}

which proves \eqref{eq:MC3}.

\begin{definition}[Left-invariant Maurer--Cartan form]
The form $\theta$ introduced above is called the left-invariant Maurer--Cartan form of $G$.
\end{definition}
	
In the next subsection we will introduce a general class of vector-valued differential forms which include the left-invariant Maurer--Cartan one forms.


\subsection{Equivariant differential forms} 

Let $G$ be a Lie group, let $(V,\rho)$ be a $G$-module (i.e.~a representation of $G$) and consider the complex $(\Omega^\bullet(G,V),d_V)$. Recall Formula \eqref{eq:func} for the definition of $\sharp$-morphism. 

\begin{definition}\label{def:Gequiv}
		      An element $\omega\in\Omega^k(G,V)$ is called $G$-equivariant if 
		\begin{equation}
			L_{g}^\ast\omega= \rho(g)^\sharp(\omega),\qquad \forall g\in G.\label{eq:equivK}
		\end{equation}

\end{definition}

\noindent Note that \eqref{eq:equivK} is equivalent to

\begin{equation}
	\omega_{g'g}\big((L_{g'})_{\ast,g}v_1,\dots,(L_{g'})_{\ast,g}v_k\big)=\rho(g')(\omega_{g}(v_1,\dots,v_k)),\label{eq:equivK1}
\end{equation}

for all $g,g'\in G$ and $v_1,\dots,v_k\in T_gG$. Formulas \eqref{eq:func} and \eqref{eq:equivK} imply that the differential of a $G$-equivariant form is still $G$-equivariant, i.e.~$d_V$ descends to a differential on $\Omega^\bullet(G,V)^G$, which will be still denoted by $d_V$, defining the complex $(\Omega^\bullet(G,V)^G,d_V)$ of the $G$-equivariant forms on $M$ with values in $V$. 

\begin{example}
The left-invariant Maurer--Cartan form of $G$ is an example of $G$-equivariant one-form. In fact, Formula \eqref{eq:MC1} implies
\[
	\theta_{g'g}\big((L_{g'})_{\ast,g}v\big)
	=(L_{{g'g}^{-1}})_{\ast,g'g}\big((L_{g'})_{\ast,g}v\big)
	=(L_{g^{-1}})_{\ast,g}(v)=\theta_g(v),
\]

which proves the statement if one chooses $(V,\rho)=(\mathfrak g,id)$, where $id:G\rightarrow \text{GL}(\mathfrak g)$ is the \emph{trivial}, i.e.~$id(g)=\text{Id}\in\text{GL}(\mathfrak g)$ for all $g\in G$, Lie group morphism from $G$ to $\text{GL}(\mathfrak g)$. Note that together with the left one can consider $\theta^R$, the right-invariant Maurer--Cartan form of $G$, i.e.~the unique one form with values in $\mathfrak g$ such that $R^\ast_g\theta^R=\theta^R$ for all $g\in G$. In this case one can prove that $L^\ast_g\theta^R=\text{Ad}^\sharp(g)\circ\theta^R$, i.e.~$\theta^R$ is $G$-equivariant with respect to the $G$-module $(\mathfrak g,\text{Ad})$, where $\text{Ad}:G\rightarrow\text{GL}(\mathfrak g)$ is the adjoint representation of $G$.
\end{example}

	Now recall that if $\mathfrak g$ is a Lie algebra and $(V,r)$ is a $\mathfrak g$-module one can define the Chevalley--Eilenberg complex $(C^\bullet (\mathfrak g,V),d_{CE})$, where $C^\bullet=\oplus_{k\geq 0}C^k(\mathfrak g,V)$, $C^k(\mathfrak g,V)=\text{Hom}(\Lambda^k\mathfrak g,V)$, whose differential $d_{CE}:C^{\bullet}(\mathfrak g,V)\rightarrow C^{\bullet+1}(\mathfrak g,V)$ is given by the following formula:
\begin{eqnarray*}
	(d_{CE}\sigma)(x_0,\dots,x_{n})&=&\sum_{i=0}^n(-1)^{i}r(x_i)\sigma(x_0,\dots,\hat{x}_{i},\dots,x_n)\\
	&+&\sum_{i<j}(-1)^{i+j}\sigma([x_i,x_j],x_0,\dots,\hat{x}_i,\dots,\hat{x}_j,\dots,x_{n}).
\end{eqnarray*}
	
Note that Formula \eqref{eq:equivK1} implies that every $G$-invariant $k$-form is completely determined by its value at $e_G$, This implies the existence of a grading-preserving linear isomorphism $ev_{e_G}:\Omega^\bullet(G,V)^G\rightarrow C^\bullet(\mathfrak g,V)$, defined for all $\omega\in\Omega^k(G,V)^G$ by
\begin{equation}
	ev_{e_G}(\omega)=\omega_{e_G}, \label{eq:evmap}
\end{equation}
where $V$ is thought of as a $G$-module via $\rho:G\rightarrow\text{GL}(V)$ and, respectively, as a $\mathfrak g$-module via $r:=\rho_{\ast,e_G}:\mathfrak g\rightarrow \text{End}(V)$. Here we have used implicitly the fact that $\text{Hom}(\Lambda^k\mathfrak g,V)\cong \Lambda^k\mathfrak g^*\ot V$. After these preliminary remarks, we can state the following important result

\begin{theorem}\label{th: DeRham and CE}
The application defined in \eqref{eq:evmap} is an isomorphism of complexes from $(\Omega^\bullet(G,V)^G,d_V)$ to $(C^\bullet (\mathfrak g,V),d_{CE})$.
\end{theorem}


\subsubsection{Two dgla's} 
\label{sssec:dgls}
	
	We will consider the following case. Let $\mathfrak h$ be a finite dimensional Lie algebra and let $\text{Der}(\mathfrak h)$ be the Lie algebra of its derivations. Suppose that $r:\mathfrak g \rightarrow\text{Der}(\mathfrak h)$ is a Lie algebra morphism. Note that these hypotheses imply that $(\mathfrak h,r)$ is a $\mathfrak g$-module. Observe that $\text{Der}(\mathfrak h)$ is a Lie algebra whose corresponding (connected and simply connected) Lie group will be denoted with $\text{Aut}(\mathfrak h)$. This is the set of linear automorphisms of $\mathfrak h$ whose group operation is the composition between elements of $\text{End}(\mathfrak h)$. The Lie group structure on $\text{Aut}(\mathfrak h)$ is obtained noticing that it is a closed subgroup of $\text{GL}(\mathfrak h)$. Note further that the morphism of Lie algebras $r$ can be \emph{integrated} to a morphism of Lie groups $\rho:G\rightarrow\text{Aut}(\mathfrak h)$, where $G$ is the (unique up to isomorphism) connected and simply connected Lie group whose Lie algebra is $\mathfrak g$.  For all $s_1\in C^j(\mathfrak g,\mathfrak h)$ and $s_2\in C^k(\mathfrak g,\mathfrak h)$, let $\{-,-\}:C^j(\mathfrak g,\mathfrak h)\times C^k(\mathfrak g,\mathfrak h)\rightarrow C^{j+k}(\mathfrak g,\mathfrak h)$ the bilinear application defined by
\begin{equation}
\label{eq:dglah} 
	\{s_1,s_2\}(x_1,\dots,x_{j+k})
	=\sum_{\sigma\in\Sigma_{j,k}}(-1)^{\vert\sigma\vert}[s_1(x_{\sigma(1)},\dots,x_{\sigma(j)}),s_2(x_{\sigma(j+1)},\dots,x_{\sigma(j+k)})],
\end{equation}
	where $\Sigma_{j,k}\subset\Sigma_{j+k}$ is the subgroup of the symmetric group $\Sigma_{j+k}$ of the $(j,k)$-shuffles, $\vert\sigma\vert$ is the signature of the permutation $\sigma$ and the bracket $[-,-]$ on the right-hand side of the previous formula is the Lie bracket defined on $\mathfrak h$. Now one can prove 
	
\begin{proposition}
\label{prop: Lie on CE}
The bilinear map defined in Formula \eqref{eq:dglah} extends to $C^\bullet(\mathfrak g,\mathfrak h)$ as a graded Lie bracket, still denoted by $\{-,-\}$, such that $(C^\bullet(\mathfrak g,\mathfrak h),\{-,-\},d_{CE})$ is a dgla.
\end{proposition}
	
Let $G$ and $\rho:G\rightarrow\text{Aut}(\mathfrak h)$ be as above and consider $(\Omega^\bullet(G,\mathfrak h)^G,[-\wedge -],d_\mathfrak h)$ (note that $(\mathfrak h,\rho)$ is a $G$-module). 

\begin{proposition} 
\label{prop: de Rham CE iso of Lie}
The evaluation map $ev_{e_G}:\Omega^\bullet(G,\mathfrak h)^G\rightarrow C^\bullet(\mathfrak g,\mathfrak h)$, see \eqref{eq:evmap}, defines an isomorphism between the corresponding dgla's. 
\end{proposition}
	
An obvious consequence of the previous results is that $\text{MC}(\Omega^\bullet(G,\mathfrak h)^G)$ and $\text{MC}(C^\bullet(\mathfrak g,\mathfrak h))$ are in one-to-one correspondence.
	
Before closing this section it is worth making the following observation. Let $H$ be a Lie group whose Lie algebra is $\mathfrak h$. Consider
\begin{eqnarray*}
	\text{Aut}(H)
	&=&\{\phi:H\rightarrow H\vert\;\phi\ \text{is a group morphism, which is}\\
	&  &\text{ a diffeomorphism of the underlying manifold}\}
\end{eqnarray*}
Note now that, by differentiation, every $\phi\in\text{Aut}(H)$ induces a unique element $\phi_{\ast,e_H}\in\text{Aut}(\mathfrak h)$ and that the application so defined is a morphism of abstract groups. Now one can show the
	
\begin{proposition}
If $H$ is connected and simply connected, then the group morphism which takes $\phi\in\text{Aut}(H)$ to $\phi_{\ast,e_H}\in\text{Aut}(\mathfrak h)$ is an isomorphism of abstract groups.
\end{proposition}
	
As a consequence of the previous result one has that the abstract group $\text{Aut}(H)$ inherits a structure of a Lie group so that the application which takes $\phi\in\text{Aut}(H)$ to $\phi_{\ast,e_H}\in\text{Aut}(\mathfrak h)$ is an isomorphism of Lie groups.

	
\bibliographystyle{alpha}

\end{document}